\documentclass[a4paper,fleqn]{cas-dc}

\usepackage[numbers]{natbib}
\renewcommand{\printorcid}{}
\usepackage{amsthm}
\usepackage{amsmath,amssymb,amsfonts}
\usepackage{algorithmic}
\usepackage[ruled,linesnumbered]{algorithm2e}
\usepackage{textcomp}
\usepackage{url}            
\usepackage{booktabs}       

\usepackage{nicefrac}
\usepackage{microtype} 
\usepackage{lipsum}
\usepackage{enumitem}
\usepackage{listings}
\usepackage{graphicx}
\usepackage{epstopdf}
\usepackage{array}
\usepackage{wrapfig}
\usepackage{mathtools,cuted}
\usepackage{color}
\usepackage{colortbl}
\usepackage{cases}
\DeclarePairedDelimiter\abs{\lvert}{\rvert}

\usepackage{stackengine,graphicx}
\stackMath
\usepackage[normalem]{ulem}
\usepackage{bm}
\usepackage{mathtools}
\usepackage{verbatim}
\usepackage{caption}
\usepackage{subcaption}

\newcommand{\xbf}{\mathbf{x}}

\newcommand{\bbf}{\mathbf{b}}
\newcommand{\ubf}{\mathbf{u}}
\newcommand{\Ubf}{\mathbf{U}}
\newcommand{\vbf}{\mathbf{v}}
\newcommand{\wbf}{\mathbf{w}}
\newcommand{\zbf}{\mathbf{z}}
\newcommand{\fbf}{\mathbf{f}}
\newcommand{\sbf}{\mathbf{s}}
\newcommand{\Fbf}{\mathbf{F}}
\newcommand{\Pbf}{\mathbf{P}}

\newcommand{\gbf}{\mathbf{g}}

\DeclareMathOperator{\prox}{\mathrm{prox}}

\newcommand{\C}{\mathbb{C}}

\DeclareMathOperator*{\argmin}{arg\,min}

\newcommand{\J}{\mathcal{J}}
\newcommand{\K}{\mathcal{K}}

\newcommand{\N}{\mathcal{N}}
\newcommand{\soft}{\mathcal{S}}

\newcommand{\M}{\mathcal{M}}

\newcommand{\G}{\mathcal{G}}

\newtheorem{theorem}{Theorem}

\begin{document}

\let\WriteBookmarks\relax
\def\floatpagepagefraction{1}
\def\textpagefraction{.001}
\shorttitle{}
\shortauthors{W. Bian et~al.}

\title [mode = title]{An Optimal Control Framework for Joint-channel Parallel MRI  Reconstruction without Coil Sensitivities}                      

\author[1]{\color{black}Wanyu Bian}
\cormark[1]
\ead{wanyu.bian@ufl.edu}

\author[1]{\color{black}Yunmei Chen}

\address[1]{Department of Mathematics, University of Florida, Gainesville, FL 32601, USA}

\author[2]{\color{black}Xiaojing Ye}
\address[2]{Department of Mathematics and Statistics, Georgia State University, Atlanta, GA 30303, USA}

\cortext[cor1]{Corresponding author.}

\begin{abstract}
\emph{Goal}: This work aims at developing a novel calibration-free fast parallel MRI (pMRI) reconstruction method incorporate with discrete-time optimal control framework. The reconstruction model is designed to learn a regularization that combines channels and extracts features by leveraging the information sharing among channels of multi-coil images. We propose to recover both magnitude and phase information by taking advantage of structured convolutional networks in image and Fourier spaces.
\emph{Methods}: We develop a novel variational model with a learnable objective function that integrates an adaptive multi-coil image combination operator and effective image regularization in the image and Fourier spaces. We cast the reconstruction network as a structured discrete-time optimal control system, resulting in an optimal control formulation of parameter training where the parameters of the objective function play the role of control variables. We demonstrate that the Lagrangian method for solving the control problem is equivalent to back-propagation, ensuring the local convergence of the training algorithm.
\emph{Results}: We conduct a large number of numerical experiments of the proposed method with comparisons to several state-of-the-art pMRI reconstruction networks on real pMRI datasets. The numerical results demonstrate the promising performance of the proposed method evidently. 
\emph{Conclusion}: The proposed method provides a general deep network design and training framework for efficient joint-channel pMRI reconstruction.
\emph{Significance}: By learning multi-coil image combination operator and performing regularizations in both image domain and k-space domain, the proposed method achieves a highly efficient image reconstruction network for pMRI. 
\end{abstract}


\ExplSyntaxOn
\keys_set:nn { stm / mktitle } { nologo }
\ExplSyntaxOff
\maketitle

\section{Introduction}
\label{sec:introduction}

Magnetic resonance imaging (MRI) is one of the most prominent medical imaging technologies with extensive clinical applications. In clinical applications, an advanced medical MRI technique known as parallel MRI (pMRI) is widely used. PMRI surrounds the scanned objects with multiple receiver coils and collects k-space (Fourier) data in parallel. PMRI can reduce the data acquisition time and has become the state-of-the-art technology in modern MRI applications. To accelerate the scan process, partial data acquisitions that increase the spacing between read-out lines in k-space are implemented in pMRI. However, this results in aliasing artifacts, and a proper image reconstruction process is necessary to recover the high-quality artifact-free images from the partial data.

Two major approaches are commonly addressed to image reconstruction in pMRI: the first approach is k-space method which interpolates the missing k-space data using the sampled ones across multiple receiver coils \cite{doi:10.1002/jmri.23639}, such as the generalized auto-calibrating partially parallel acquisition (GRAPPA) \cite{griswold2002generalized} and  simultaneous acquisition of spatial harmonics (SMASH) \cite{sodickson1997simultaneous}. The other approach is the class of image space method that eliminate the aliasing artifacts in the image domain by solving a system of equations that relate the image to be reconstructed and partial k-space data through coil sensitivities, such as SENSitivity Encoding (SENSE) \cite{pruessmann1999sense}.
Coil sensitivity maps are indispensable and required to be accurately pre-estimated in traditional SENSE-based methods.
Traditional pMRI reconstruction methods in image space follows SENSE-based framework, which is formulated as an optimization problem that minimizes a summation of a data fidelity term and a weighted regularization term. The detailed explanation about this formulation can be refer to \cite{knoll2020deep}.

In recent years, we have witnessed fast developments of SENSE-based pMRI  reconstruction incorporate with deep-learning based methods \cite{knoll2020deep, lu2020pfista, tavaf2021grappa,hammernik2021systematic, lv2021pic}. There are two critical issues that need to be carefully addressed.
The first issue is on the choice of regularization including the regularization weight.
The regularization term is of paramount importance to the severely ill-posed inverse problem of pMRI reconstruction due to the significant undersampling of k-space data. 
In the past decades, most traditional image reconstruction methods employ handcrafted regularization terms, such as the total variation (TV). 
In recent years, a class of unrolling methods that mimic classical optimization schemes are developed, where the regularization is realized by deep networks whose parameters are learned from data. However, the learned regularization is often cast as a black-box that is difficult to interpret, and the training can be very data demanding and prone to overfitting especially when the networks are over-parameterized \cite{lecun2015deep,279181,VanishingGradient}. 

The second issue is due to the unavailability of accurate coil sensitivities $\{ \sbf_i\}$ in practice. Inaccurate coil sensitivity maps lead to severe biases that degrade the quality of reconstructed $\vbf$. 
One way to eliminate this issue is to (regularize and) reconstruct  multi-coil images, and combine channels into a full-body image in the final step by taking some hand-crafted methods such as the root of sum of squares (RSS). Different from RSS, our method proposed a learnable multi-coil combination operator to combine channel-wise multi-coil images.

                                                          
In this paper, we tackle the aforementioned issues in a discrete-time optimal control framework to optimize the variational pMRI reconstruction model.
We highlight several main features of our framework as follows.

\begin{enumerate}[leftmargin=*]
    \item Unlike most existing methods which regularize and reconstruct multi-coil images, we employ regularization in both image and Fourier spaces to improve reconstruction quality. 
    \item Our method advocates a learned adaptive combination operator that first merges multi-coil images into a full-body image with a complete field of view (FOV), followed by an effective regularization on this image. 
    This is in sharp contrast to existing methods which only combine reconstructed multi-coil images in the final step, whereas our regularizer leverages the combination operator in each iteration which improves parameter efficiency.
    \item We employ a complex-valued neural network as the coil combination operator to recover both magnitude and phase information of pMRI images when coil sensitivity is unavailable. This combination method benefits from the coil information shared among multiple channels, which is distinct from most hand-crafted coil combination methods.  
   \item We propose a novel deep reconstruction network whose structure is determined by the discrete-time optimal control system for minimizing the objective function, which yields an optimal control formulation where the parameters of the combination and regularization operators play the role of control variables of the discrete dynamic system. The optimal value of these parameters is obtained by a Lagrangian method which can be implemented using back-propagation. 
   
\end{enumerate}

We consider two clinical pMRI sequences of knee images and verified the effective performance of the proposed combination operator, different initial reconstructions, complex convolutions, and domain-hybrid network in the Ablation Studies. The proposed network recovers both magnitude and phase information of pMRI images. The effect of the aforementioned techniques demonstrate evident improvement of reconstruction quality and parameter efficiency using our method. For reproducing the experiment, our code is available at \url{https://github.com/1lol/pMRI_optimal_control}.

This remainder of the paper is organized as follows: In section \ref{sec:related}, we provide an overview of recent developments in pMRI, cross-domain reconstructions, complex-valued CNNs, and optimal control inspired deep training models that related to our work. We present our proposed problem settings and reconstruction network architecture in detail in Section \ref{sec:proposed}. Extensive numerical experiments and analyses on a variety of clinical pMRI data are presented in Section \ref{sec:experiment}. Section \ref{sec:conclusion} concludes this paper.

\section{Related work}
\label{sec:related}

In recent years, we have witnessed fast developments of medical imaging incorporate with deep-learning based methods \cite{huang2020medical, huang2020magnetic, huang2020mri}. Most existing deep-learning based methods rendering end-to-end neural networks mapping from the partial k-space data to the reconstructed images \cite{WANG2020136, zhou_pmri, Quan2018CompressedSM,8417964, zhu2018}. These approaches require an excessive amount of training data, and the designed networks are cast as black boxes whose underlying mechanism can be very difficult to interpret. To mitigate this issue, a number of unrolling methods were proposed to map existing optimization algorithms to structured networks where each phase of the networks correspond to one iteration of an optimization algorithm \cite{doi:10.1002/mrm.26977, adler2018learned, Aggarwal_2019,8550778, cheng2019model,NIPS2016_6406,zhang2018ista, 8067520}.
In what follows, we focus on the recent developments in deep-learning based image reconstruction methods for pMRI.

Variational Network (VN) \cite{doi:10.1002/mrm.26977} was introduced to unroll the gradient descent algorithm as a reconstruction network which requires precalculated sensitivities $\{\sbf_i\}$ as input.
%
%
MoDL \cite{Aggarwal_2019} proposed a weight sharing strategy in a recursive network to learn the regularization parameters by unrolling the conjugate gradient method.
%
%
%
%
%
Several methods explored different strategies to avoid using pre-calculated coil sensitivity maps for pMRI reconstruction.
Blind-PMRI-Net \cite{10.1007/978-3-030-32251-9_80} proposed pMRI model by regularizing sensitivity maps and MR image, where their network alternatively estimates coil images, sensitivities and single-body image by three subnets. 
De-Aliasing-Net \cite{10.1007/978-3-030-32248-9_4} proposed a de-aliasing reconstruction model with that applied split Bregman iteration algorithm without explicit coil sensitivity calculation. The de-aliasing network explored cross-correlation among channels and spatial redundancy which provoked a desirable performance.
LINDBERG \cite{wang2017learning} explored calibration-free pMRI technique which uses adaptive sparse coding to obtain joint-sparse representation precisely by equipping a joint sparsity regularization to extract desirable cross-channel relationship. This work proposed to alternatively update sparse representation, sensitivity encoded images, and K-space data. 
%
Adaptive-CS-Net \cite{pezzotti2020adaptive} is a leading method in 2019 fastMRI challenge \cite{zbontar2018fastmri} that unrolled modified ISTA-Net$^{+}$ \cite{zhang2018ista}. 
The proposed calibration-free pMRI method distincts from above related works in terms of the learnable multi-coil combination operator to adaptively combine channels of the updated multi-coil images through iterations.

Recently, cross-domain methods exhibits its significance in medical imaging \cite{  doi:10.1002/mrm.27201, wang2020ikwi, sriram2020grappanet, 10.1007/978-3-030-59713-9_41,   pmlr-v102-souza19a,  10.1007/978-3-030-59713-9_34,  10.1007/978-3-030-59713-9_37, souza2019hybrid}.
%
KIKI-net \cite{doi:10.1002/mrm.27201} iteratively applied k-space CNN, image domain CNN and interleaved data consistency operation for single-coil image reconstruction. 
%
%
CDF-Net \cite{ 10.1007/978-3-030-59713-9_41} further shows adding communication between spatial and frequency domain gives a boost in performance. Their results indicated that domain-specific network has individual strong points and disadvantages in restoring tissue-structure. Our reconstruction model is inspired of cross-domain reconstruction, the difference is that we solve for the reconstruction model with cross-domain regularization functions through a learnable optimization algorithm instead of an end-to-end network.


Our network applies complex-valued convolutions and activation functions. MRI data are complex-valued, and the phase signals also carry important pathological information such as in quantitative susceptibility mapping \cite{cole2020analysis, sandino2020compressed}. 
%
%
%
Cole $et$ $al.$ \cite{cole2020analysis} investigated the performance of complex-valued convolution and activation functions has better reconstruction over the model with real-valued convolution in various network architectures. 
%
DeepcomplexMRI \cite{WANG2020136} was developed to recover multi-coil images by unrolling an end-to-end complex-valued network.

In supervised learning, deep residual neural networks can be approximated as discretizations of a classical optimal control problem of a dynamical system, where training parameters can be viewed as control variables \cite{weinan2017proposal, lidynamical,li2019deep}. Control inspired learning algorithms introduced a new family of network training models which connect with dynamical systems. Pontryagin's maximum principle (PMP) \cite{PMP} was explored as necessary optimality conditions for the optimal control \cite{li2017maximum}, \cite{pmlr}, these works devise the discrete method of successive approximations (MSA) \cite{MSA} and its variance for solving PMP. 
Neural ODE \cite{chen2018neural} models the continuous dynamics of hidden states by some certain types of neural networks such as ResNet, the forward propagation is equivalent to one step of discretatized ordinary differential equations (ODE). Inspired by Neural ODE \cite{chen2018neural}, Chen $et$ $al.$ \cite{chen2020mri} modeled ODE-based deep network for MRI reconstruction. In this paper, we model the optimization trajectory as a discrete dynamic process from the view of the method of Lagrangian Multipliers.

The present work is a substantial extension of the preliminary work in \cite{10.1007/978-3-030-61598-7_2} using domain-hybrid network with a trained initialization to solve for an optimal control pMRI joint-channel reconstruction problem when coil-sensitivity is unavailable. More comprehensive empirical study is conducted in this work.

\section{Proposed Method}
\label{sec:proposed}

\subsection{Background} \label{background}

PMRI as well as general MRI reconstruction can be formulated as an inverse problem. Consider a pMRI system with $c$ receiver coils acquiring 2D MR images at resolution $m\times n$ (we treat a 2D image $\vbf \in \mathbb{C}^{m\times n}$ and its column vector form $\vbf \in \mathbb{C}^{mn}$ interchangeably hereafter). Let $\Pbf \in \mathbb{R}^{p\times mn}  (p \le mn)$ be the binary matrix representing the undersampling mask with $p$ sampled locations in k-space, $\sbf_i\in\mathbb{C}^{mn}$ the coil sensitivity, and $\fbf_i \in \mathbb{C}^{p}$ the \emph{partial} k-space data at the $i$-th receiver coil for $i=1,\dots,c$. The partial data $\fbf_i$ and the image $\vbf$ are related by $\fbf_i = \Pbf \Fbf (\sbf_i \odot \vbf) + \mathbf{n}_i$, where $\sbf_i$ is the (unknown) sensitivity map at the $i$-th coil and $\odot$ denotes entrywise product of two matrices, $\Fbf \in \mathbb{C}^{mn\times mn}$ stands for the (normalized) discrete Fourier transform that maps an image to its Fourier coefficients, and $\mathbf{n}_i$ is the unknown acquisition noise in k-space at the $i$-th receiver coil. Then the variational model for image reconstruction can be cast as an optimization problem as follows:
\begin{equation}\label{eq:PFS}
    \min_{\vbf} \ \sum\nolimits^{c}_{i=1} \frac{1}{2} \| \Pbf \Fbf (\sbf_i \odot \vbf)- \fbf_i\|^2 + R(\vbf),
\end{equation}
where $\vbf\in \mathbb{C}^{m n}$ is the single full-body MR image to be reconstructed, $R(\vbf)$ is a regularization on the image $\vbf$, and $\| \wbf \|^2 := \| \wbf \|_2^2 = \sum_{j=1}^n |w_j|^2$ for any complex vector $\wbf = (w_1,\dots,w_n)^{\top} \in \mathbb{C}^n$.
Our approach is based on uniform Cartesian k-space sampling. Table \ref{tab:notations} displays the notations and their descriptions that we used in the paper.

\begin{table}
    \centering
    \caption{Some notations and meanings that used throughout this paper.}
    \resizebox{\linewidth}{46mm}{ 
    \begin{tabular}{ll}
    \toprule
    Expression     &  Description \\
    \midrule
    $c$ & total number of receiver coils\\
    $\ubf= (\ubf_1,\dots, \ubf_{c}) $ & multi-coil MRI data\\
    $\fbf= (\fbf_1,\dots,\fbf_{c}) $ & partial k-space measurement\\
    $\sbf = (\sbf_1,\dots, \sbf_{c}) $ & coil sensitivity map\\
    $\vbf $ & full fov image that need to be reconstructed \\
    $\ubf^*$ & ground truth multi-coil MRI data\\
    $\vbf^*$ & ground truth single body MRI data\\
    $\Fbf $ & discrete Fourier transform\\
    $\Fbf^{H} $ & inverse discrete Fourier transform\\
    $\Pbf $ & undersampling trajectory\\
    $\mathbf{n}$ & measurement noise\\
    $\gbf$ & algorithm unrolling network\\
    $\gbf_0$ & initial network\\
    $\J,\G,\tilde{\G}, \tilde{\J}$ & image space convolutional operators in $\gbf$ \\
    $\K$ & k-space convolutional operators in $\gbf$ \\
    $\K_0$ & k-space convolutional operator in $\gbf_0$\\
    $h$ & data fidelity term\\
    $R$ & regularization\\
    \text{RSS} & square root of the sum of squares\\
    $ t=1,\cdots,T$ & phase number \\
    $ k=1,\cdots,K$ & number of iterations for Alg \eqref{alg:1}\\
    $\Ubf=(\ubf{(0)}, \cdots, \ubf{(T)})^{\intercal}$ & collection of predicted multi-coil images\\
    & at each phase\\
    $\Theta=(\theta(0),\cdots, \theta(T))^{\intercal}$ & collection of control variable (parameters)\\
    & at each phase\\    
    $\Lambda = (\lambda(0), \cdots, \lambda(T))^{\intercal}$ & collection of Lagrangian multipliers of \eqref{eq:bilevel}\\
    \bottomrule
    \end{tabular}}
    \label{tab:notations}
\end{table}

\subsection{Problem Settings} \label{subsec:setting}

%
We propose a unified deep neural network for calibration-free pMRI reconstruction by recovering images from individual receiver coils jointly that does not require any knowledge of coil-wise sensitivity profile. 
Denote that $\ubf_i$ is the MR image at the $i$-th receiver coil and hence is related to the full body image $\vbf$ by $\ubf_i = \sbf_i  \odot \vbf$. 
On the other hand, the image $\ubf_i$ corresponds to the partial k-space data $\fbf_i$ by $\fbf_i = \Pbf \Fbf \ubf_i + \mathbf{n}_i$, and hence the data fidelity term is formulated as least squares $ \frac{1}{2}\sum^{c}_{i=1} \|\Pbf \Fbf \ubf_i - \fbf_i\|^2$.
We also need a suitable regularization $R$ on the images $\{\ubf_i\}$.
However, these images have severely geometrically inhomogeneous contrasts due to the physical variations in the sensitivities across the image domain at different receiver coils. Therefore, it is more appropriate and effective to apply regularization to the single full-body image $\vbf$ than to individual coil images in $( \ubf_1,\dots,\ubf_c )$.


To address the issue of proper regularization, we propose to learn a nonlinear operator $\J$ that combines $\{\ubf_i\}$ into the image $\vbf = \J (\ubf) \in \mathbb{C}^{m\times n}$, where $\ubf = (\ubf_1,\dots, \ubf_{c}) \in \mathbb{C}^{m \times n \times c}$ represents the channel-wise multi-coil MRI data that consists of $ \ubf_i$ for $ i=1,\cdots,c$.
Then we apply a suitable $R$ to the image $\vbf$. 
%
We also introduce a k-space $R_f$ on $\Fbf \ubf_i$ to take advantage of Fourier information and enhance the model performance.


We denote $\fbf = (\fbf_1,\dots,\fbf_{c}) \in \mathbb{C}^{p \times c}$ as the partial k-space measurements at $c$ sensor coils. Suppose that we are given $N$ data pairs $\{(\fbf^{(j)}, \ubf^{*(j)}) \}_{j=1}^N$ for training the network, where $\ubf^{*(j)}$ is the ground truth multi-coil MR data with index $j \in \{ 1, \cdots, N\}$. Let $\Theta$ represents the parameters that need to be learned from network by minimizing the loss function $\ell$. We formulate the network training as a bilevel optimization problem, where the lower level is to update $ \ubf$ with fixed trainable parameters $\Theta$ and upper level is to update $\Theta$ that learned from the training data by minimizing loss function $\ell$.
\begin{subequations}\label{learnable_pmri}
\begin{align}
& \min_{\Theta} \ \ \frac{1}{N}\sum^N_{j=1}  \ell(  \ubf^{(j)}_{\Theta}) ,\ \ \ \label{learnable_pmri_upper} \\
& \mathrm{s.t.} \ \  \ubf_{\Theta}^{(j)} = \argmin_{\ubf^{(j)}}  \phi_{\Theta}( \ubf^{(j)}), \ \  \label{learnable_pmri_lower}  
\end{align}
\end{subequations}
with $\phi_{\Theta}$ defined as below:
\begin{equation}\label{eq:m}
\phi(\ubf):= \frac{1}{2} \sum\limits^{c}_{i=1} \| \textbf{PF} \ubf_i - \fbf_i \|^2  +  R(\J(\ubf)) + R_f(\Fbf \ubf_i).
\end{equation}
The objective function $\phi$ is the variational model for pMRI reconstruction, in which $\Theta$ is the set collects all the parameters that learned from the regularizers $R \circ \J$ and $R_f \circ \Fbf$, so $\phi$ is depending on $\Theta$.
Problem \ref{learnable_pmri} is formulated in the scenario of reconstructing the multi-coil MRI data.  The final reconstruction result $\ubf_{\Theta}$ is the forward network output that dependent on network parameters $\Theta$.   

The deep learning approach for pMRI reconstruction in lower level problem \eqref{learnable_pmri_lower} can be cast and formulated as a discrete-time optimal control system. 
We use one data sample $(\fbf, \ubf^*)$ and omit the average and data indexes for notation simplicity. The architecture of deep unrolling method follows the iterations of optimization algorithms and solve for the minimizer of the following problem: 
\begin{subequations}\label{eq:bilevel}
\begin{align}
    \min_{\Theta}\ \ \  & \ell(\ubf_{\Theta}), \label{eq:bilevelupper} \\ 
    \mathrm{s.t.}\ \ \ & \ubf{(t)} =   \gbf(\ubf{(t-1)},\theta(t)),\ t= 1,\cdots, T, \label{eq:bilevelconstraint}\\
 & \ubf{(0)} =  \gbf_0(\fbf,\theta(0)), \label{eq:init}
\end{align}
\end{subequations}
where $\Theta = (\theta(0),\cdots, \theta(T))^{\intercal}$ is the collection of control variables $\theta(t)$ at all time steps (phases) respectively.  In \eqref{eq:bilevelupper},  $\ubf_{\Theta} = \ubf(T)$, which is the output image from the last $T$-th phase of the entire network so that we want to find optimal $\ubf(T)$ that close to  $\argmin_{\ubf}  \phi_{\Theta}(\ubf)$. Equations \eqref{eq:bilevelconstraint}, \eqref{eq:init} are inspired by deep unrolling algorithm for solving the lower level constraint \eqref{learnable_pmri_lower}. Given $\theta(t)$, $\ubf(t)$ is the state of updated reconstruction multi-coil data from $t$-th phase for each $t=0,\cdots, T $. $\gbf$ is a multi-phase unrolling network inspired by the proximal gradient algorithm, and the output of $\gbf(\cdot) \in \C^{m \times n \times c}$  is the updated multi-coil MRI data from each phase. The network $\gbf$ is the intermediate mapping from $\ubf{(t)}$ to $\ubf{(t+1)}$ for $t=0,\cdots, T-1 $, whose structure is explained in Section \ref{MLM} for minimizing the variational model \eqref{eq:m}. The network $\gbf_0$ with initial control parameter $\theta(0)$ maps the partial k-space measurement $ \fbf$ to an initial reconstruction $ \ubf{(0)}$ as the input of this optimal control system. 

To summarize in brief, we solve for the minimizer (reconstruction result) of the lower level problem \eqref{learnable_pmri_lower} in a discrete-time optimal control framework \eqref{eq:bilevel}. The dynamic system \eqref{eq:bilevelconstraint}, \eqref{eq:init} as the constrain of  \eqref{eq:bilevelupper} is modeled as the optimal control system of the variational model \eqref{eq:m}.

The loss function $\ell(\ubf{(T)})$ measures the discrepancy between the final state $\ubf{(T)}$ and the reference image $\ubf^*$ obtained using full k-space data in the training data set. We set the loss function in  \eqref{learnable_pmri_upper} and \eqref{eq:bilevelupper} for the proposed method as:
\begin{equation}\label{eq:loss}
\begin{aligned}
& \ell(\ubf_{\Theta}) = \ell(\ubf{(T)}) = \sum\nolimits^{c}_{i=1} \gamma \| \ubf_i{(T)} - \ubf^*_i\| +  \\ 
& \| |\J(\Bar{\ubf}{(T)})|- \text{RSS}(\ubf^*)\| + \eta \| \text{RSS}(\Bar{\ubf}{(T)}) - \text{RSS}(\ubf^*)\|,  
\end{aligned}
\end{equation}
where $\text{RSS}(\ubf^*)= (\sum_{i=1}^{c} |\ubf^*_i|^2)^{1/2} \in\mathbb{R}^{m \times n}$ is the pointwise root of sum of squares across the $c$ channels of $\ubf^*$, $|\cdot|$ is the pointwise modulus, and $\gamma, \eta > 0$ are prescribed weight parameters. 

The motivation of applying learnable image space regularization $R \circ \J$ and k-space regularization $R_f \circ \Fbf$ in model \eqref{eq:m} is explained in the following: (i) Image domain network recovers the high spatial resolution, but may not suppress some artifacts. Frequency domain network is more suitable to remove high-frequency artifacts. (ii) Image domain and k-space information are equivalent due to the global linear transformation, but adding nonlinear activations with CNNs can feasibly improve the efficacy of network learning and boost the reconstruction performance.
We parametrize the combination operator $\J$ by CNNs since the partial k-space data were scanned by multiple coil arrays, and introduce cross-correlation among channels of coil-images which could be compatible for CNN structure.

\subsection{Design of $\gbf$ and $\gbf_0$}\label{subsec:g}

Denote  function $ h(\ubf) $ as the data fidelity term, one of the famous traditional method for solving problem $ \min_{\ubf} h(\ubf) + R(\ubf) $ is the
\emph{proximal gradient algorithm} \cite{parikh2014proximal}:
\begin{subequations}\label{ista_alg}
\begin{align}
\bbf^l &= \ubf^l - \alpha^l \nabla h(\ubf^l),\\
\ubf^{l+1} &= \prox_{\alpha^l R} (\bbf^l),
\end{align}
\end{subequations}
where $\alpha^l >0$ is the step size, the proximity operator $\prox_{\alpha R}$ defined below: 
\begin{equation}\label{prox}
\prox_{\alpha R}(\bbf) := \argmin_{\xbf} \left\{\frac{1}{2 \alpha} \| \xbf - \bbf \|^2 + R(\xbf) \right\}, \end{equation} 
where $[\xbf]_i = \xbf_i \in \mathbb{C}^{m \times n}$ for any $\xbf \in \mathbb{C}^{m \times n \times c}$.
The proposed network structure is inspired by \eqref{ista_alg} for solving \eqref{learnable_pmri_lower} which can be cast as an iterative procedure with $T$ phases in the discrete dynamic system \eqref{eq:bilevelconstraint} and \eqref{eq:init}.
We parametrize the $t$-th phase consists of three steps:
\begin{small}
\begin{subequations}\label{eq:newmodel} 
\begin{align}
\bbf_i{(t)}  = \ubf_i{(t)} - \rho_t  \Fbf^{H}  \Pbf^{\top}  (\Pbf \Fbf \ubf_i{(t)} - \fbf_i),  & \quad i = 1,\cdots, c, \label{eq:newbi}  \\
\bar{\ubf}_i{(t)}   = [\prox_{\rho_t R(\J(\cdot)) } (\bbf{(t)})]_i, & \quad i = 1,\cdots, c, \label{eq:newubar} \\
{\ubf}_i{(t+1)}  = [\prox_{ \rho_t R_f(\Fbf(\cdot))}  (\bar{\ubf}{(t)})]_i, & \quad i = 1,\cdots, c \label{eq:newui},
\end{align}
\end{subequations}
\end{small}
for $t=0,\cdots, T-1 $ and $\bbf{(t)} = (\bbf_1{(t)},\dots,\bbf_{c}{(t)}) \in \mathbb{C}^{m \times n \times c}$,  $\rho_t>0$ is the step size. $F$ denotes the normalized discrete Fourier transform and $F^{H}$ the complex conjugate transpose (i.e., Hermitian transpose) of $F$, here $F^{H}$ is the inverse discrete Fourier transform. 

The step \eqref{eq:newbi} computes $ \bbf{(t)} $ by applying the gradient decent algorithm to minimize the data fidelity term in \eqref{eq:m} which is straightforward to compute.
The first proximal update step \eqref{eq:newubar} equipped with the joint regularizer $ R(\J(\cdot))$ and upgrades its input $ \bbf_i{(t)} $ to a multi-coil image $\bar{\ubf}_i{(t)}$.
Ideally, the regularization $R\circ \J$ can be parameterized as a deep neural network whose parameters can be adaptively learned from data, however, in such case $\prox_{\rho_t R(\J(\cdot)) }$ does not have closed form and can be difficult to compute.
As an alternative, the proximity operator $\prox_{\rho_t R(\J(\cdot)) }$ can be directly parametrized as a learnable denoiser and solve \eqref{eq:newubar} in each iteration. 
For the similar reason,  the proximity operator $\prox_{R_f(\Fbf(\cdot))}$  in \eqref{eq:newui} can also be
parametrized as CNNs, which further improves the accuracy of the k-space measurement.
In both \eqref{eq:newubar} and \eqref{eq:newui}, the regularizers $R \circ \J$ and $R_f \circ \Fbf $ extract complex features through neural networks and the proximity points can be learned in denoising network via ResNet \cite{7780459} structure.

We frame step \eqref{eq:newubar} incorporates joint reconstruction to update coil-images via ResNet \cite{7780459}: $\Bar{\ubf}{(t)} =  \bbf{(t)} + \M( \bbf{(t)})$,  where $ \M$ represents a multi-layer CNN by executing approximation to the proximal mapping in the image space.
We propose to first learn a nonlinear operator $\J$ that combines $\{ \bbf_i\}$ into the image $ \zbf = \J (\bbf_1,\dots,\bbf_{c}) \in \mathbb{C}^{m\times n}$ with homogeneous contrast. Then apply a nonlinear operator $\G$ on $\zbf$ with $ \G(\zbf) \in \mathbb{C}^{m\times n\times N_f}$ to extract a $N_f$-dimensional features. The nonlinear operator $\J$ contains four convolutions with an activation function in between, each convolution obtains kernel size $ 3\times3$. The nonlinear operator $\G$ consists of three convolutions with filter size $ 9\times 9$. For the sake of improving the capacity of the network, $ \tilde{\J}$ and $ \tilde{\G}$ are employed as adjoint operators of $ \J$ and $ \G$ respectively with symmetric structure and parameters are trained separately. $\G \circ \J $ was designed in the sense of playing a role as an encoder and $\tilde{\J} \circ \tilde{\G}$ as a decoder. Therefore, image domain network can be parametrized as compositions of four CNN operators: $\M = \tilde{\J} \circ \tilde{\G} \circ \G \circ \J $ with output $ \M( \bbf{(t)})\in \mathbb{C}^{m \times n\times c}$. 
This step \eqref{eq:newubar} outputs the multi-coil data $ \Bar{\ubf}  = (\Bar{\ubf}_1, \dots, \Bar{\ubf}_{c})  \in \mathbb{C}^{m \times n \times c}$ from image domain network and we apply the combination operator $\J$ on $\Bar{\ubf}$ to obtain a full FOV MR single-channel image $ \vbf = \J(\Bar{\ubf})  \in \mathbb{C}^{m\times n}$ that we desired for reconstruction.

Step \eqref{eq:newui} leverages k-space information and further suppresses the high-frequency artifacts. The output  $\Bar{\ubf}{(t)} $ from \eqref{eq:newubar} is passed to a k-space domain network by a ResNet structure
$\ubf{(t+1)} = \ \Bar{\ubf}{(t)} +  \Fbf^{H}  \K \big( \Fbf(\Bar{\ubf}{(t)} ) \big)$, where $ \Fbf^{H}  \K  \Fbf$ is a CNN operator to refine and further improve the accuracy of the k-space data from each coil. The CNN  $\K$  consists four convolutions, the last convolution kernel numbers meets the channel number of coil images.

Therefore, \eqref{eq:newmodel} is proceed in the following scheme:
\begin{small}
\begin{subequations}  \label{eq:scheme}
\begin{align}
\bbf_i{(t)} =  \ \ubf_i{(t)} - \rho_t \Fbf^{H}  \Pbf^{\top} (\Pbf \Fbf \ubf_i{(t)} - \fbf_i),& \quad i = 1,\cdots, c, \label{eq:schemeb} \\
\bar{\ubf}_i{(t)}  =  \ \bbf_i{(t)} + \M( \bbf_i{(t)}), & \quad i = 1,\cdots, c, \label{eq:schemeu-bar}   \\
\ubf_i{(t+1)} = \ \Bar{\ubf}_i{(t)} +   \Fbf^{H}  \K \big( \Fbf(\Bar{\ubf}_i{(t)} ) \big), & \quad i = 1,\cdots, c, \label{eq:schemeu} 
\end{align}
\end{subequations}
\end{small}
for $ t = 0,\cdots,T-1$.
Now we can derive the function $\gbf$ described in \eqref{eq:bilevel} by combining \eqref{eq:scheme}:
\begin{subequations}\label{eq:g}
\begin{align}
 \ubf_i{(t+1)}&  =  ( \N - \rho_t \Fbf^{H}  \Pbf^{\top} \Pbf \Fbf -  \N( \rho_t \Fbf^{H}  \Pbf^{\top} \Pbf \Fbf) ) \ubf_i{(t)} \notag\\
&  + \ubf_i{(t)} + \rho_t \Fbf^{H}  \Pbf^{\top} \fbf_i-\N(\rho_t \Fbf^{H}  \Pbf^{\top} \fbf_i) \label{eq:funcg}\\
& = \gbf(\ubf_i{(t)}), \label{eq:gu}
\end{align}
\end{subequations}
where $ \N =\M + \Fbf^{H} \K \Fbf + \Fbf^{H} \K \Fbf  \M$, and the function $\gbf$ maps $\ubf_i{(t)}$ to $\ubf_i{(t+1)}$ for $t=0,\cdots, T-1 $ defined in \eqref{eq:funcg}.

Our proposed reconstruction network is composed of a prescribed $T$ phases, the initial input of the network is designed as $\ubf{(0)} =  \gbf_0(\fbf, \theta(0)):= \Fbf^{H}(\fbf+ \K_0(\fbf))$ where $\K_0$ is a CNN operator applied to $\fbf$ in residual learning to interpolate the missing data and generate a pseudo full k-space. 
With the chosen initial $\{ \ubf_i{(0)}\}$ and partial k-space data $\{\fbf_i\}$ as input, the network performs the update \eqref{eq:scheme} in the $t$-th phase for $t=0,\cdots, T-1$  and finally the entire network reconstructs coil-images $\ubf{(T)}$ and $ \J (\Bar{\ubf}{(T)}) $, which is the final single-body image reconstructed as a by-product (complex-valued). 

Fig. \ref{fig:all-phases} displays the flowchart of our entire network procedure. The initial reconstruction suppresses artifacts caused by undersampling.  We display the flowchart of these CNN structures of each phase in Fig. \ref{fig:t+1phase}. %
We apply complex-valued convolutions where multiplications between complex numbers are performed and use componentwise complex-valued activation function $\C$ReLU($a+ib$) = ReLU($a$) + $i$ReLU($b$) as suggested in \cite{cole2020analysis}. 
The Landweber update step \eqref{eq:schemeb} plays a role in increasing communication between image space and k-space. The learnable step size $\rho_t$ controls the speed and stability of the convergence. In the image space denoising step \eqref{eq:schemeu-bar}, the operator $\J$ extracts feature across all the channels so that spatial resolution is improved and tissue details are recovered in the channel-combined image. The CNN $\M$ carrying channel-integration $\J$ proceed $T$ times with shared weights in every two phases therefore the spatial features across channels are learned in an efficient way. However, the oscillatory artifacts could be misinterpreted as real features, which might be sharpened.  In the k-space denoising update step \eqref{eq:schemeu}, the k-space network $\K$ is compatible with low-frequency information, so it releases the high-frequency artifacts and recovers the structure of the image \cite{doi:10.1002/mrm.27201, wang2020ikwi}. Therefore 
iterating the networks $\M$ and $\K$ in different domains with their individual effects,  the performance compensates and the shortcomings of both networks offset each other. We evaluate the effect of hybrid domain reconstruction in ablation studies. 
Furthermore, the prescribed denoising networks in \eqref{eq:schemeu-bar} and \eqref{eq:schemeu} refine and update coil-images in each iteration. This iterative procedure triggers the reconstruction quality of $ \J (\Bar{\ubf}{(t)}) $ get successively enhancement. 
 
\begin{figure}[t]
    \centering
    \includegraphics[width=1\linewidth]{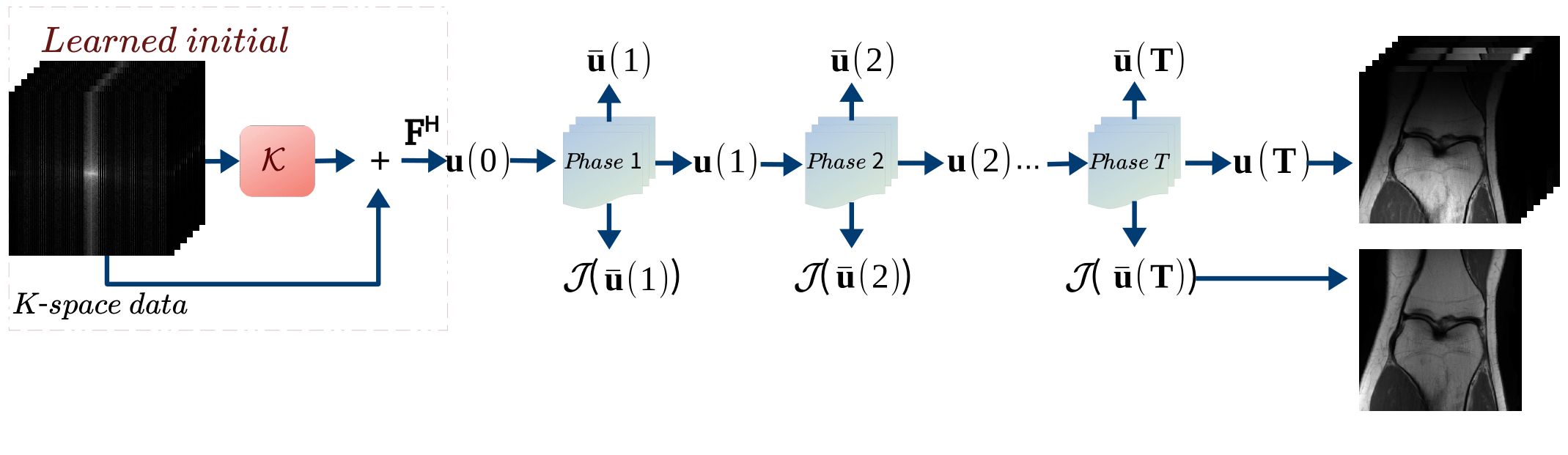}
    \caption{The proposed framework paradigm for  all  phases.}
    \label{fig:all-phases}
\end{figure}
\begin{figure}[t]
    \centering
    \includegraphics[width=1\linewidth]{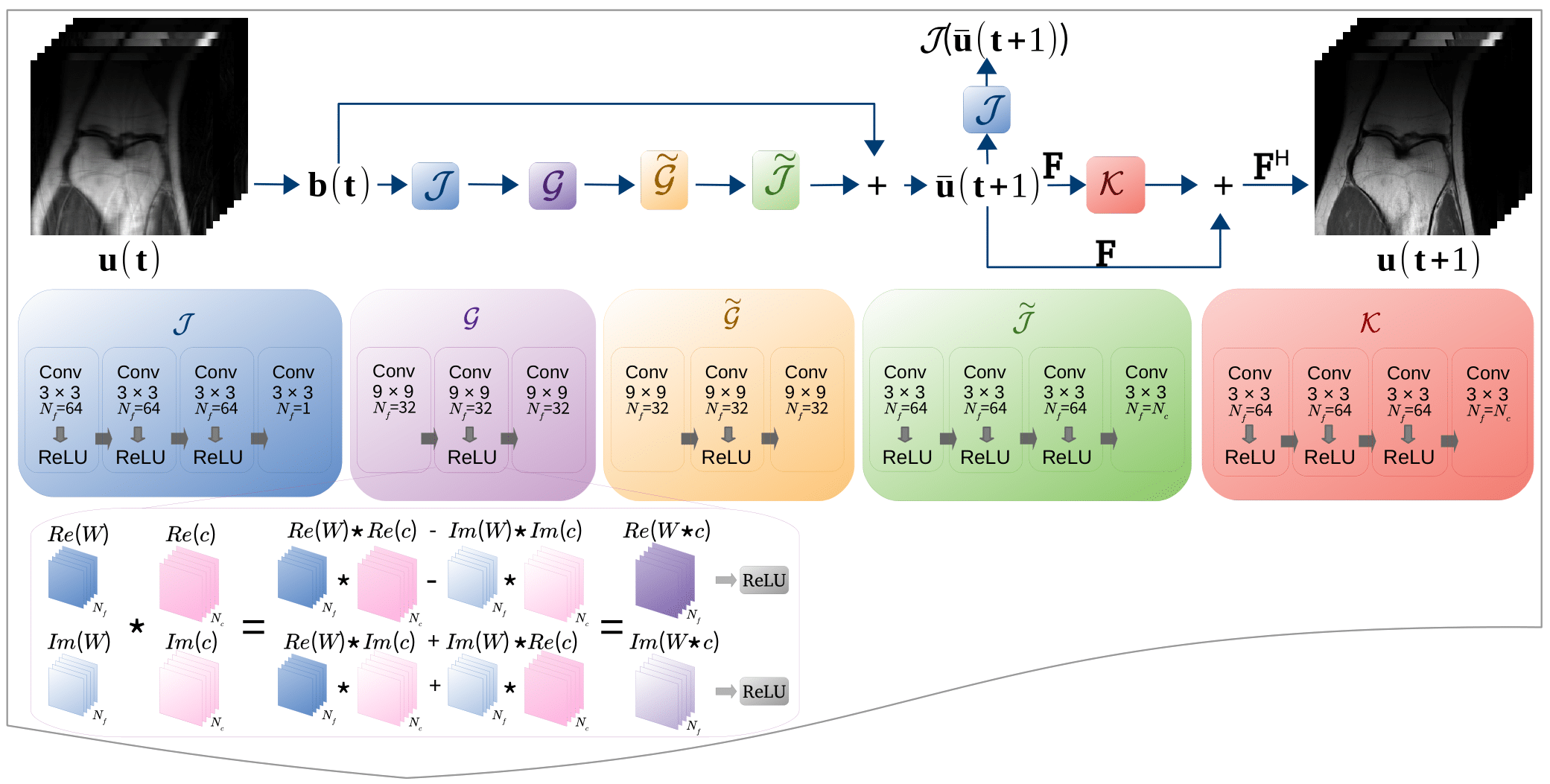}
    \caption{The proposed framework paradigm at $ t+1$-th phase.
   $W = Re(W) + i  Im(W)$ representing complex convolution filter kernels.}
    \label{fig:t+1phase}
\end{figure}

\subsection{Network training from the view of the Method of Lagrangian Multipliers (MLM)}\label{MLM}
The network parameters to be solved from \eqref{eq:bilevel} are $\Theta = \{\theta(t): t = 0,\cdots,T\}$, where $\theta(t)=\{\rho_{t},\J_{t},\G_{t},\tilde{ \G}_{t}, \tilde{ \J}_{t}, \K_{t} \}$ for $t =1, \cdots, T$ and $\theta(0)=\K_{0} $.

The control problem \eqref{eq:bilevel} can be solved by using MLM. The corresponding Lagrangian function is
\begin{equation}
\label{eq:lag}
\begin{aligned}
L(\textbf{U},\Theta;\Lambda) = & \ell(\ubf{(T)}) +  \langle \lambda(0) , \ubf{(0)} - \gbf_0 \big(\fbf , \theta(0) \big) \rangle \\
&+ \sum\nolimits^{T}_{t=1} \langle \lambda(t), \ubf{(t)} - \gbf\big(\ubf{(t-1)},  \theta(t) \big) \rangle,
\end{aligned}
\end{equation}
where $\textbf{U}=(\ubf{(0)}, \cdots, \ubf{(T)})^{\intercal}$ is the collection of all the states $\ubf{(t)}$,  $\Lambda = (\lambda(0), \cdots, \lambda(T))^{\intercal}$ are Lagrangian multipliers of \eqref{eq:bilevel}. 
The algorithm proceed in the iterative scheme to update $\Theta^k$, for each training epoch  $k=0, \cdots, K$, $\Theta^{k}=(  \theta^{k}(0),\cdots, \theta^{k}(T))^{\intercal}$, $\textbf{U}^k =( \ubf^{k}{(0)}, \cdots,  \ubf^{k}{(T)})^{\intercal}$, and $\Lambda^k = (\lambda^k(0), \cdots,  \lambda^k(T))^{\intercal}$ . 

If $(\textbf{U}^*, \Theta^*;\Lambda^*)$ minimizes Lagrange function \eqref{eq:lag}, by the first order optimality condition, we have
\begin{subnumcases}{}
 \partial_{\textbf{U}}  L (\textbf{U}^*, \Theta^*; \Lambda^*) = 0 \label{dU}\\
 \partial_{\Theta} L (\textbf{U}^*, \Theta^*; \Lambda^*) = 0 \label{dTheta}\\
 \partial_{\Lambda}  L (\textbf{U}^*, \Theta^*; \Lambda^*) = 0  \label{dLambda}
\end{subnumcases}
\textbullet Fix $\Theta = \Theta^k$, define $(\textbf{U}^k,\Lambda^k):= \argmin_{\textbf{U},\Lambda}  L(\textbf{U},\Theta^k;\Lambda) $, then by the first order optimality condition for minimizing $L$ w.r.t $ \lambda(t)$ for $ t=0,\cdots,T$,  $ (\textbf{U}^k,\Lambda^k)$  should satisfy 
\begin{equation}
\begin{aligned}
&\partial_{\lambda(0)}\langle \lambda(0), \ubf{(0)} - \gbf_0(\fbf , \theta^k(0)) \rangle \Big|_{(\ubf^k{(0)}, \lambda^k(0))} =0  \\
&\implies \ubf^k{(0)} = \gbf_0(\fbf , \theta^k(0)) \label{eq:ut0}\\
\end{aligned}
\end{equation}
\begin{equation}
\begin{aligned}
&\partial_{\lambda(t)}\langle \lambda(t), \ubf{(t)} - \gbf(\ubf{(t-1)},\theta^k(t) ) \rangle \Big|_{(\ubf^k{(t)}, \lambda^k(t))} =0 \\
&\implies \ubf^k{(t)} = \gbf(\ubf^k{(t-1)}, \theta^k(t)), \ t = 1,\cdots, T. \label{eq:ut1} 
\end{aligned}
\end{equation}

Then by the first order optimality condition for minimizing $L$ w.r.t  $\ubf{(t)}$, for $t=T$, we get
\begin{equation}
\begin{aligned}
&\partial_{\ubf{(T)}} [\ell(\ubf{(T)}) + \langle \lambda(t), \ubf{(T)} \rangle ] \Big|_{(\ubf^k{(t)}, \lambda^k(t))} = 0\\ 
&\implies \lambda^k(T) = - \partial_{\ubf{(T)}}  \ell(\ubf^k{(T)}), \label{eq:lambdaT} 
\end{aligned}
\end{equation}
for $t = 0,\cdots,T-1$:
\begin{small}     
   \begin{eqnarray}     
&\partial_{\ubf{(t)}}  [\langle \lambda(t), \ubf{(t)} \rangle- \langle \lambda(t+1), \gbf(\ubf{(t)}, \theta^k{(t+1)}) \rangle] \Big|_{(\ubf^k{(t)}, \lambda^k(t))}  = 0 \notag\\ 
&\implies \lambda^k(t) = \langle \lambda^k(t+1), \partial_{\ubf{(t)}}  \gbf(\ubf^k{(t)}, \theta^k{(t+1)}) \rangle,  \label{eq:lambdat}        
   \end{eqnarray}              
\end{small}    
\textbullet Fix  $(\textbf{U}^k,\Lambda^k) $ for updating $\Theta$,  we compute the gradient $\partial_{\Theta}L(\textbf{U}^k,\Theta; \Lambda^k)$:
\begin{subequations}
\label{eq:dtheta}
\begin{align}
&\partial_{\theta(t)}  L(\textbf{U}^k,\theta(t);\Lambda^k)  = \partial_{\theta(t)} [ -\langle \lambda^k(t), \gbf(\ubf^k{(t-1)}, \theta(t)) \rangle ] \notag\\
& =  -\langle \lambda^k(t), \partial_{\theta(t)} \ \gbf(\ubf^k{(t-1)}, \theta(t)) \rangle, t = 1,\cdots, T, \label{eq:dtheta1}\\
&\partial_{\theta(0)}  L(\textbf{U}^k,\theta(0);\Lambda^k) =  -\langle \lambda^k(0), \partial_{\theta(0)}  \gbf_0(\fbf, \theta(0)) \rangle.\label{eq:dtheta0}
\end{align}
\end{subequations} 
\begin{theorem}\label{thm}
$\partial_{\Theta}L(\textbf{U}^k,\Theta;\Lambda^k) = \partial_{\Theta} \ell(\ubf^k{(T)}(\Theta) )$.
\end{theorem}
\begin{proof}
First we show the following holds for $t=0,\cdots, T$:
\begin{equation}
\label{eq:everyt}
\lambda^k(t)  = - \partial_{\ubf{(t)}}\ell(\ubf^k{(T)}) 
\end{equation}

From \eqref{eq:lambdaT}, \eqref{eq:everyt} is true when $ t=T$.
Suppose \eqref{eq:everyt} true for $ t=\tau \in \{1,\cdots ,T\} $.
From \eqref{eq:lambdat}, we have
\begin{subequations}
\begin{align}
\lambda^k(\tau-1) & = \langle \lambda^k(\tau), \partial_{\ubf{(\tau-1)}} \ \gbf(\ubf^k{(\tau-1)}, \theta^k(\tau)) \rangle  \notag\\
& = \langle - \partial_{\ubf{(\tau)}} \ell(\ubf^k{(T)}), \partial_{\ubf{(\tau-1)}} \ubf^k{(\tau)}  \rangle  \\
& = - \partial_{\ubf{(\tau-1)}}\ell(\ubf^k{(T)})\label{eq:lambda_T-1}
\end{align}
\end{subequations}
Thus, \eqref{eq:everyt} holds for $ t= \tau-1$. By the principle of induction, \eqref{eq:everyt} is ture for all $t=0, \cdots, T$.

Hence, \eqref{eq:dtheta1} reduces to
\begin{subequations}
\label{eq:dL}
\begin{align}
\partial_{\theta(t)} & L(\textbf{U}^k,\theta(t);\Lambda^k) =  -\langle \lambda^k(t), \partial_{\theta(t)} \ \gbf(\ubf^k{(t-1)}, \theta(t) ) \rangle \notag \\
& =  \langle \partial_{\ubf{(t)} } \ell(\ubf^k{(T)})  , \partial_{\theta(t)} \ubf^k{(t)} \rangle  \\
& = \partial_{\theta(t)} \ell ( \ubf^k{(T)} ), \ t = 1,\cdots, T. 
\end{align}
\end{subequations} 
Also \eqref{eq:dtheta0} gives $\partial_{\theta(0)} L(\textbf{U}^k,\theta(0);\Lambda^k) 
= \partial_{\theta(0)} \ell ( \ubf^k{(T)} )$.
Therefore, we derive 
$ \partial_{\Theta} L(\textbf{U}^k,\Theta;\Lambda^k) = \partial_{\Theta} \ell ( \ubf^k{(T)}(\Theta))$.
\end{proof}
This theorem further implies that:
Since the gradients are the same,  applying SGD Algorithms or its variance such as Adam \cite{kingma2014adam} to minimize loss function $ \ell$ is equivalent to perform the same algorithm on $L$. Network training algorithm using MLM can be summarized in  Algorithm \ref{alg:1}. 

\begin{algorithm}
\caption{Network Training by MLM}\label{alg:1}
\textbf{Hyperparameter:} $K$ (\#Iterations)\\
\textbf{Initialize:} Initial guess $\theta^{0}(t) \in \Theta^0, t =0,\cdots, T$\\
\For{$k = 0$ \KwTo $K-1$}{
Set $ \ubf^k{(0)} = \gbf_0(\fbf , \theta(0)) $\\
\For{$t=1$ \KwTo $T$}{
$ \ubf^k{(t)} = \gbf(\ubf^k{(t-1)}, \theta^{k}(t))$ }
Set $ \lambda^k(T) =   - \partial_{\ubf}  \ell(\ubf^k{(T)})$ \\
\For{$t = T-1$ \KwTo  $0$ }{$\lambda^k(t) = \langle \lambda^k{(t+1)}, \partial_{\ubf}  \gbf(\ubf^k{(t)}, \theta^{k}(t+1)) \rangle$ }
\For{$t = 0$ \KwTo  $T$ }{
Set $\theta^{k+1}(t) = \argmin_{\theta(t)} L (\textbf{U}^k, \theta(t);\Lambda^k) $. }
}
\textbf{Output:}  $ \theta^{K}(t) , t=0,\cdots,T $.
\end{algorithm}
\section{Experimental Results}
\label{sec:experiment}
\subsection{Data Set}

The data in our experiments was acquired by a 15-channel knee coil array with two pulse sequences: a proton density weighting with (FSPD) and without (PD) fat suppression in the coronal direction from \url{https://github.com/VLOGroup/mri-variationalnetwork}. We used a regular Cartesian sampling mask with 31.56\% sampling ratio as shown in the lower-right of Fig. \ref{PD}. Each of the two sequences data includes images of 20 patients, we select 27-28 central image slices from 19 patients, which amount to 526 images each of size $ 320 \times 320$ as the training dataset, and 15 central image slices are picked from one patient that is not included in the training data set as testing dataset. 

\subsection{Implementation}
\label{Implementation}
 We evaluate classical methods GRAPPA \cite{griswold2002generalized}, SPIRiT \cite{doi:10.1002/mrm.22428}, and deep learning methods VN \cite{doi:10.1002/mrm.26977}, De-AliasingNet \cite{10.1007/978-3-030-32248-9_4}, DeepcomplexMRI \cite{WANG2020136} and Adaptive-CS-Net \cite{pezzotti2020adaptive} over the 15 testing  Coronal  FSPD and Coronal  PD knee images with regular Cartesian sampling in terms of PSNR, structural similarity (SSIM) \cite{wang2004image} and relative error RMSE.
The following equations are computations of SSIM, PSNR and RMSE between  reconstruction $\vbf = |\J(\Bar{\ubf})| $ and ground truth single-body image $\vbf^*$:
\begin{equation}
    SSIM = \frac{(2\mu_{\vbf} \mu_{\vbf^*} + C_1)(2\sigma_{\vbf \vbf^*} + C_2)}{(\mu_{\vbf}^2 + \mu_{\vbf^*}^2+C_1)( \sigma_{\vbf}^2 + \sigma_{\vbf^*}^2 + C_2)},
\end{equation}
 where $\mu_{\vbf}, \mu_{\vbf^*}$ are local means of pixel intensity, $\sigma_{\vbf}, \sigma_{\vbf^*}$ denote the standard deviation and $\sigma_{\vbf \vbf^*} $ 
is covariance between $\vbf$ and $\vbf^* $, $ C_1 = (k_1 L)^2, C_2 = (k_2 L)^2$ are two constants that avoid denominator to be zero, and $ k_1 =0.01, k_2 =0.03$. $L$ is the largest pixel value of the magnitude of coil-images.
\begin{equation}
    PSNR = 20\log_{10} \big(  \max(\abs{\vbf^*})  \big/ \frac{1}{N}\| \vbf^* - |\J(\Bar{\ubf})|  \|^2 \big),
\end{equation}
where $N$ is the total number of pixels in the magnitude of ground truth.
\begin{equation}
    RMSE = \| \vbf^* - |\J(\Bar{\ubf})|  \| / \| \vbf^* \|.
\end{equation}
The relative error between the multi-coil reconstruction $\ubf$ and the ground truth $\ubf^*$ is defined as
\begin{equation}
    RMSE = \sqrt{  \sum\nolimits^{c}_{i=1} \| \ubf^*_i - \ubf_i\|^2 / \sum\nolimits^{c}_{i=1} \| \ubf^*_i  \|^2  }.
\end{equation}

All the experiments are implemented and tested in TensorFlow \cite{abadi2016tensorflow} on a Windows workstation with Intel Core i9 CPU at 3.3GHz and an Nvidia GTX-1080Ti GPU with 11GB of graphics card memory. The parameters in proposed networks are using Xavier initialization \cite{glorot2010understanding} to initialize $\theta^{0}(t) , t=0,\cdots,T $. Indeed, solving line 13 in the Algorithm \ref{alg:1} using any stochastic gradient descent algorithm such as Adam \cite{kingma2014adam} has the same performance as minimizing loss function w.r.t $\theta(t)$ using the same algorithm to update $ \theta^{k+1}(t)$, because of the equivalence of the gradients: $ \partial_{\Theta} L(\textbf{U}^k,\Theta;\Lambda^k) = \partial_{\Theta} \ell ( \ubf^k{(T)}(\Theta))$, which is proved by Theorem \ref{thm}. TensorFlow provides optimized APIs for automatic differentiation which helps to build highly performant input pipelines and keeps high GPU utilization. In order to  implement a more stable gradient calculation, 
we replace line 13 by minimizing $\ell$ with the Adam algorithm. The network was trained with total epochs $K=700$ to update $\Theta^k$. We apply exponentially decay learning rate 0.0001, $\beta_1 = 0.9, \beta_2 = 0.999, \epsilon =10^{-8}$ and mini-batch size of 2 is used in Adam optimizer.  The initial step size is set to $\rho_0 =1$ for both real and imaginary parts and we choose  $\gamma = 10^{-3}, \eta = 10^{-4}$ in \eqref{eq:loss}. The proposed network was implemented with $T=4$ and parameters trained in the network $\M$ are shared for every two phases.

\subsection{Comparison with existing methods}\label{sub:compare}

The average numerical performance  with standard deviations of the proposed method and several state-of-the-art methods are summarized in Table \ref{tab:result}. The proposed method achieves the best reconstruction quality in terms of PSNR/SSIM/RMSE. Our method is parameter efficient due to the network $\M$ share parameters in every two phases so that the entire network reduces more than 1/4 of learnable parameters. We listed the trained parameter numbers and inference time in table \ref{tab:result}.  Fig. \ref{fig:phase_images} and Table \ref{tab:phase_result} indicate that reconstruction performance get improved progressively as $t$ increases.
\begin{figure}
    \centering
\includegraphics[width=0.24\linewidth]{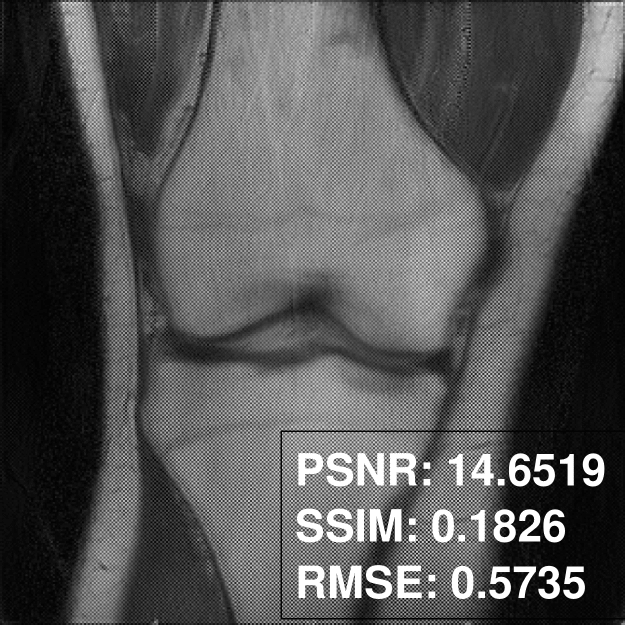}
\includegraphics[width=0.24\linewidth]{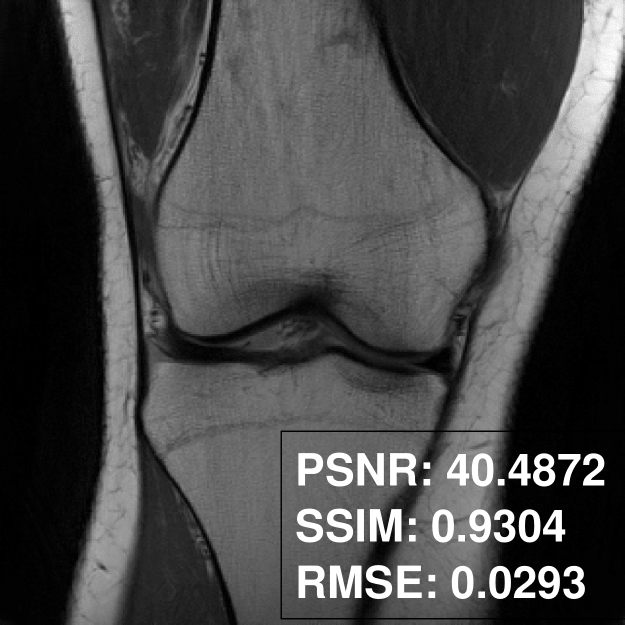}
\includegraphics[width=0.24\linewidth]{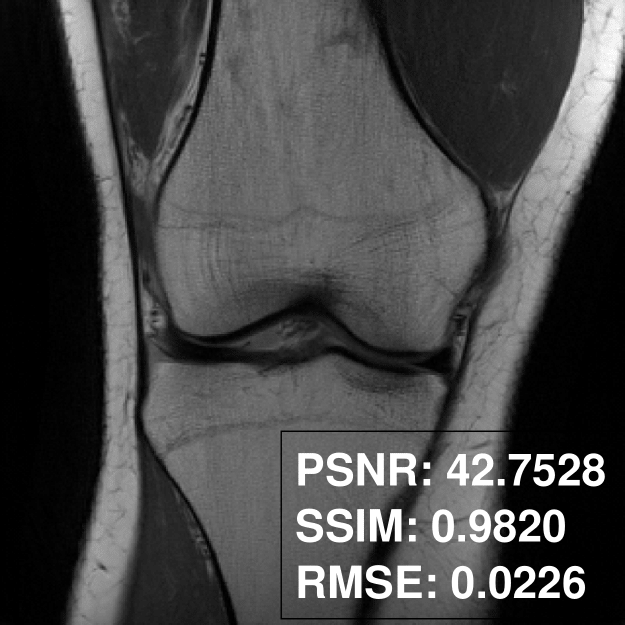}
\includegraphics[width=0.24\linewidth]{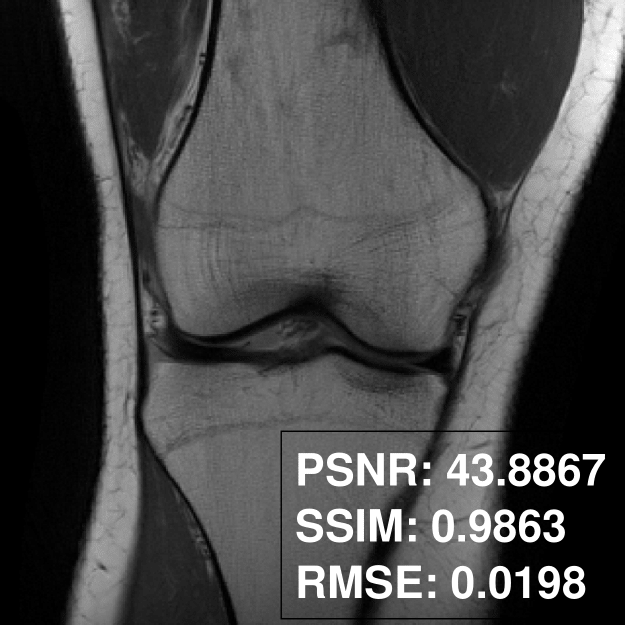}
\caption{ Reconstructed channel-combined images $ \J( \bar{\ubf}{(t)})$ for $ t =1,2,3,4$ of proposed method.}
    \label{fig:phase_images}
\end{figure} 
\begin{table}
\caption{Tested Average PSNR for FSPD data and PD data for each phase.}
\centering
\resizebox{\linewidth}{11mm}{ \begin{tabular}{ccc}
\toprule
Phase number  &  PSNR of FSPD data  & PSNR of PD data  \\
\midrule
1 & 16.2007 & 15.5577\\
2 & 33.7201 & 40.5228\\
3 & 36.9523 & 42.6623\\
4 & 40.7101 & 44.8120\\
 \bottomrule
\end{tabular}}
\label{tab:phase_result}
\end{table}
GRAPPA and SPIRiT adopted calibration kernel size $ 5\times 5$ and Tikhonov regularization in the calibration was set to be 0.01. Tikhonov regularization in the reconstruction was set as $10^{-3}$ for implementing SPIRiT, which took 30 iterations.
In training and testing of VN and DeepcomplexMRI, the network and parameter settings are used as stated in their paper and code. We modified De-AliasingNet by increasing filter numbers to be 64 in 20 iterations to improve the performance for fair competition. In the CDF-Net, we set 128 channels as an initial layer with a depth of 4 in all the three Frequency-Informed U-Nets. We choose $ \sigma$ to be  Softplus activation function. We implemented Adaptive-CS-Net with a total of 15 reconstruction blocks including 2 of 16 filters with size $ 3\times3$, 3 of 32 filters with size $5\times5$, and 2 of 64 filters with size $5\times5$, the blocks are mixed with 2 and 3 scales, which was controlled as the largest tolerance of our GPU capacity. We input the data consistency as prior knowledge of training data to the network, neighboring slices with the center slice are used as input.  For fair competitions, CDF-Net and Adaptive-CS-Net have employed complex convolutions and activation $\C$ReLU.

VN applied precalculated coil sensitivities, and output full-body image. De-AliasingNet, DeepcomplexMRI, CDF-Net, and Adaptive-CS-Net both output multi-coil images and do not require coil sensitivity maps. DeepcomplexMRI uses the adaptive coil combination method \cite{Walsh2000682}, the other networks all applied RSS.  CDF-Net and proposed network perform cross-domain reconstruction, other learning-based methods perform on image domain.
Comparison between referenced pMRI  methods and proposed methods are shown in Fig.~\ref{PD} for PD images and FSPD images are in Fig.~\ref{FSPD}.
We observe the evidence that deep learning-based methods significantly outperform classical methods GRAPPA and SPIRiT in reconstruction accuracy. 
In learning-based methods, the tested images from VN and De-AliasingNet are more blurry than other ones and lost sharp details in some complicated tissue, other methods display only slight differences in the detail.

\begin{table*}
\caption{Quantitative evaluations of the reconstructions on the Coronal FSPD \& PD data and reconstruction time for each of the referenced methods. Time (in seconds) refers to the testing time for each method.} 
\centering
\resizebox{\linewidth}{18mm}{ 
\begin{tabular}{l|ccc|ccc|cc}
\toprule
& & FSPD data & & & PD data & & \\
Method   & PSNR                  & SSIM               & RMSE       & PSNR                & SSIM                & RMSE    &   Time & Parameters\\\midrule
GRAPPA~\cite{griswold2002generalized}   & 24.9251$\pm$0.9341    & 0.4827$\pm$0.0344  & 0.2384$\pm$0.0175 & 30.4154$\pm$0.5924  & 0.7489$\pm$0.0207   & 0.0984$\pm$0.0030 &  280s  & N/A\\
SPIRiT~\cite{doi:10.1002/mrm.22428}  & 28.3525$\pm$1.3314    & 0.6509$\pm$0.0300  & 0.1614$\pm$0.0203 & 32.0011$\pm$0.7920  & 0.7979$\pm$0.0306   & 0.0824$\pm$0.0082 &  43s & N/A\\
VN~\cite{doi:10.1002/mrm.26977} & 30.2588$\pm$1.1790    & 0.7141$\pm$0.0483  & 0.1358$\pm$0.0152 & 37.8265$\pm$0.4000  & 0.9281$\pm$0.0114   & 0.0422$\pm$0.0036 & 0.16s & 0.13 M\\
De-AliasingNet~\cite{10.1007/978-3-030-32248-9_4} & 36.1017$\pm$1.2981    & 0.8941$\pm$0.0269 & 0.0697$\pm$0.0119 & 41.2151$\pm$0.7872  & 0.9711$\pm$0.0033   & 0.0285$\pm$0.0015 & 0.92s & 0.34 M\\ 
DeepcomplexMRI~\cite{WANG2020136} & 36.5706$\pm$1.0215   & 0.9008$\pm$0.0190 & 0.0654$\pm$0.0049 & 41.5756$\pm$0.6271  & 0.9679$\pm$0.0031   & 0.0274$\pm$0.0018 & 1.04s & 1.45 M\\
CDF-Net~\cite{10.1007/978-3-030-59713-9_41} & 40.0405$\pm$1.5518 & 0.9520$\pm$0.0182 & 0.0443$\pm$0.0070 & 43.7943$\pm$1.9888 & 0.9879$\pm$0.0026 & 0.0215$\pm$0.0042 & 0.47s & 3.44 M\\
Adaptive-CS-Net~\cite{pezzotti2020adaptive} & 40.1846$\pm$1.4780 & 0.9534$\pm$0.0175 & 0.0435$\pm$0.0065 & 44.1131$\pm$1.5596 & 0.9878$\pm$0.0026 & 0.0206$\pm$0.0027 & 1.57s & 5.59 M\\
Proposed & \textbf{40.7101$\pm$1.5357} & \textbf{0.9619$\pm$0.0144} & \textbf{0.0408$\pm$0.0051} & \textbf{44.8120$\pm$1.3185} & \textbf{0.9886$\pm$0.0023} & \textbf{0.0189$\pm$0.0018} & 0.52s & 2.92 M\\
 \bottomrule
\end{tabular}}
\label{tab:result}
\end{table*}
\begin{figure*}
\includegraphics[width=0.10\linewidth, angle=180]{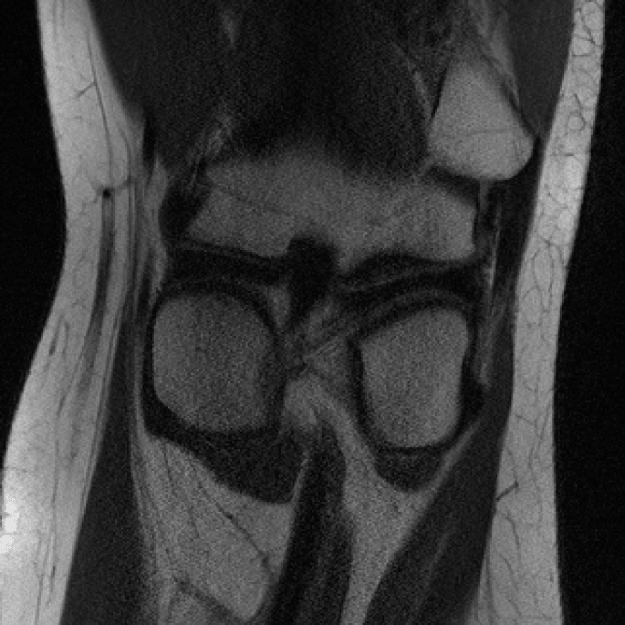}
\includegraphics[width=0.10\linewidth, angle=180]{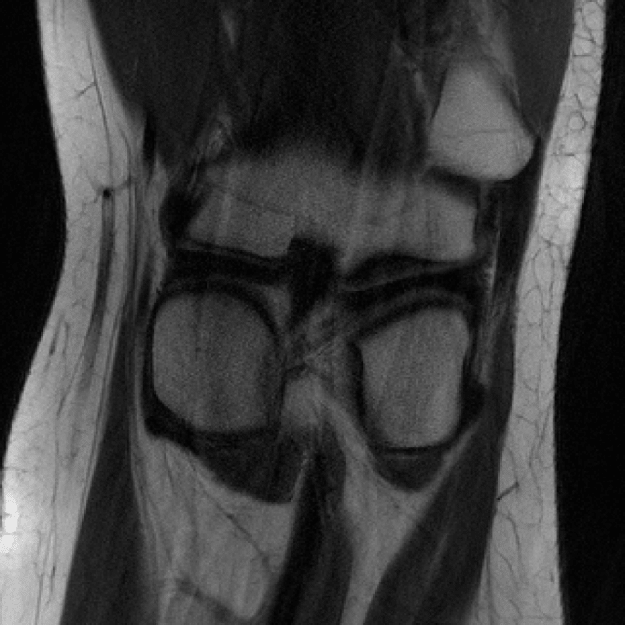}
\includegraphics[width=0.10\linewidth, angle=180]{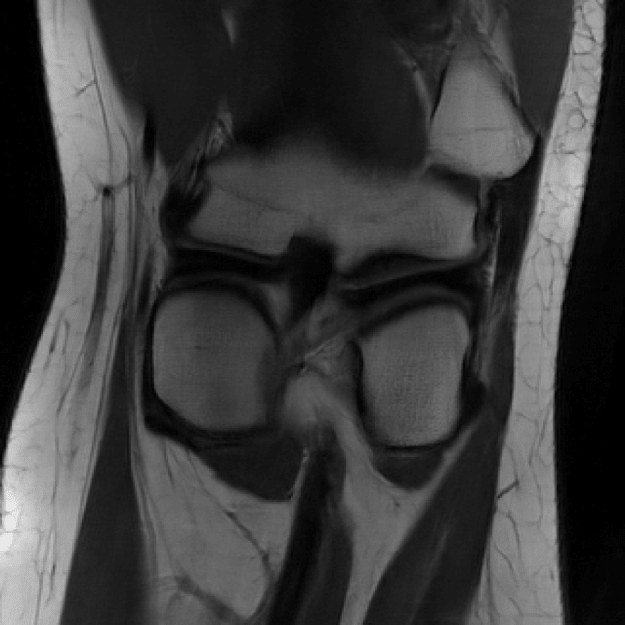}
\includegraphics[width=0.10\linewidth, angle=180]{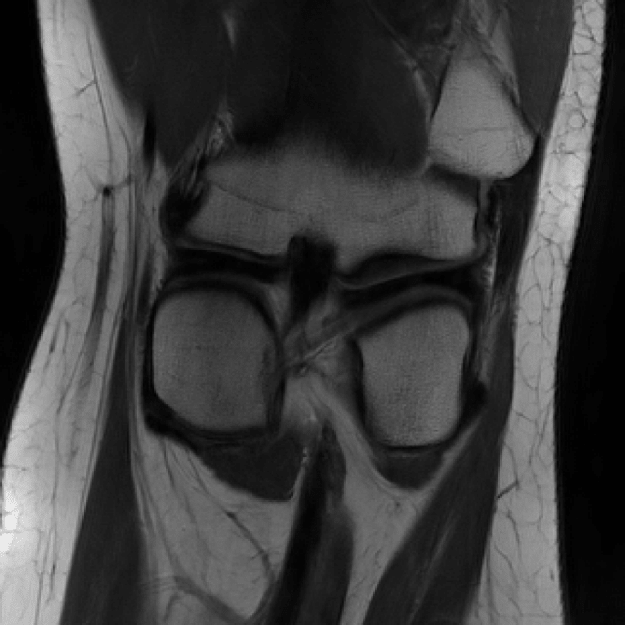}
\includegraphics[width=0.10\linewidth, angle=180]{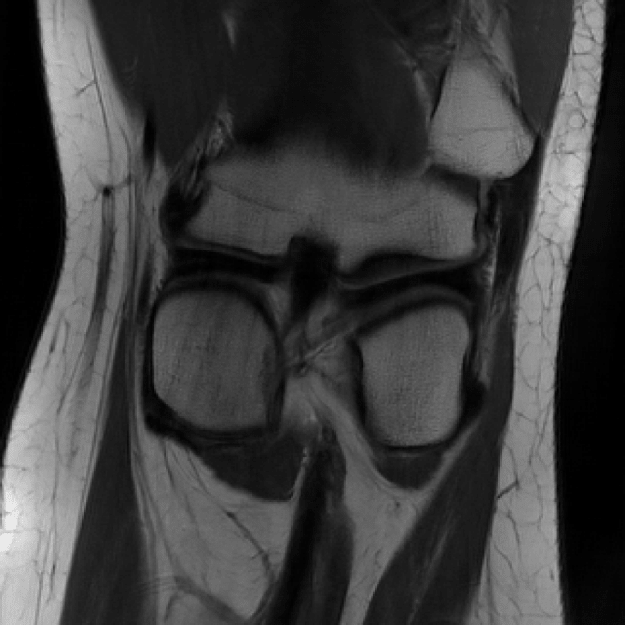}
\includegraphics[width=0.10\linewidth, angle=180]{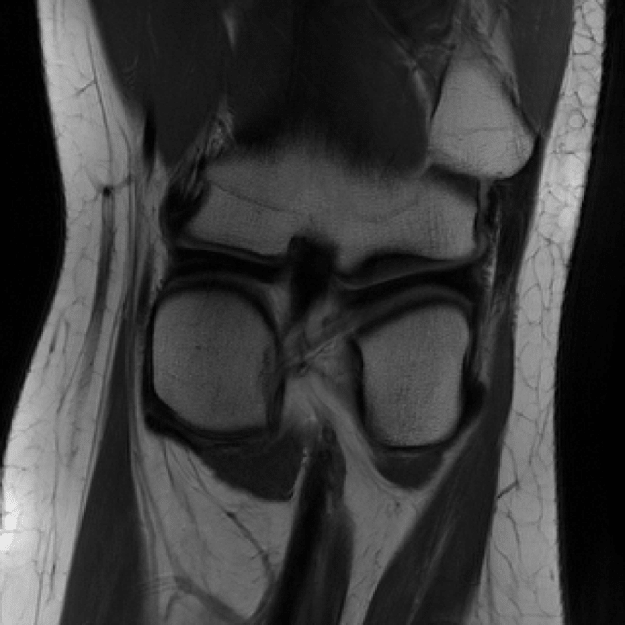}
\includegraphics[width=0.10\linewidth, angle=180]{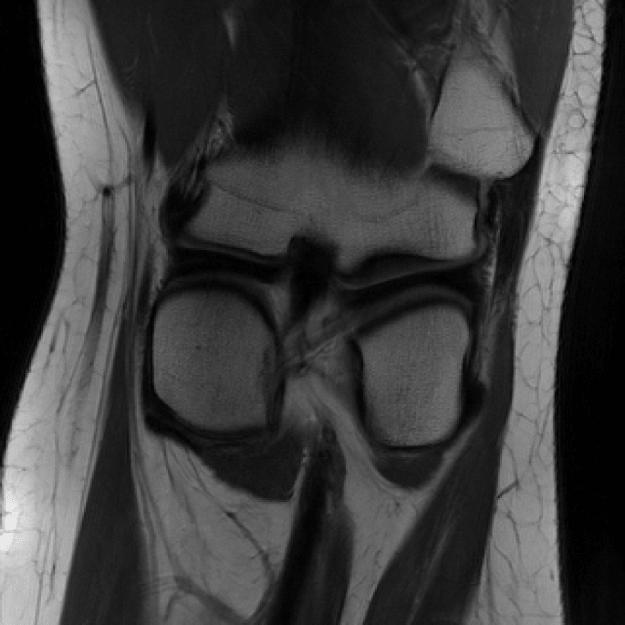}
\includegraphics[width=0.10\linewidth, angle=180]{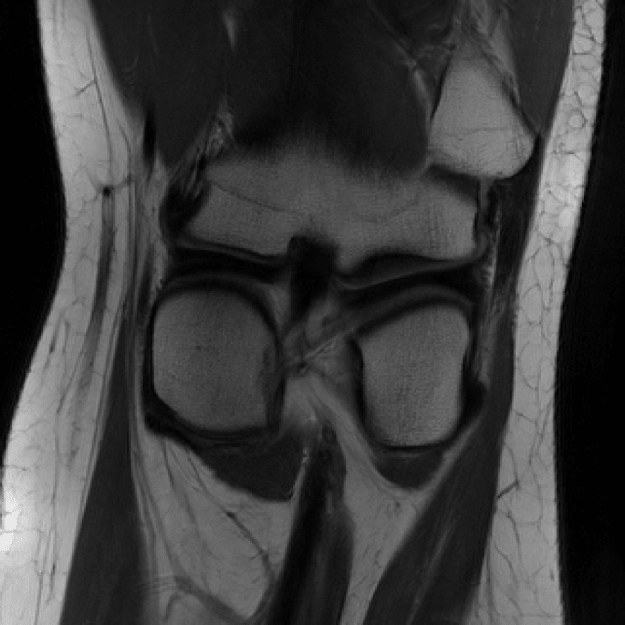}
\includegraphics[width=0.10\linewidth, angle=180]{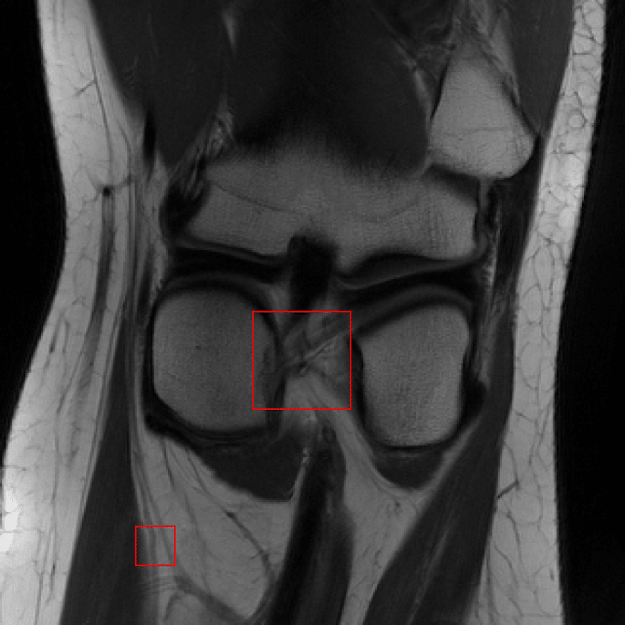}\\
\includegraphics[width=0.10\linewidth, angle=180]{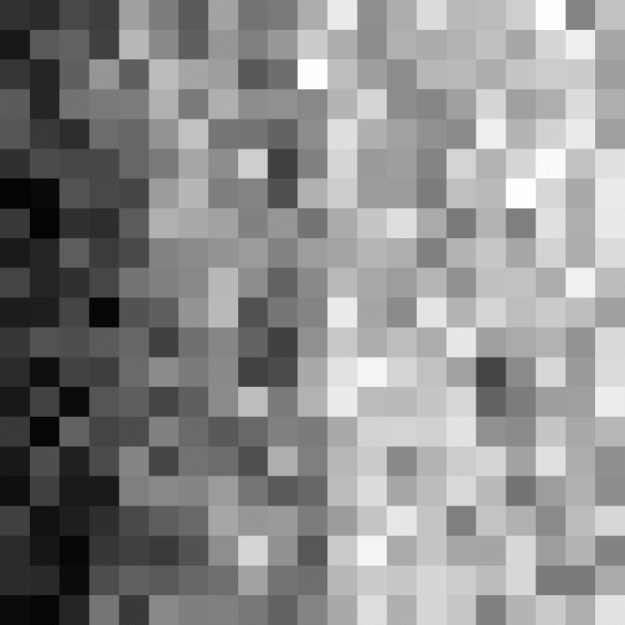}
\includegraphics[width=0.10\linewidth, angle=180]{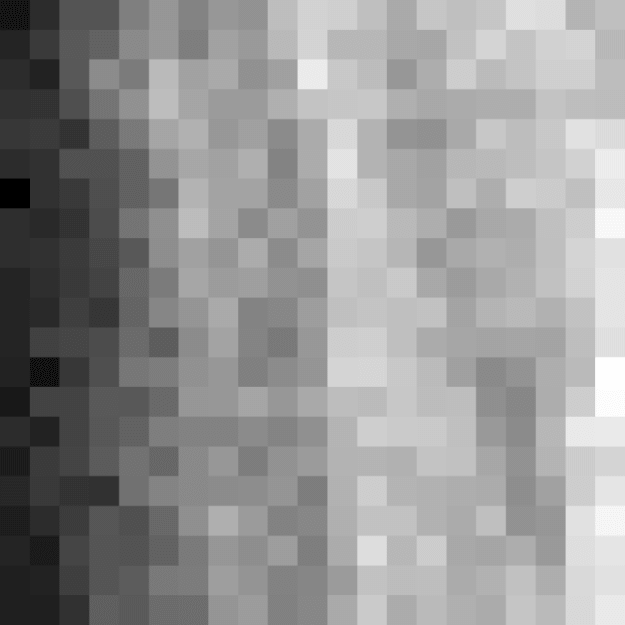}
\includegraphics[width=0.10\linewidth, angle=180]{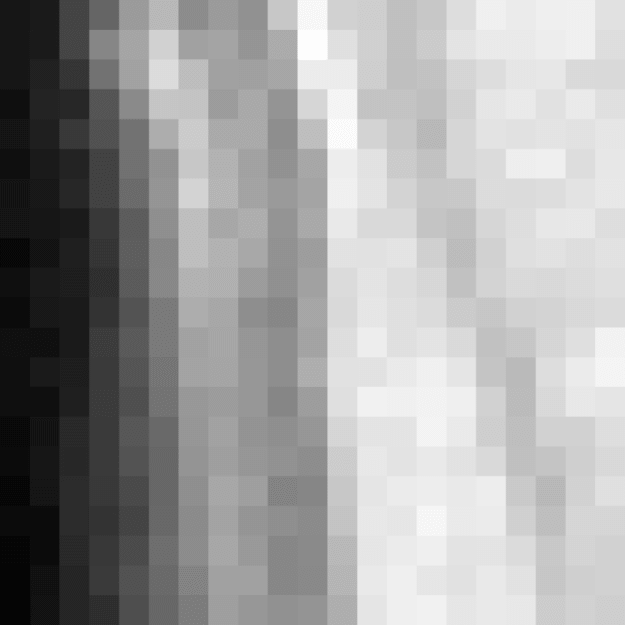}
\includegraphics[width=0.10\linewidth, angle=180]{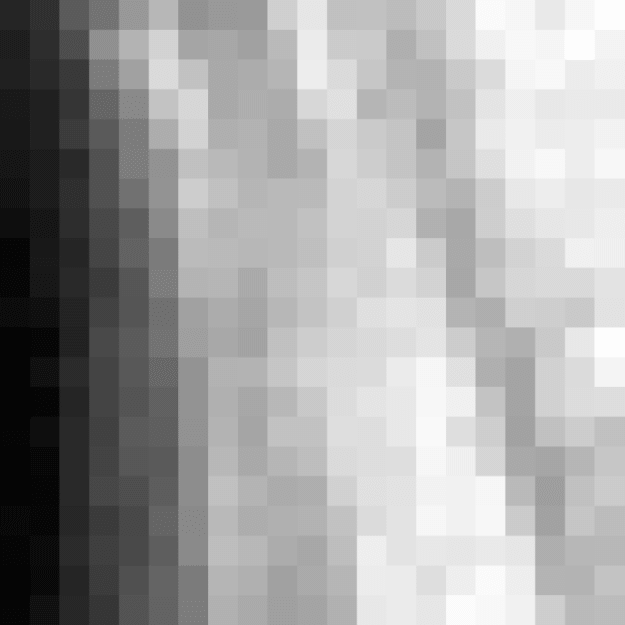}
\includegraphics[width=0.10\linewidth, angle=180]{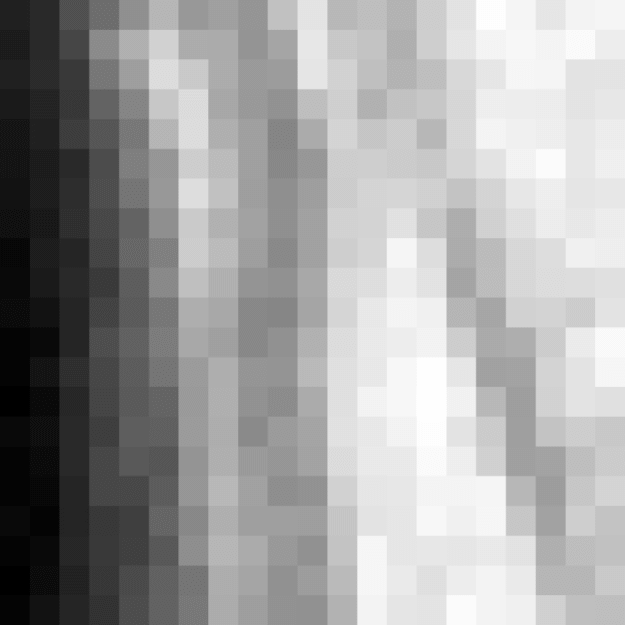}
\includegraphics[width=0.10\linewidth, angle=180]{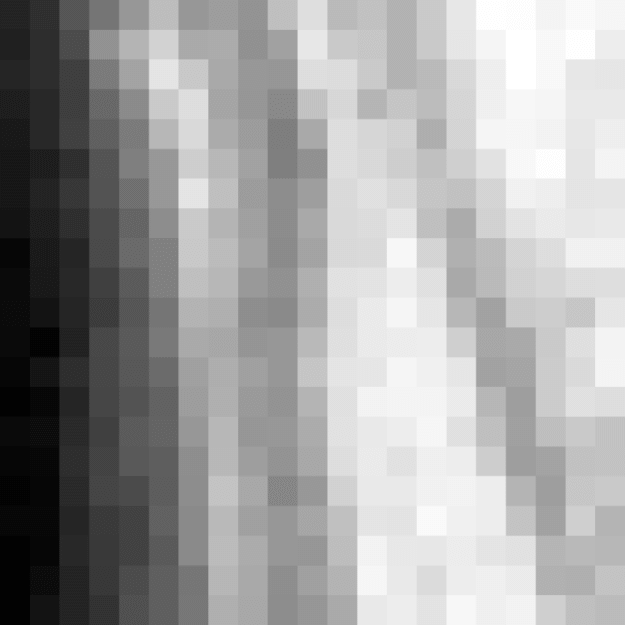}
\includegraphics[width=0.10\linewidth, angle=180]{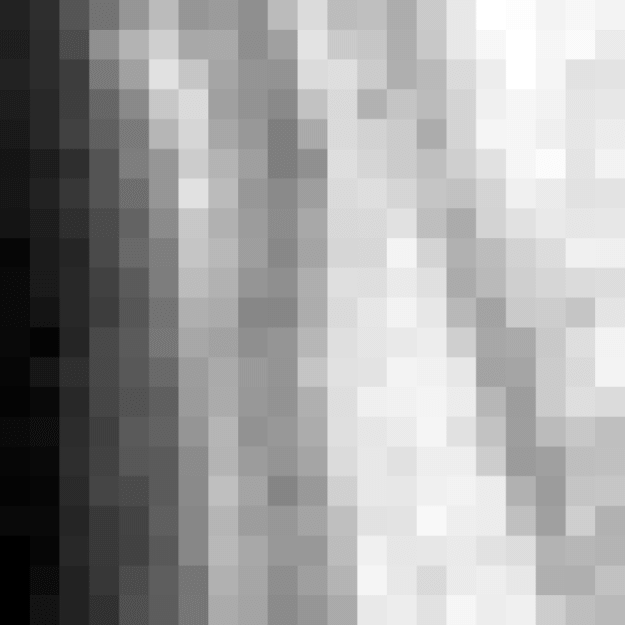}
\includegraphics[width=0.10\linewidth, angle=180]{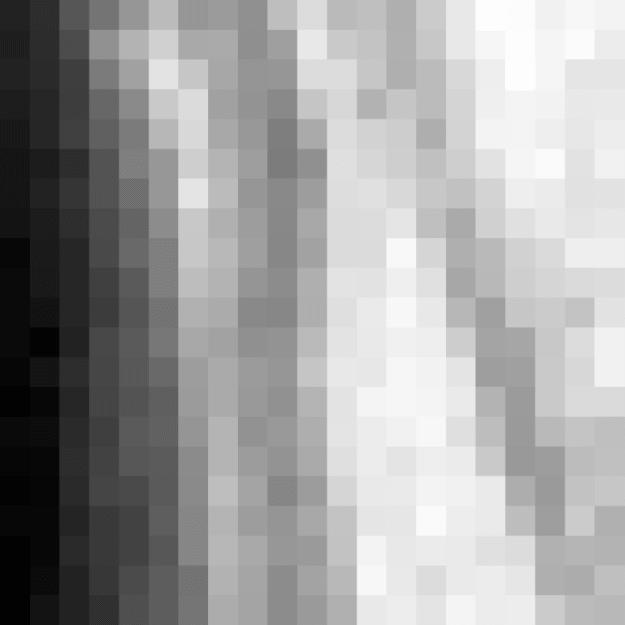}
\includegraphics[width=0.10\linewidth, angle=180]{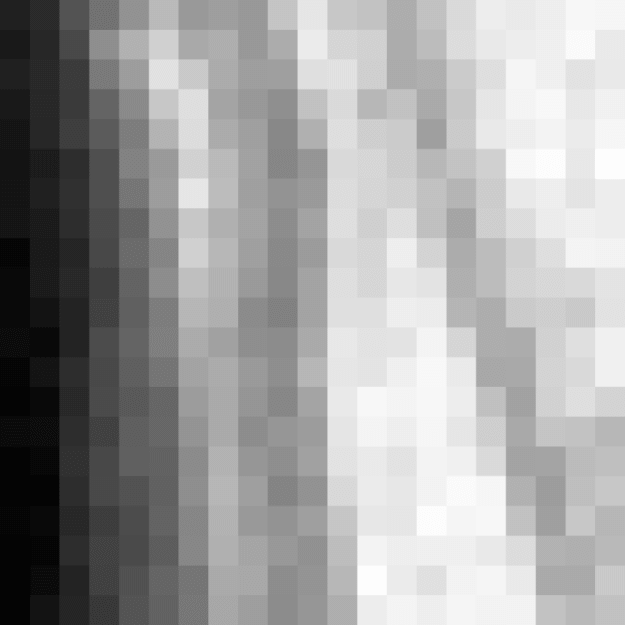}\\
\includegraphics[width=0.10\linewidth, angle=180]{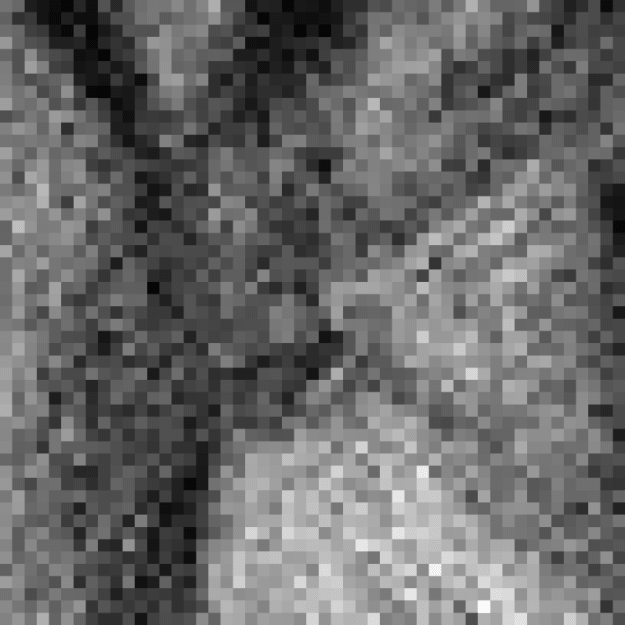}
\includegraphics[width=0.10\linewidth, angle=180]{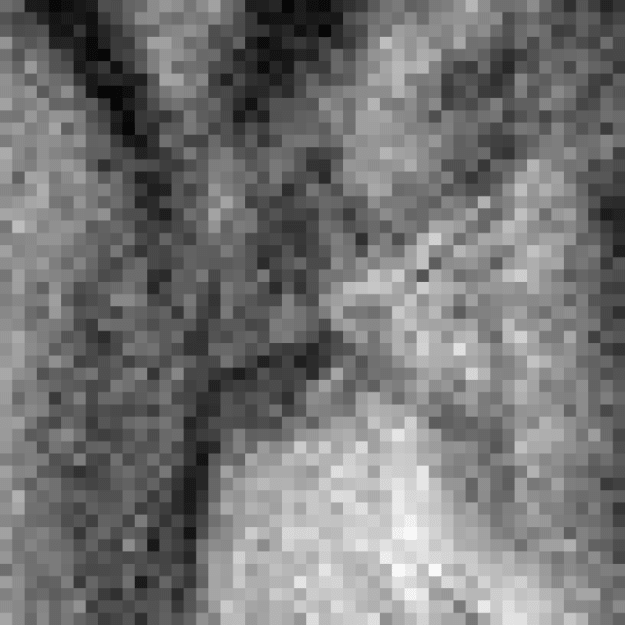}
\includegraphics[width=0.10\linewidth, angle=180]{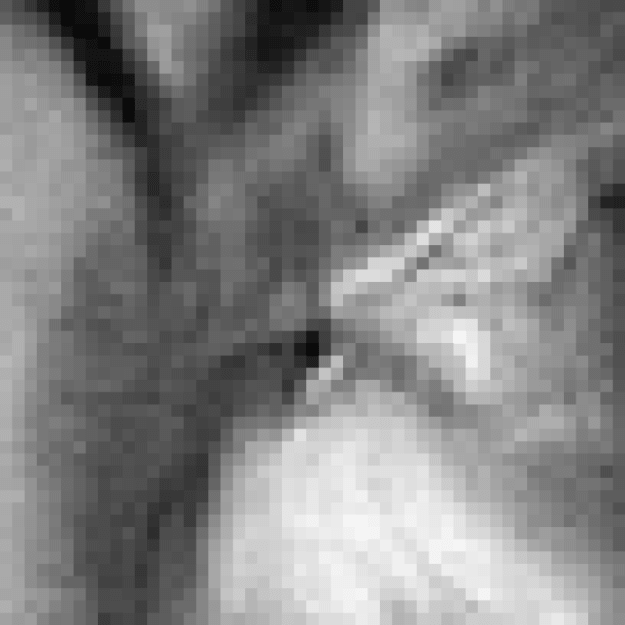}
\includegraphics[width=0.10\linewidth, angle=180]{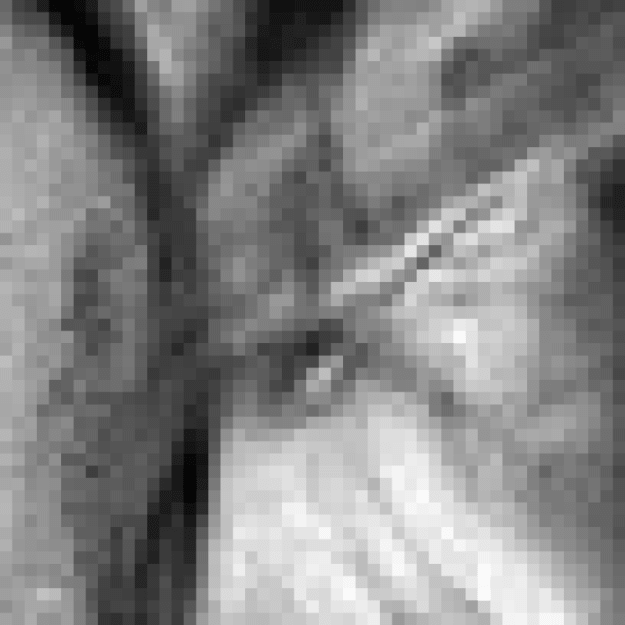}
\includegraphics[width=0.10\linewidth, angle=180]{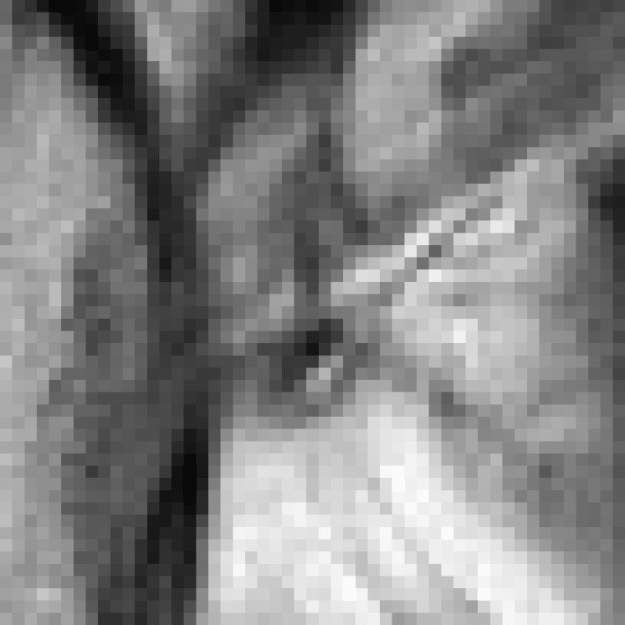}
\includegraphics[width=0.10\linewidth, angle=180]{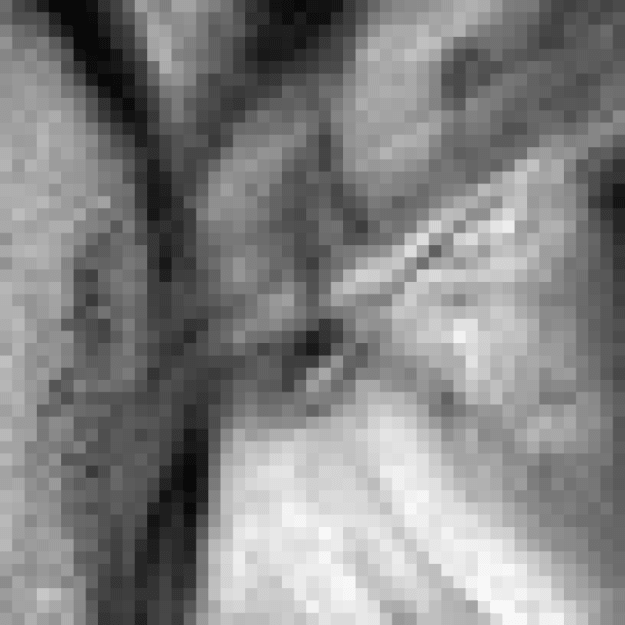}
\includegraphics[width=0.10\linewidth, angle=180]{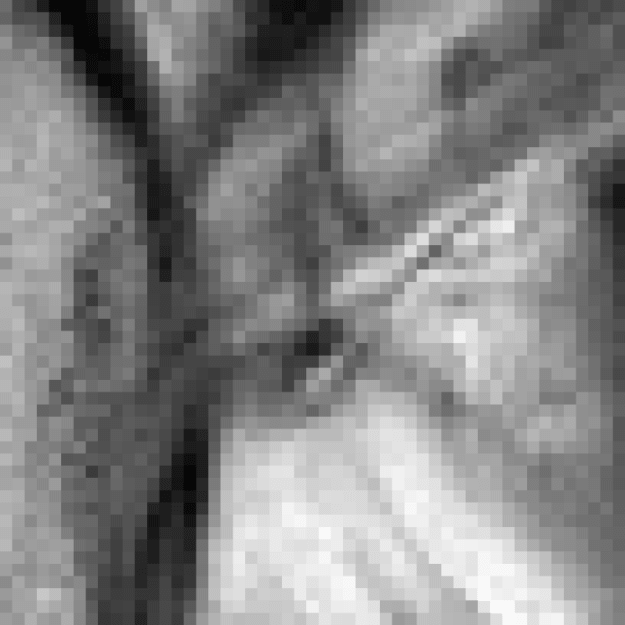}
\includegraphics[width=0.10\linewidth, angle=180]{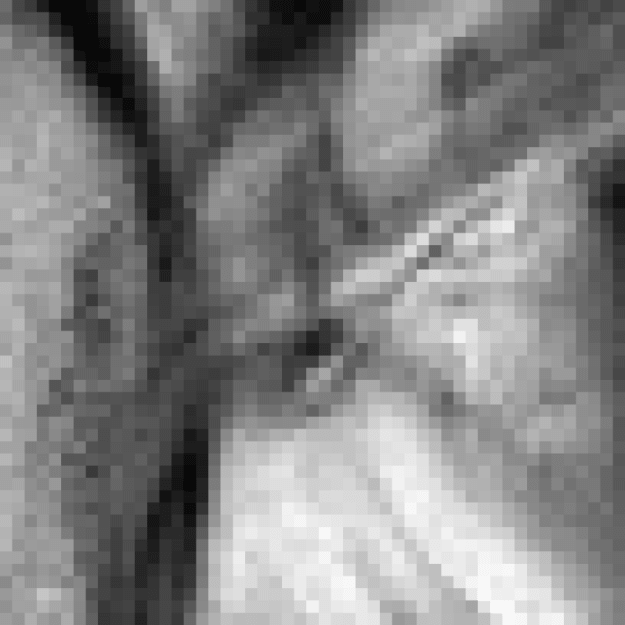}
\includegraphics[width=0.10\linewidth, angle=180]{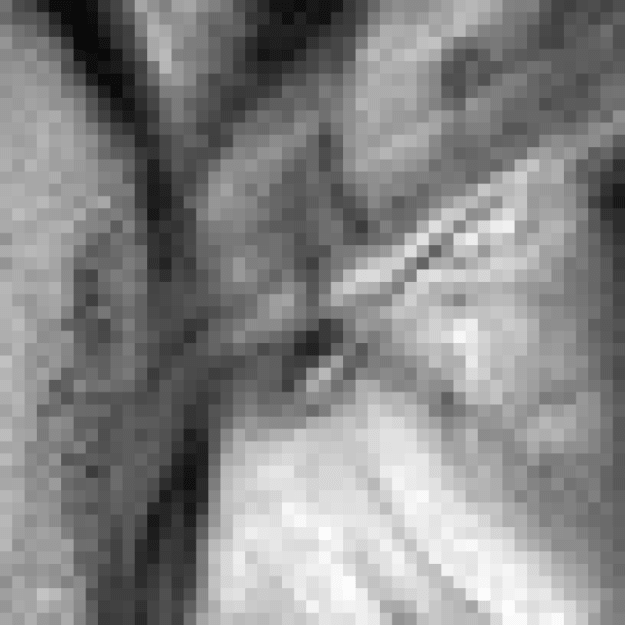}\\
\includegraphics[width=0.10\linewidth, angle=180]{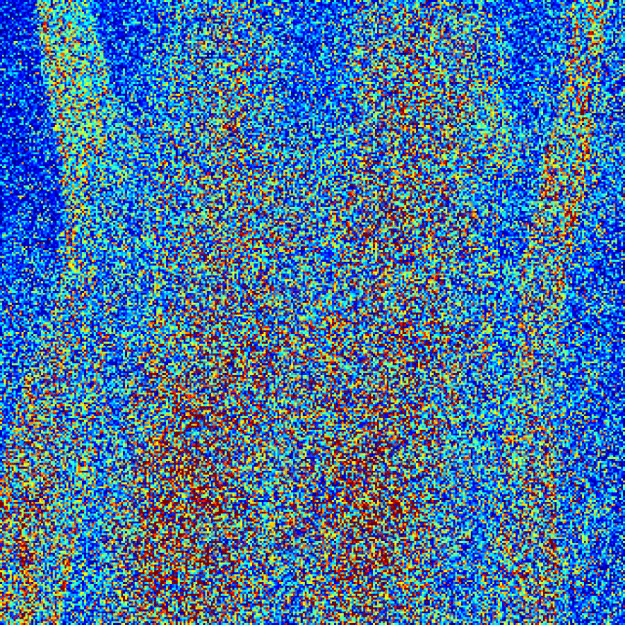}
\includegraphics[width=0.10\linewidth, angle=180]{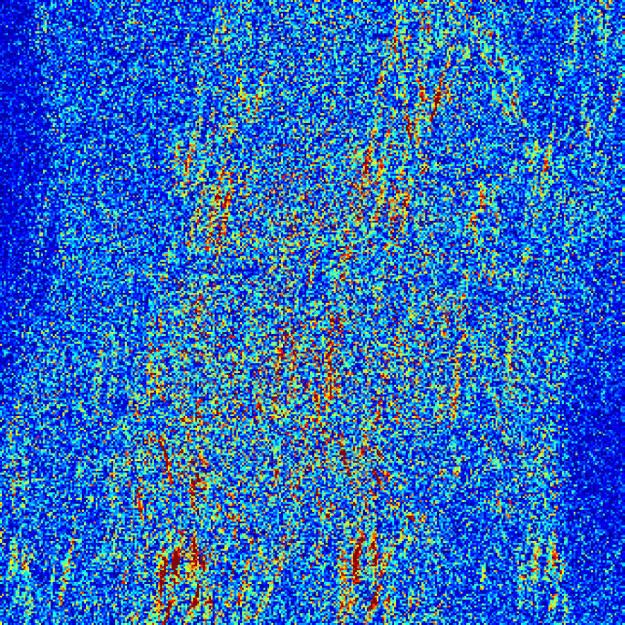}
\includegraphics[width=0.10\linewidth, angle=180]{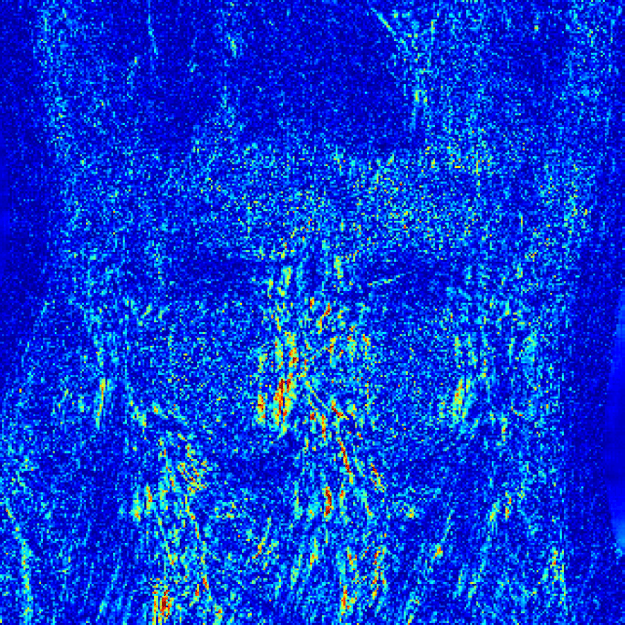}
\includegraphics[width=0.10\linewidth, angle=180]{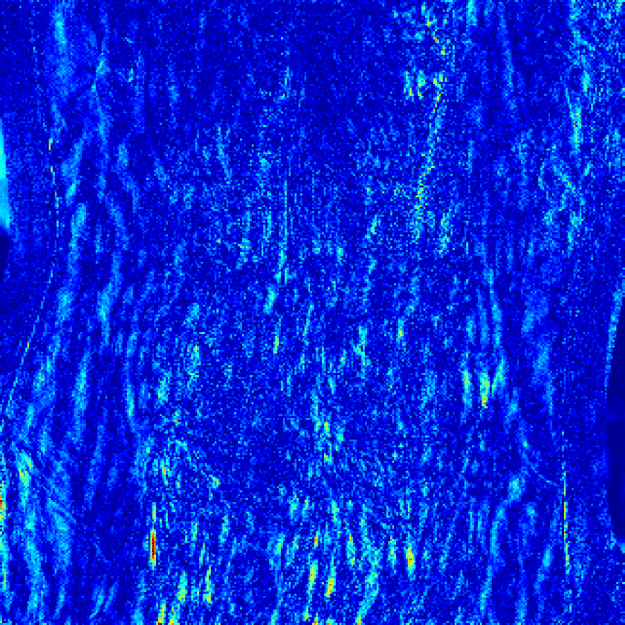}
\includegraphics[width=0.10\linewidth, angle=180]{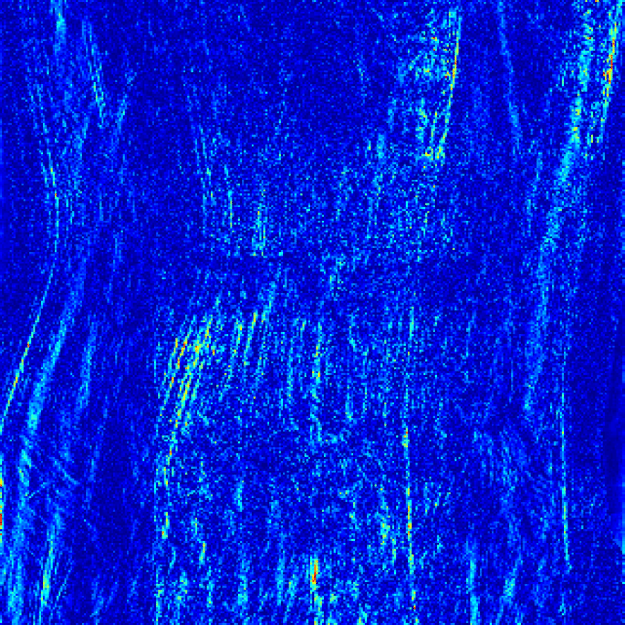}
\includegraphics[width=0.10\linewidth, angle=180]{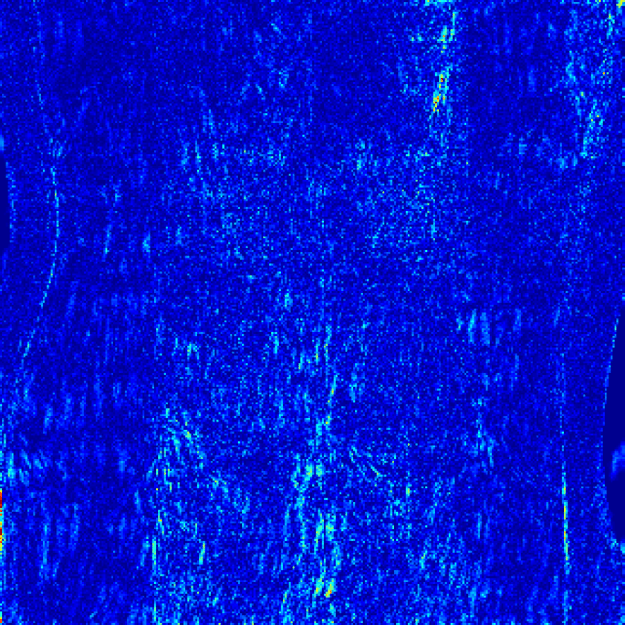}
\includegraphics[width=0.10\linewidth, angle=180]{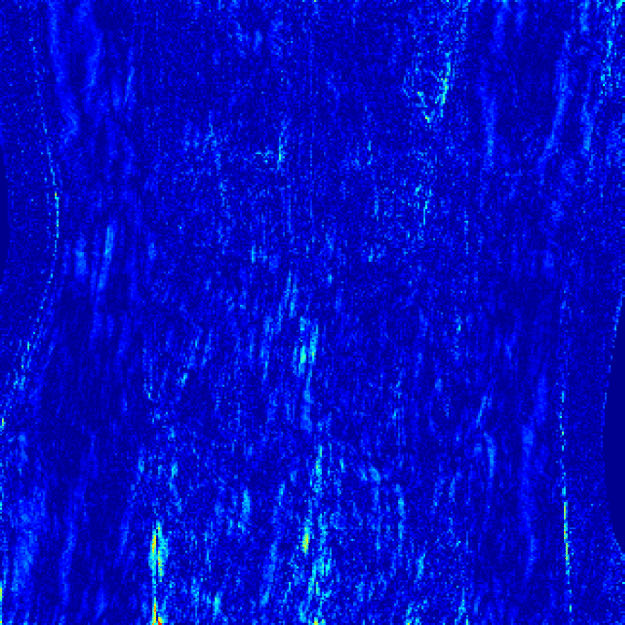}
\includegraphics[width=0.10\linewidth, angle=180]{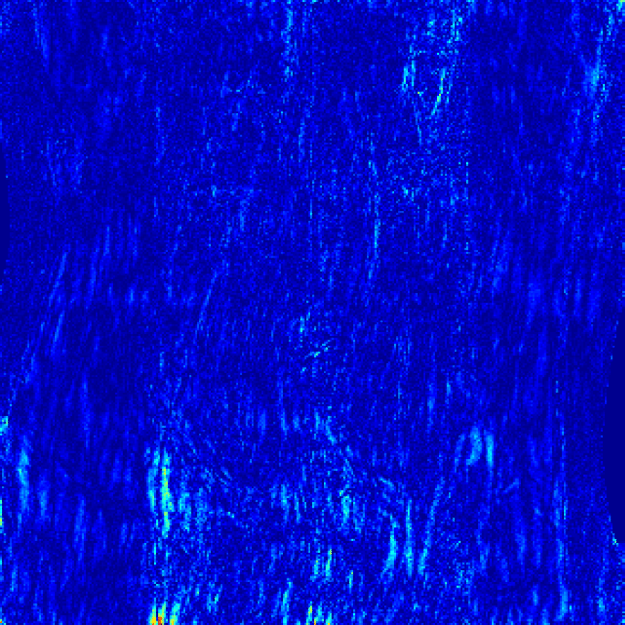}
\includegraphics[width=0.10\linewidth, angle=180]{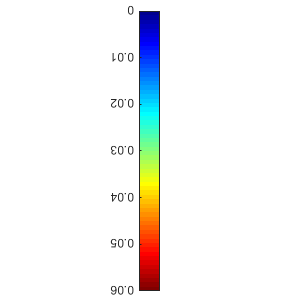}\\
\includegraphics[width=0.10\linewidth]{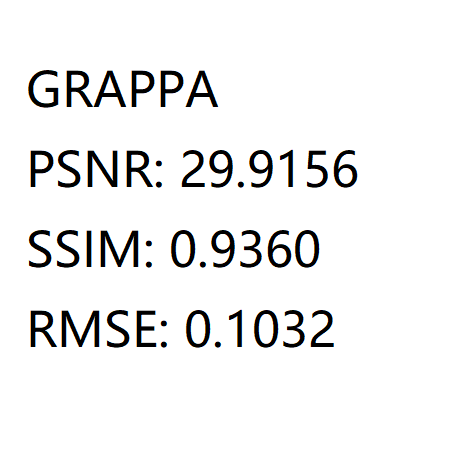}
\includegraphics[width=0.10\linewidth]{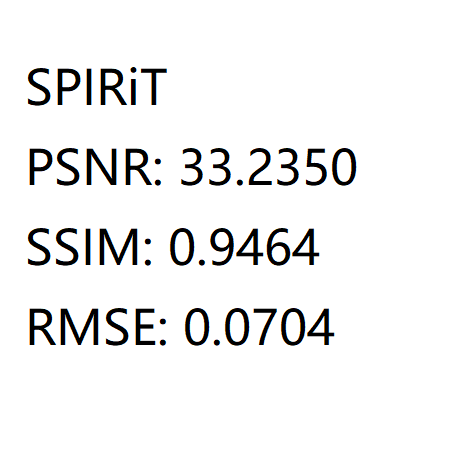}
\includegraphics[width=0.10\linewidth]{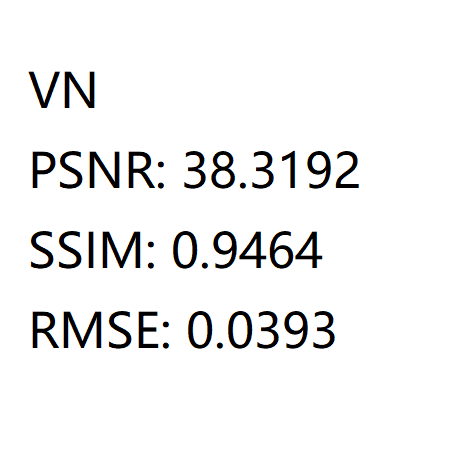}
\includegraphics[width=0.10\linewidth]{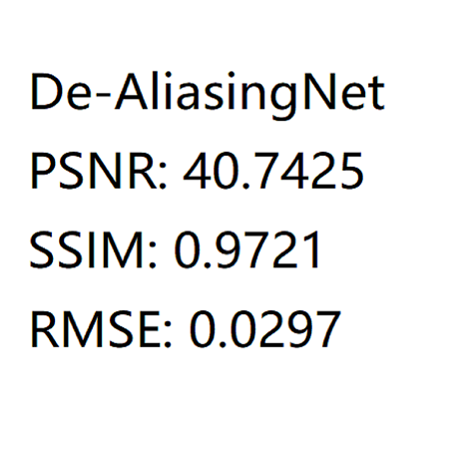}
\includegraphics[width=0.10\linewidth]{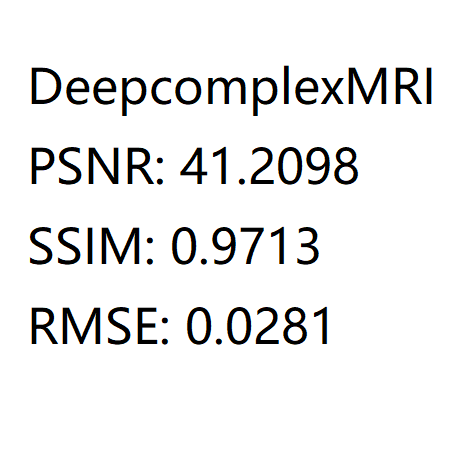}
\includegraphics[width=0.10\linewidth]{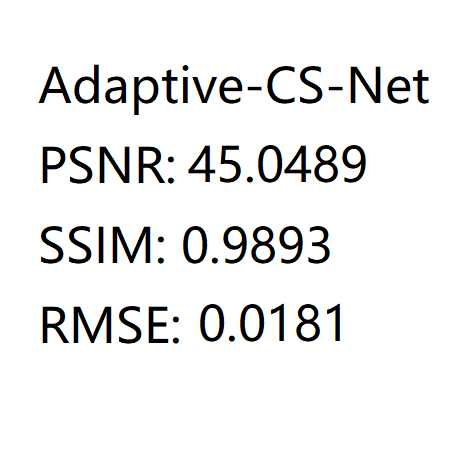}
\includegraphics[width=0.10\linewidth]{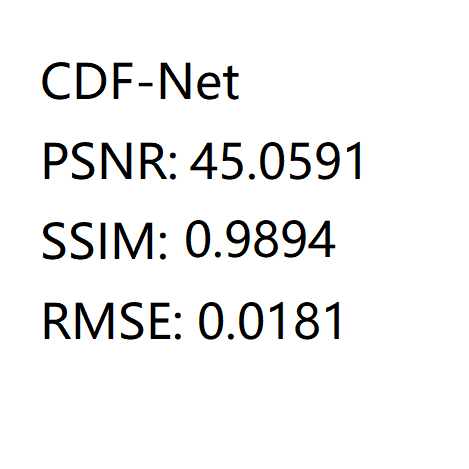}
\includegraphics[width=0.10\linewidth]{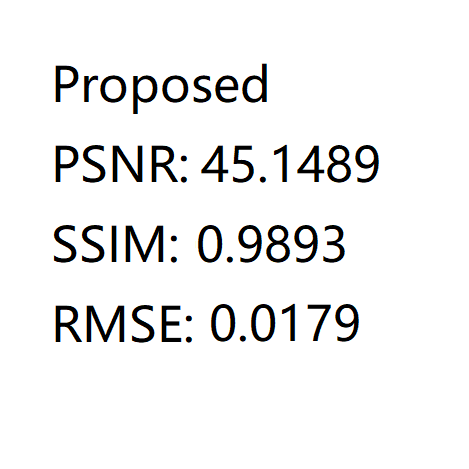}
\includegraphics[width=0.10\linewidth]{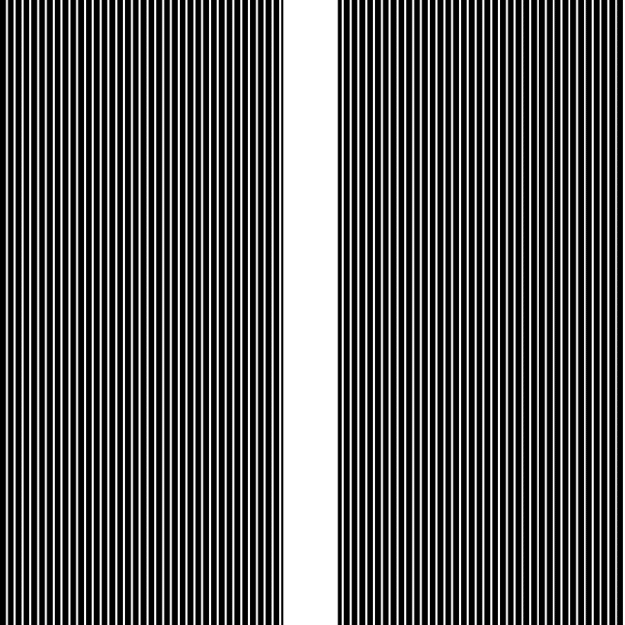}
\caption{Qualitative comparison results of reconstruction methods on the Coronal PD knee image. The top row shows reconstructed images and the referenced image, the second and third-row are corresponding zoomed-in ROIs of the red box area (draw on the rightmost ground truth image), the fourth row shows corresponding pointwise absolute error maps and the last row shows corresponding values and regular Cartesian sampling (31.56\% rate) mask.} 
\label{PD}
\end{figure*}
\begin{figure*}
\centering
\includegraphics[width=0.075\linewidth, angle=180]{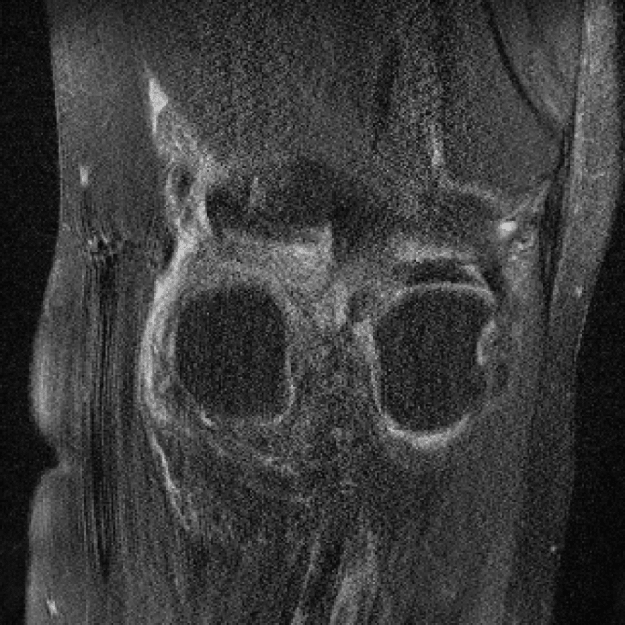}
\includegraphics[width=0.075\linewidth, angle=180]{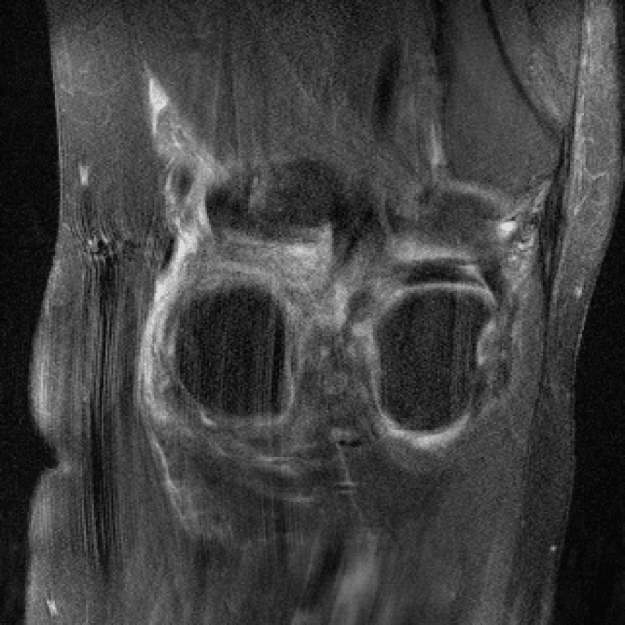}
\includegraphics[width=0.075\linewidth, angle=180]{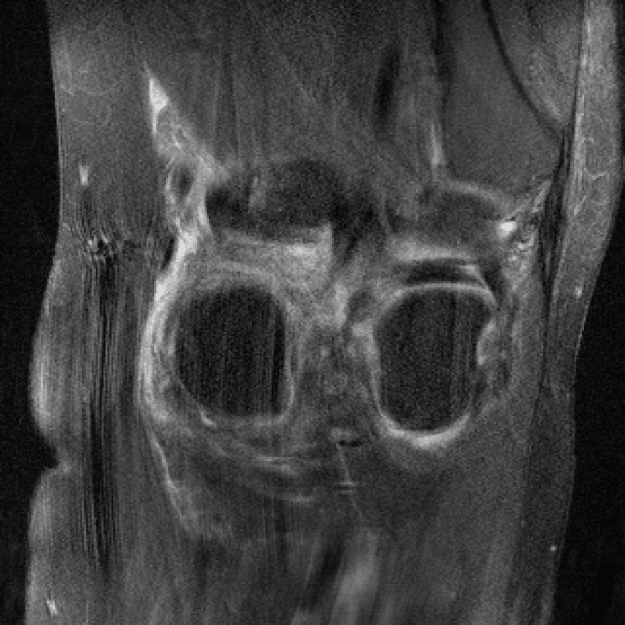}
\includegraphics[width=0.075\linewidth, angle=180]{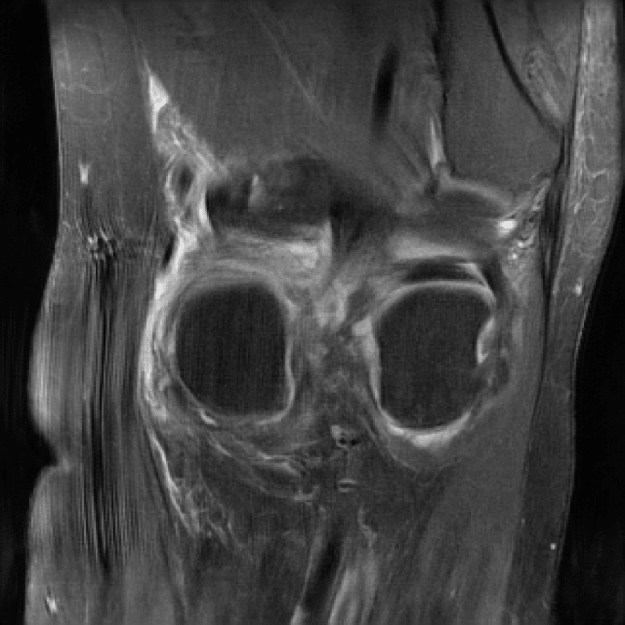}
\includegraphics[width=0.075\linewidth, angle=180]{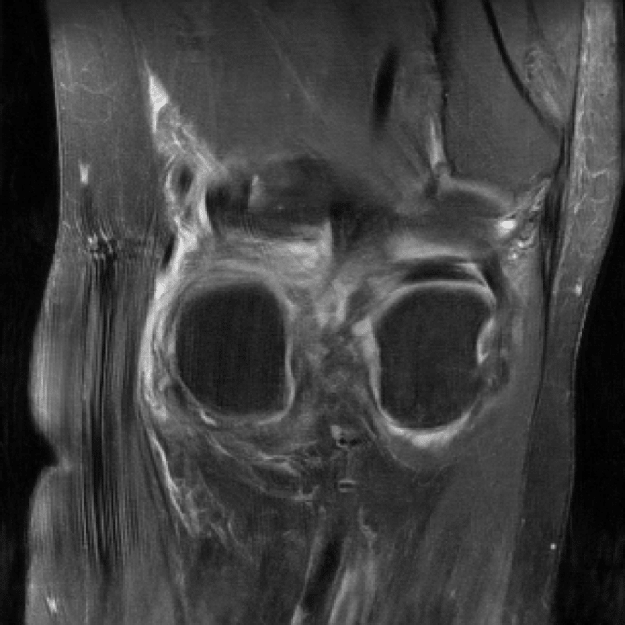}
\includegraphics[width=0.075\linewidth, angle=180]{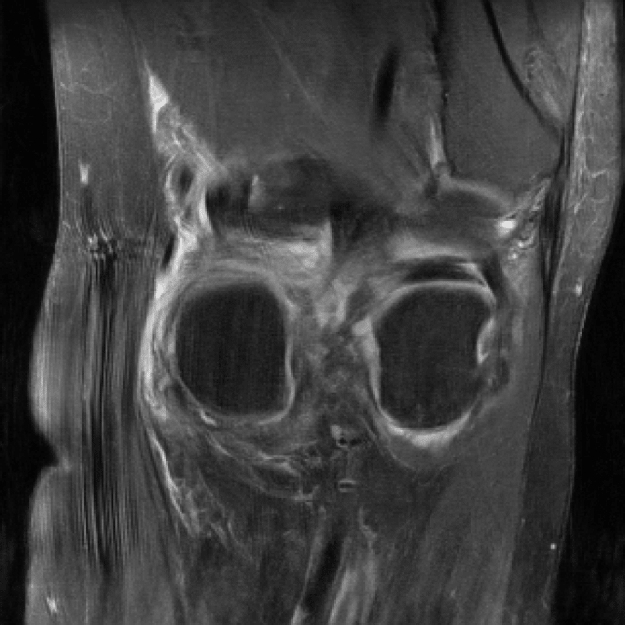}
\includegraphics[width=0.075\linewidth, angle=180]{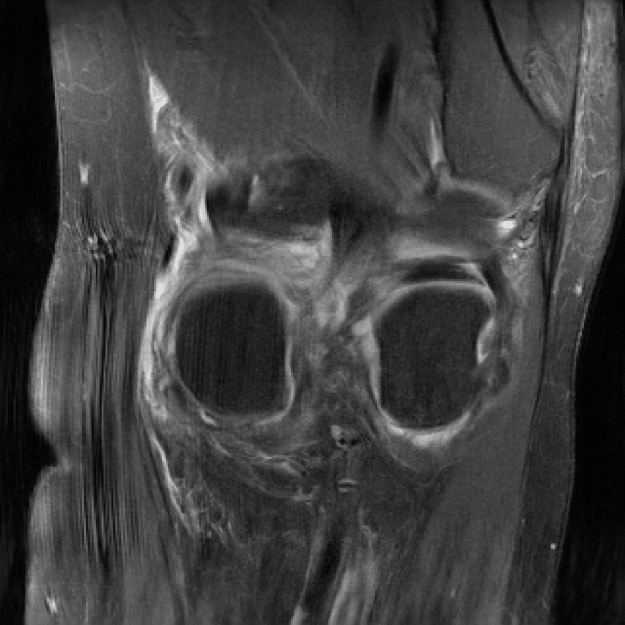}
\includegraphics[width=0.075\linewidth, angle=180]{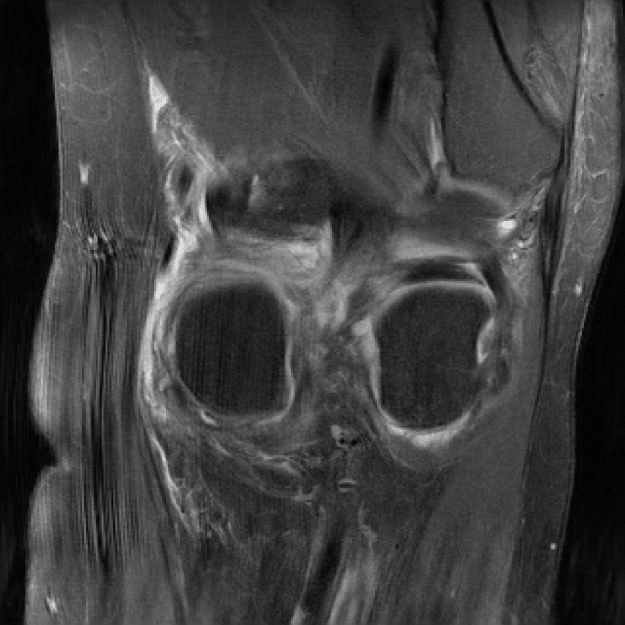}
\includegraphics[width=0.075\linewidth, angle=180]{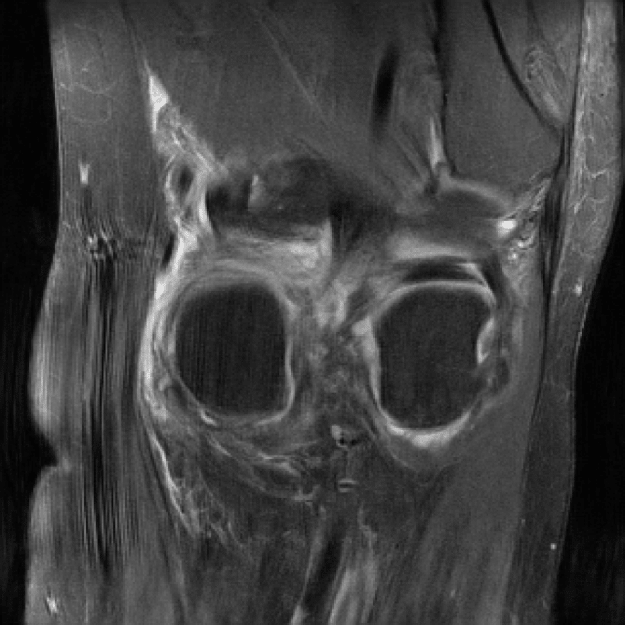}
\includegraphics[width=0.075\linewidth, angle=180]{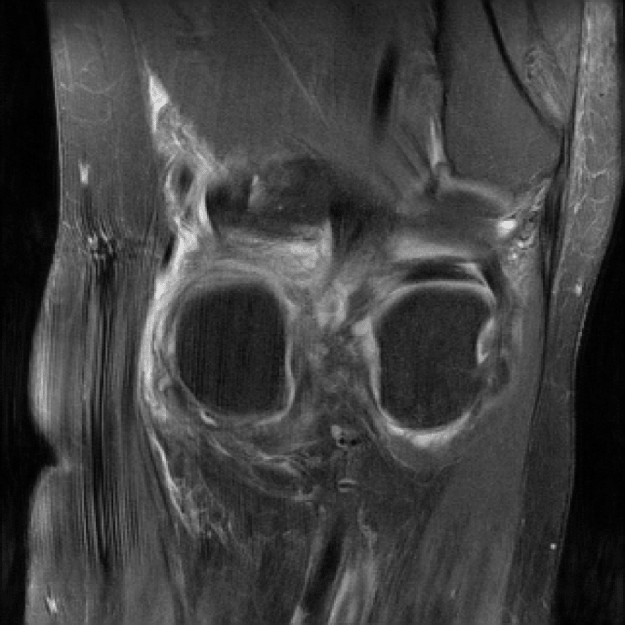}
\includegraphics[width=0.075\linewidth, angle=180]{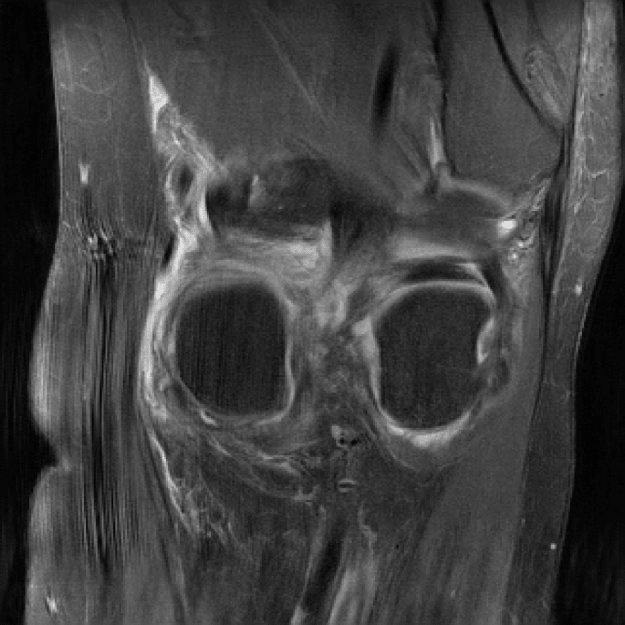}
\includegraphics[width=0.075\linewidth, angle=180]{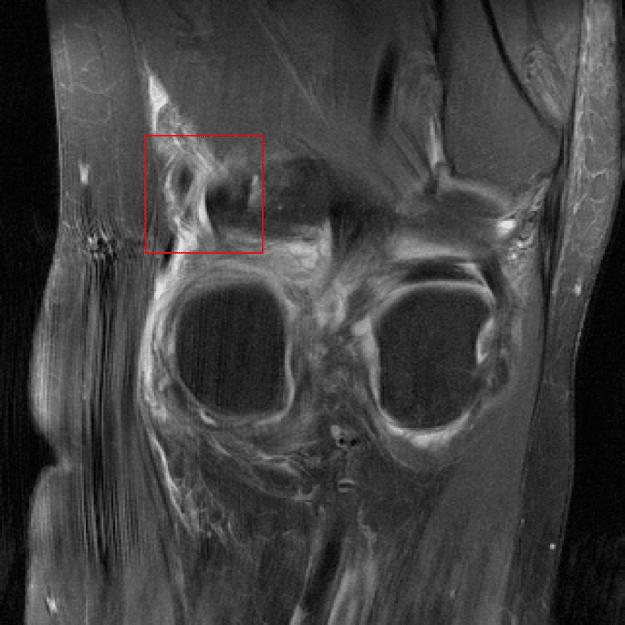}\\
\includegraphics[width=0.075\linewidth, angle=180]{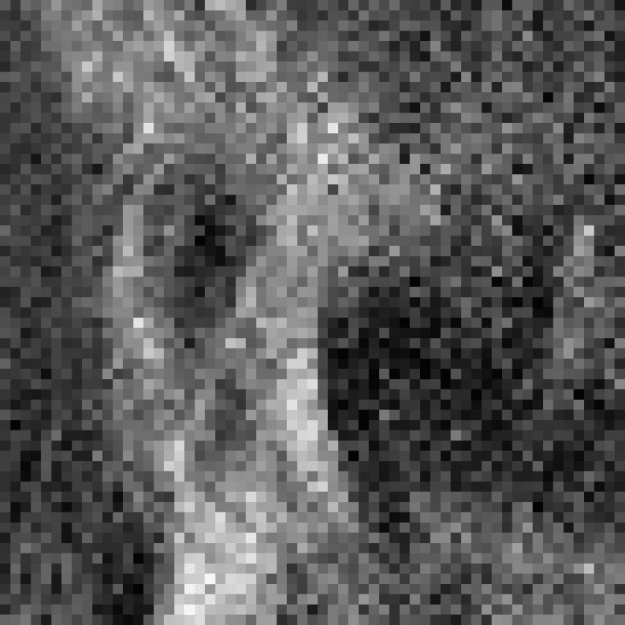}
\includegraphics[width=0.075\linewidth, angle=180]{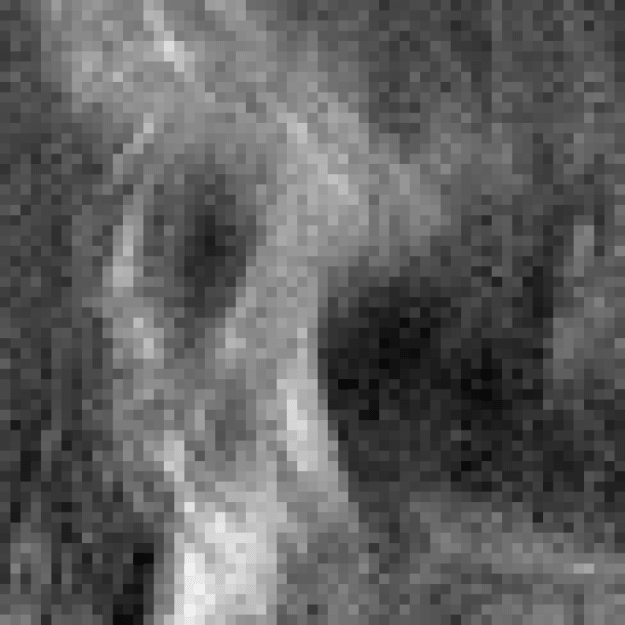}
\includegraphics[width=0.075\linewidth, angle=180]{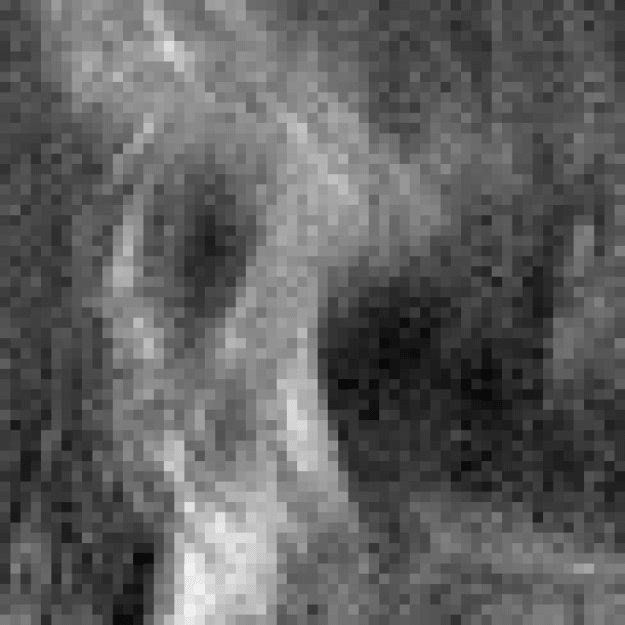}
\includegraphics[width=0.075\linewidth, angle=180]{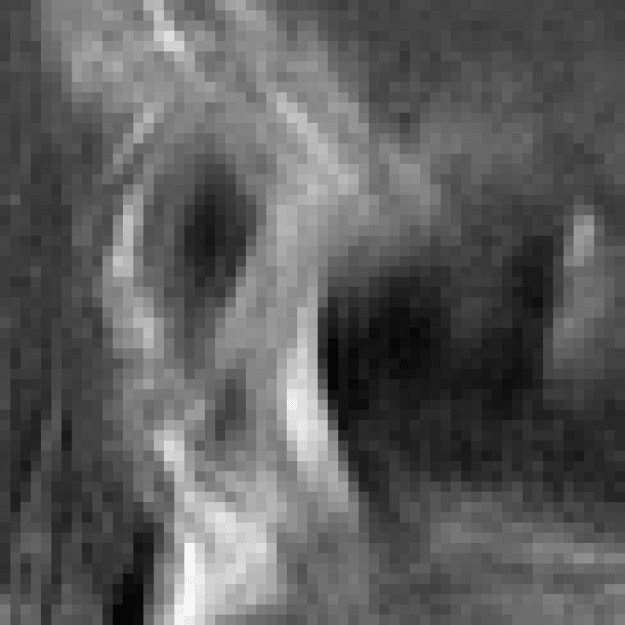}
\includegraphics[width=0.075\linewidth, angle=180]{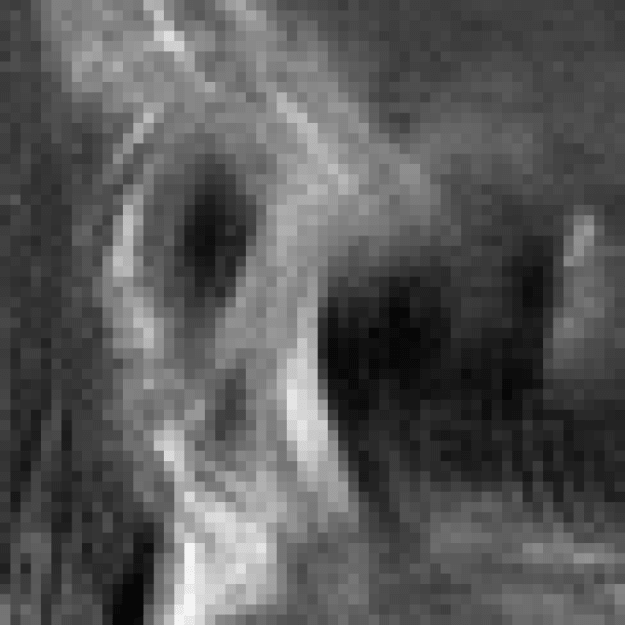}
\includegraphics[width=0.075\linewidth, angle=180]{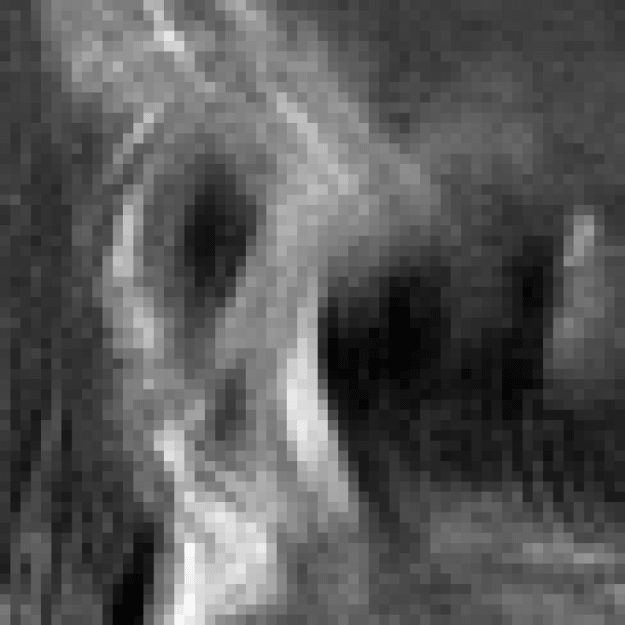}
\includegraphics[width=0.075\linewidth, angle=180]{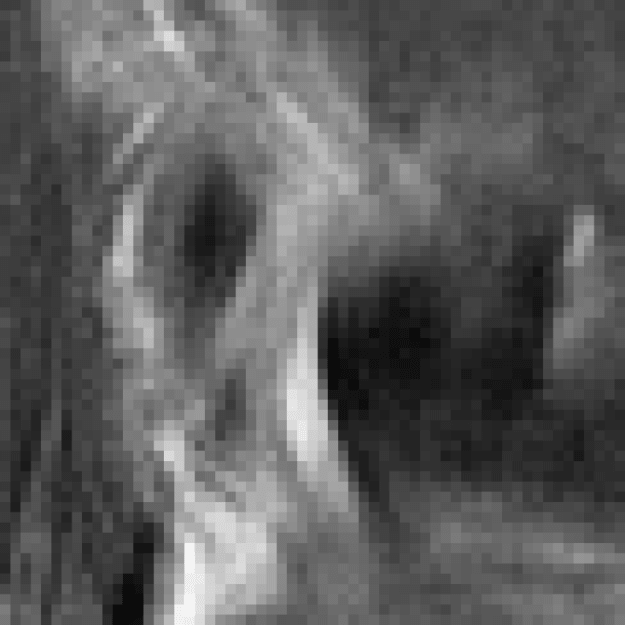}
\includegraphics[width=0.075\linewidth, angle=180]{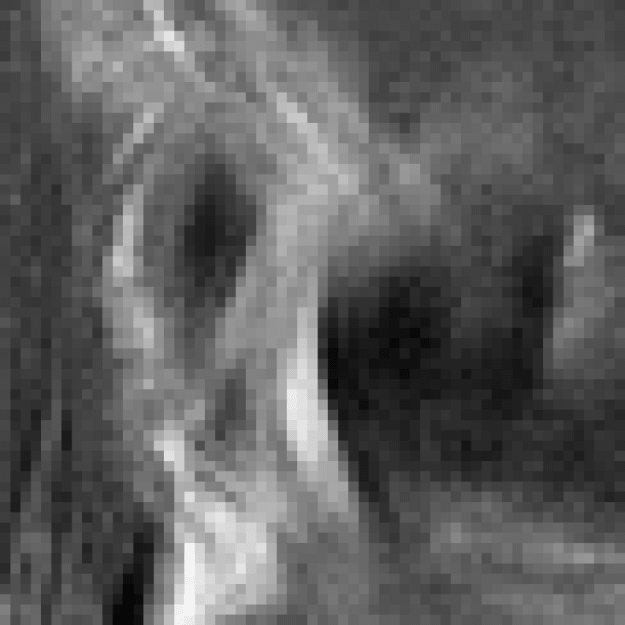}
\includegraphics[width=0.075\linewidth, angle=180]{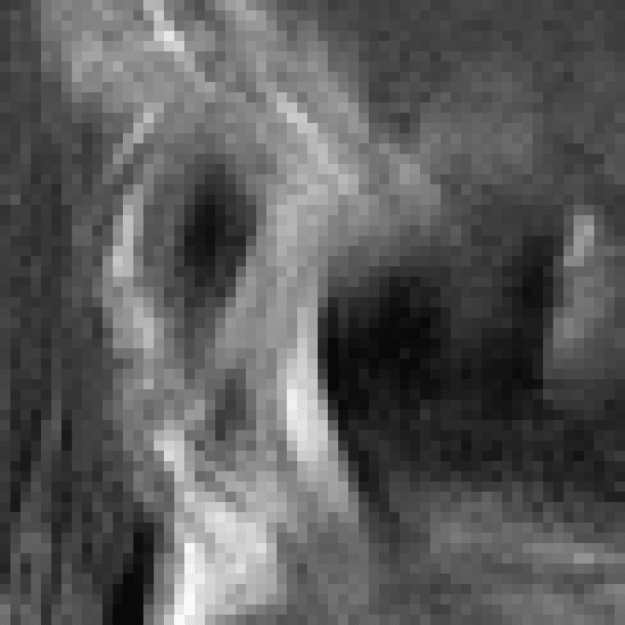}
\includegraphics[width=0.075\linewidth, angle=180]{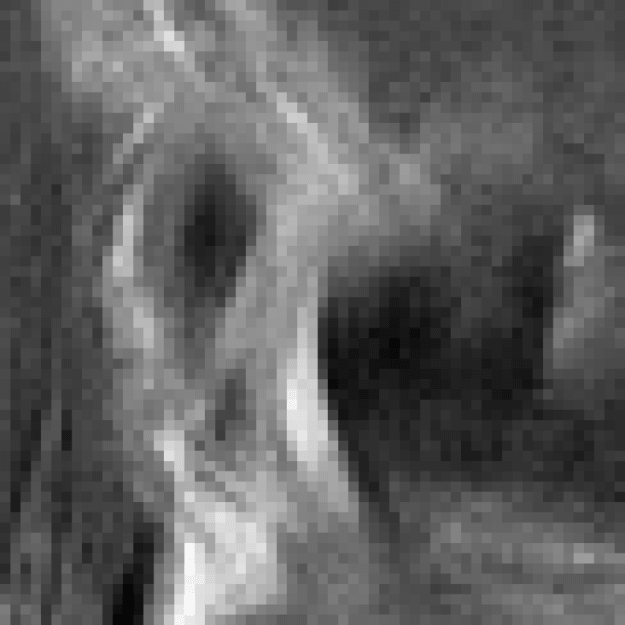}
\includegraphics[width=0.075\linewidth, angle=180]{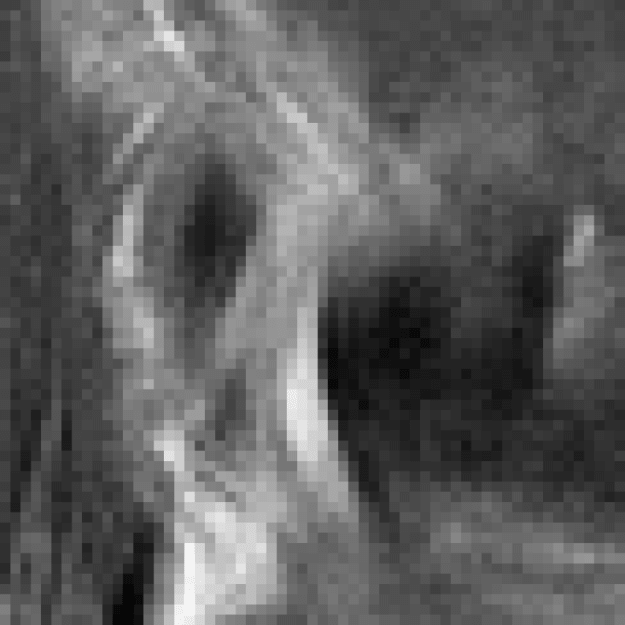}
\includegraphics[width=0.075\linewidth, angle=180]{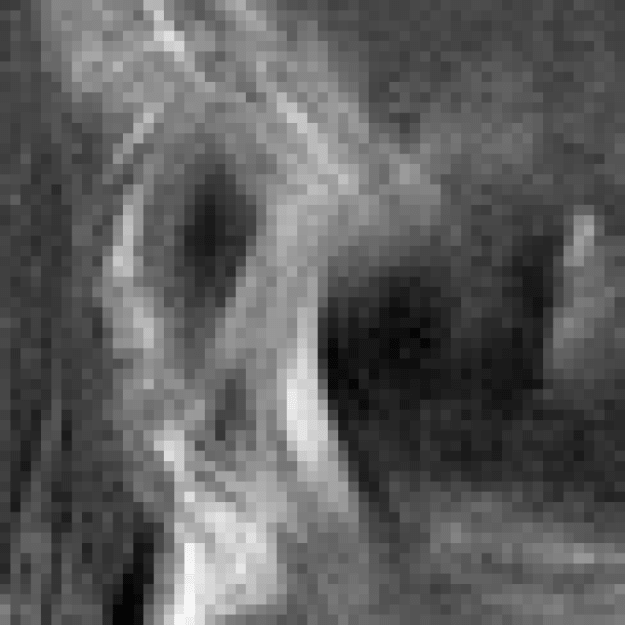}\\
\includegraphics[width=0.075\linewidth, angle=180]{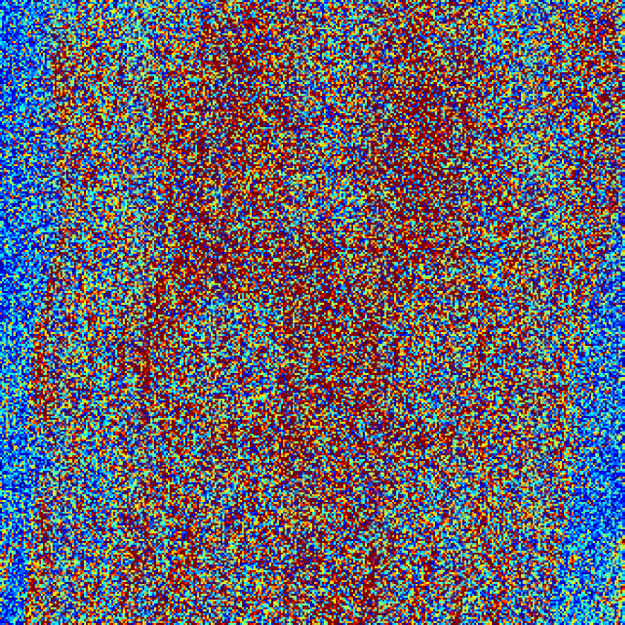}
\includegraphics[width=0.075\linewidth, angle=180]{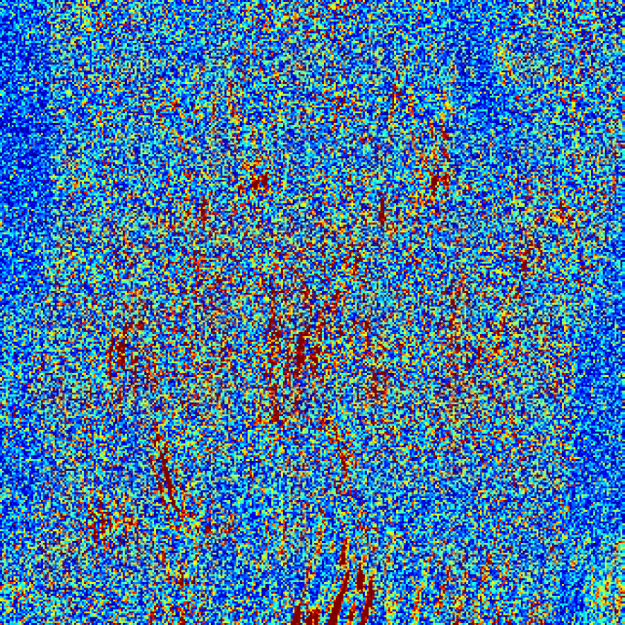}
\includegraphics[width=0.075\linewidth, angle=180]{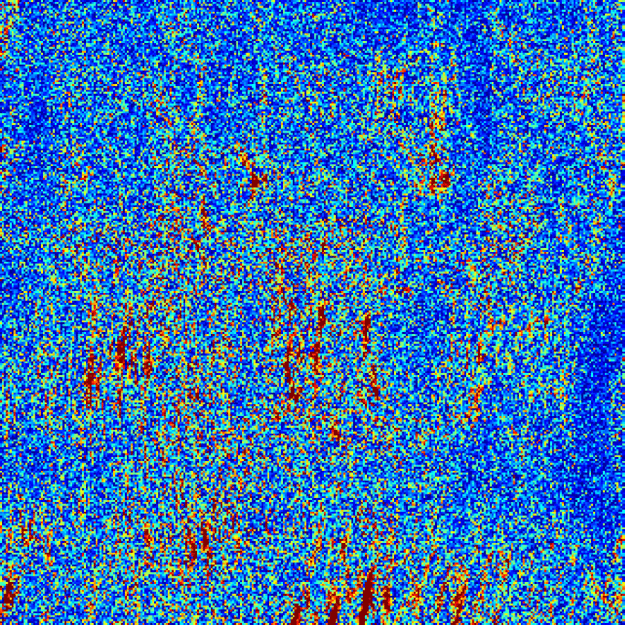}
\includegraphics[width=0.075\linewidth, angle=180]{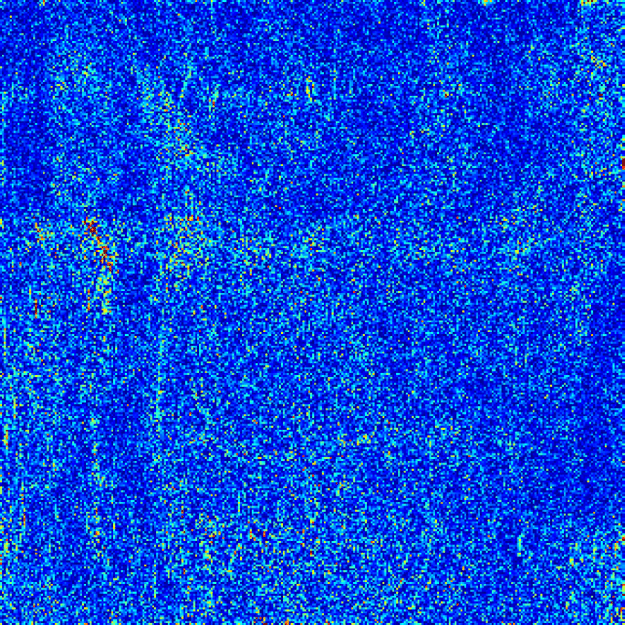}
\includegraphics[width=0.075\linewidth, angle=180]{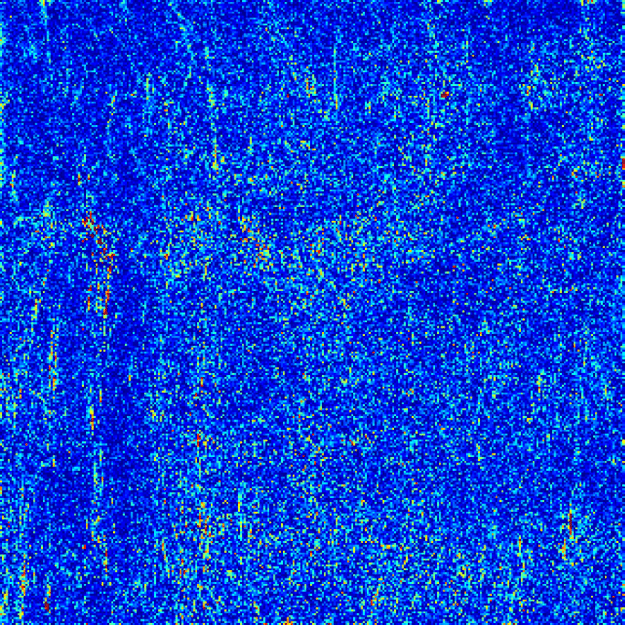}
\includegraphics[width=0.075\linewidth, angle=180]{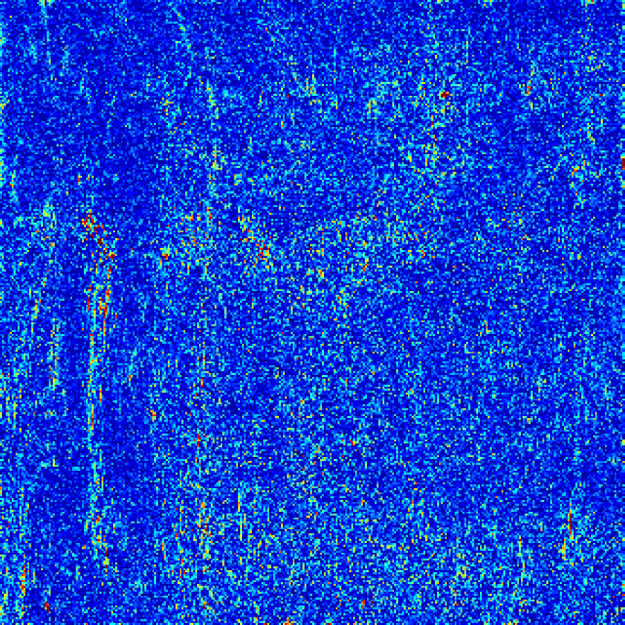}
\includegraphics[width=0.075\linewidth, angle=180]{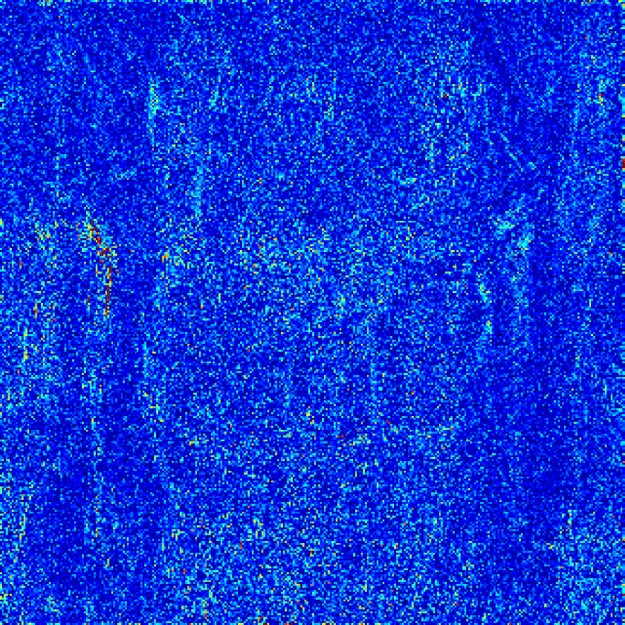}
\includegraphics[width=0.075\linewidth, angle=180]{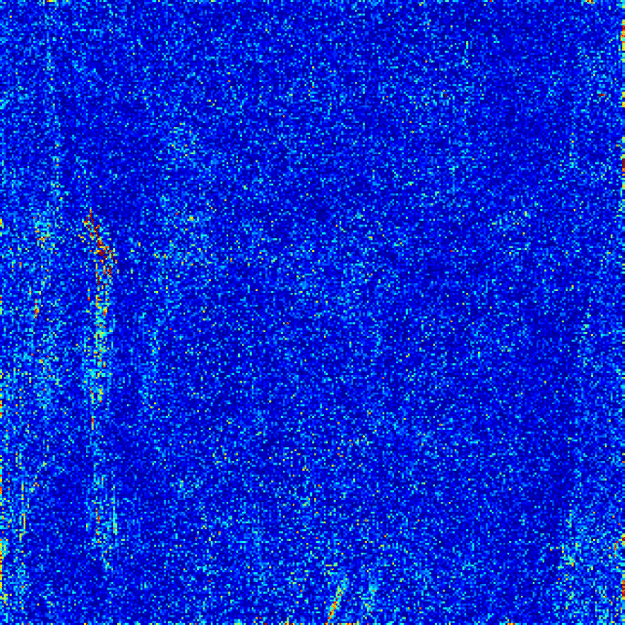}
\includegraphics[width=0.075\linewidth, angle=180]{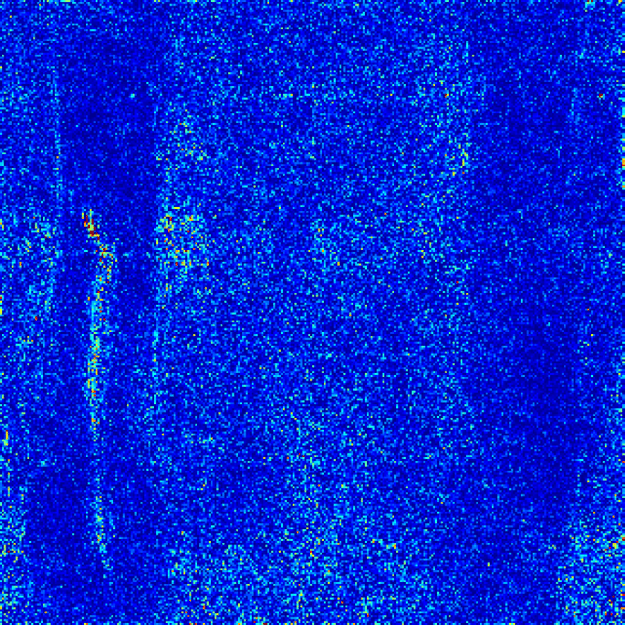}
\includegraphics[width=0.075\linewidth, angle=180]{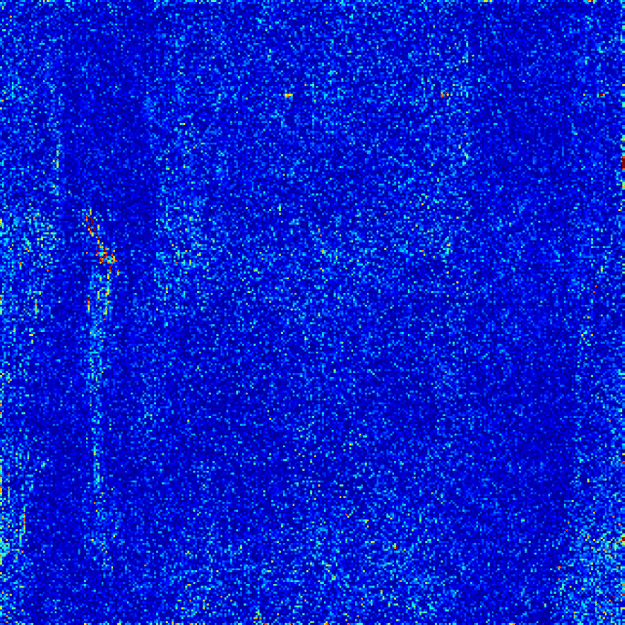}
\includegraphics[width=0.075\linewidth, angle=180]{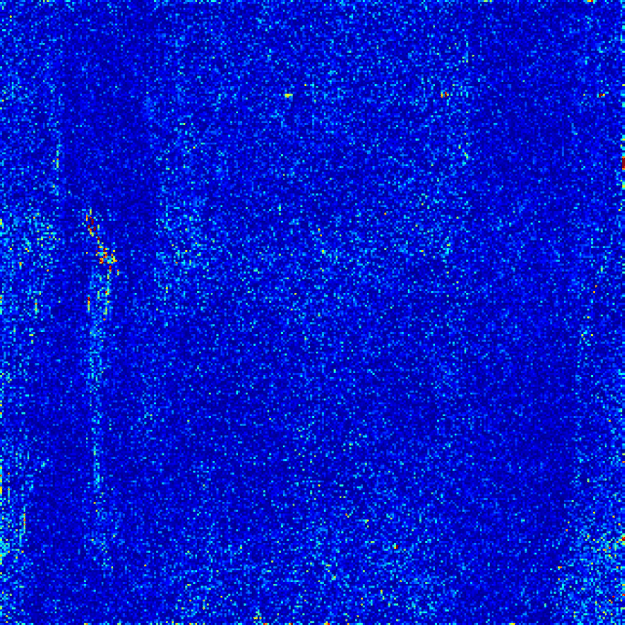}
\includegraphics[width=0.075\linewidth, angle=180]{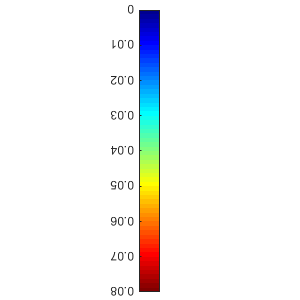}\\
\includegraphics[width=0.075\linewidth]{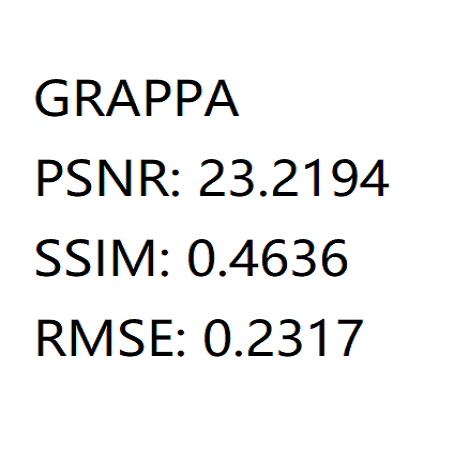}
\includegraphics[width=0.075\linewidth]{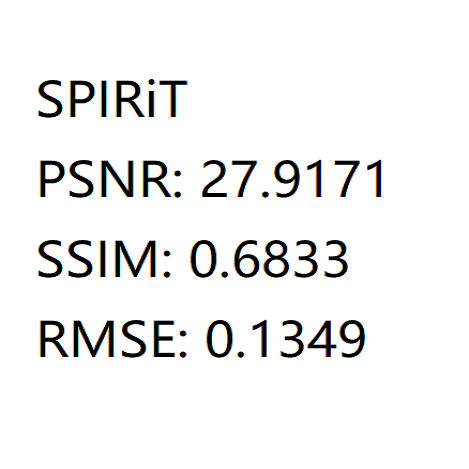}
\includegraphics[width=0.075\linewidth]{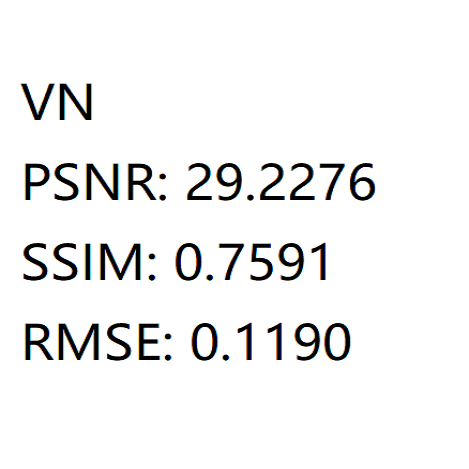}
\includegraphics[width=0.075\linewidth]{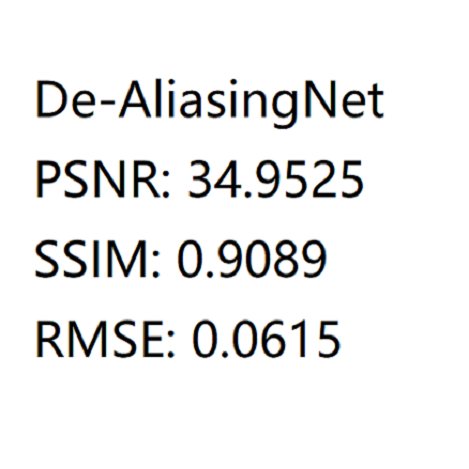}
\includegraphics[width=0.075\linewidth]{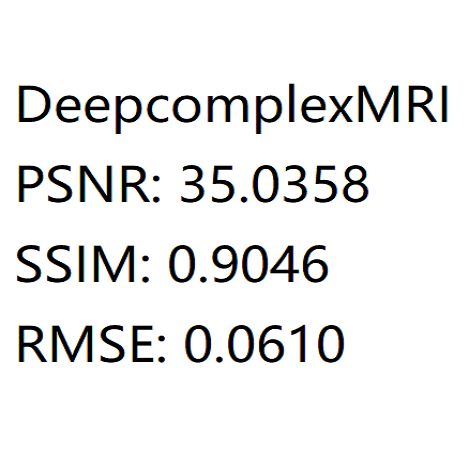}
\includegraphics[width=0.075\linewidth]{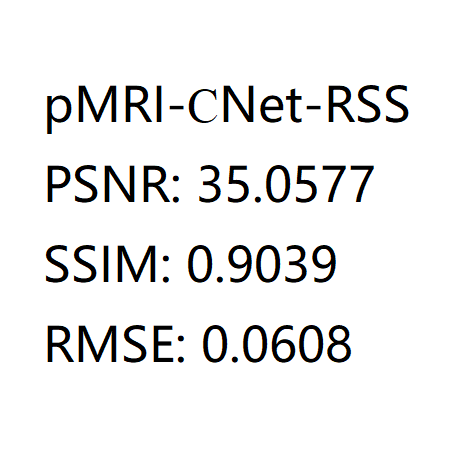}
\includegraphics[width=0.075\linewidth]{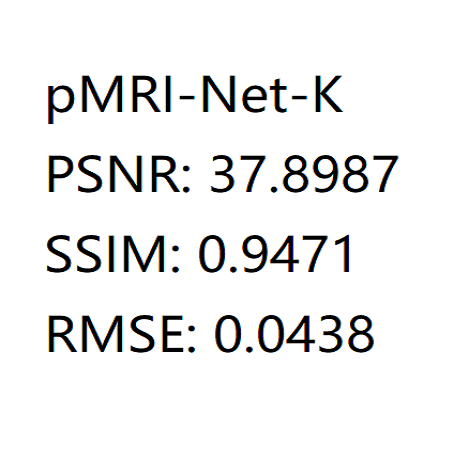}
\includegraphics[width=0.075\linewidth]{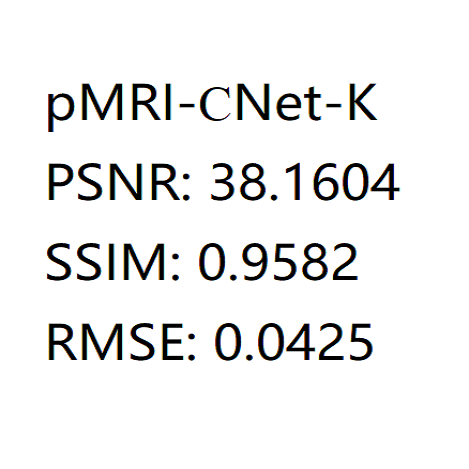}
\includegraphics[width=0.075\linewidth]{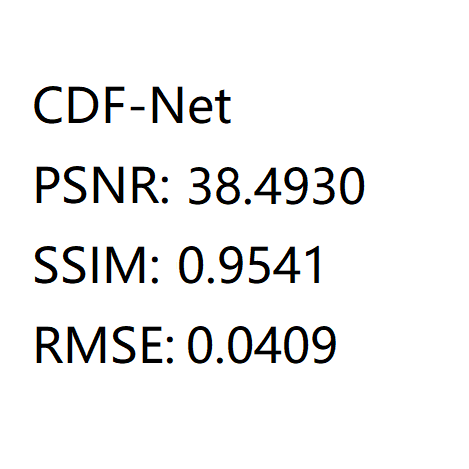}
\includegraphics[width=0.075\linewidth]{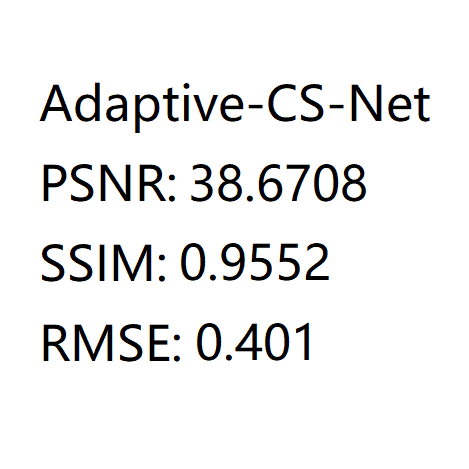}
\includegraphics[width=0.075\linewidth]{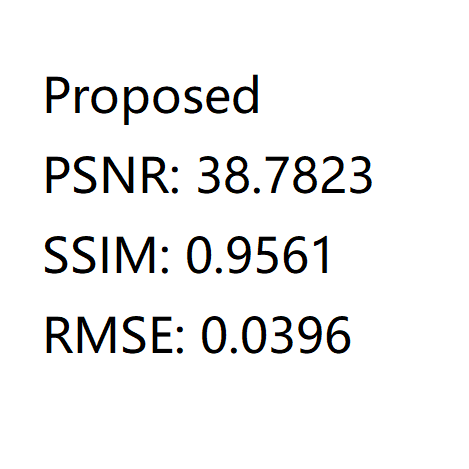}
\includegraphics[width=0.075\linewidth]{mask.png}
\caption{Qualitative comparison results of reconstruction methods on the Coronal FSPD knee image. From top to bottom: reconstructed images of different methods and the reference image, corresponding zoomed-in ROIs, corresponding pointwise absolute error maps and color bar, corresponding evaluation metric values, and regular Cartesian sampling (31.56\% rate) mask. }
\label{FSPD}
\end{figure*}

\begin{table}
\caption{Mean square error between $\ubf$ and $\ubf^* $ for PD data comparing to the methods that reconstruct multi-coil images. }
\centering
 \begin{tabular}{ccc}
\toprule
Method   & RMSE between $\ubf$ and $  \ubf^*$   \\
\midrule
De-AliasingNet~\cite{10.1007/978-3-030-32248-9_4} & 0.1573\\
DeepcomplexMRI~\cite{WANG2020136} & 0.1253\\
CDF-Net~\cite{10.1007/978-3-030-59713-9_41} & 0.1196\\
Adaptive-CS-Net~\cite{pezzotti2020adaptive}  & 0.1096\\
pMRI-$\C$Net-K~\eqref{eq:loss_ablation}&  0.0538\\
Proposed~\eqref{eq:loss}   & 0.0505 \\
 \bottomrule
\end{tabular}
\label{tab:ui_mse}
\end{table}
\subsection{Ablation Studies}\label{subsec:Ablation}
The experiments introduced in this section implemented only in the image domain, in which \eqref{eq:scheme} is replaced by:
\begin{subequations}\label{eq:bu_ablation}
\begin{align}
\bbf_i{(t)} & =  \ubf_i{(t)} - \rho_t \Fbf^{H}  \Pbf^{\top} (\Pbf \Fbf \ubf_i{(t)} - \fbf_i), \label{eq:b_ablation}  \\
\ubf_i{(t+1)} & = \bbf_i{(t)} + \tilde{\J} \circ \tilde{\G} \circ \soft_{\alpha_t} ( \G \circ \J (\bbf_i{(t)})). \label{eq:u_ablation}
\end{align}
\end{subequations}
$ \soft_{\alpha_t}$ represents soft shrinkage operator, we set $ \alpha_0 = 0 $ for both real and imaginary parts. Details for this model was explained in the conference paper \cite{10.1007/978-3-030-61598-7_2}. The proximal operator is learned in residual update \eqref{eq:u_ablation} and only performs in the image domain.

We conduct a series of experiments to test the effects of several important components in the network proposed. Table \ref{tab:result_our} displays the comparison of proposed ablation study. All the experiments in Table \ref{tab:result_our} were implemented by \eqref{eq:bu_ablation} except for the last one ``proposed'' method which unrolled \eqref{eq:scheme}.
\begin{table}
    \centering
    \caption{Labels of the variations of the proposed pMRI reconstruction network.} 
    \begin{tabular}{ll}
    \toprule
    Label     &  Meaning \\
    \midrule
    -Net & Real-valued convolution/activation \\
    -$\C$Net & Complex-valued convolution/activation \\
    -RSS & Using root of sum of squares in place of $\J$\\
    -ZF & Zero-filling as initial $\ubf{(0)}$  \\
    -SP & SPIRiT as initial $\ubf{(0)}$ \\
    -K & $\Fbf^{H}(\fbf + \K(\fbf))$ as initial $\ubf{(0)}$ \\
    \bottomrule
    \end{tabular}
    \label{tab:labels}
\end{table}
To specify which of the components are employed, we append the labels shown in Table \ref{tab:labels} with different variations of the ablated pMRI networks.  For instance, the network with real-valued convolution and activation function with zero-filled initialization is denoted by pMRI-Net-ZF, etc.
\begin{table*}
\caption{Quantitative evaluations of the reconstructions on the Coronal FSPD \& PD data using different variations of the proposed methods (Labels of variations are explained in Tabel \ref{tab:labels}. The experiments without being labeled loss functions are trained with \eqref{eq:loss_body_ablation}).}
\begin{center}
\resizebox{\textwidth}{20mm}{ 
\begin{tabular}{l|ccc|ccc|cc}
\toprule
&& FSPD data & & & PD data \\
Method   & PSNR      & SSIM      & RMSE       & PSNR     & SSIM     & RMSE     & Phases T   & Parameters \\\midrule
pMRI-$\C$Net-RSS & 36.4887$\pm$0.9787 & 0.9002$\pm$0.0197 &  0.0661$\pm$0.0051  & 41.2897$\pm$0.8430  & 0.9281$\pm$0.0357 & 0.0285$\pm$0.0037 & 5 & 5.03 M\\
pMRI-Net-ZF & 37.8475$\pm$1.2086 & 0.9212$\pm$0.0236 & 0.0568$\pm$0.0069 & 42.4333$\pm$0.8785 & 0.9793$\pm$0.0023 & 0.0249$\pm$0.0024 & 5 & 5.03 M \\
pMRI-Net-SP & 38.0205$\pm$0.8125 & 0.9291$\pm$0.0183 & 0.0555$\pm$0.0057 & 42.7435$\pm$0.4856 & 0.9754$\pm$0.0047 & 0.0239$\pm$0.0019 & 5 & 5.03 M \\
pMRI-$\C$Net-ZF &  38.1157$\pm$1.3776 &  0.9277$\pm$0.0257 & 0.0552$\pm$0.0085 & 42.7859$\pm$1.1285  & 0.9818$\pm$0.0026 & 0.0241$\pm$0.0045 & 5 & 5.03 M\\
pMRI-$\C$Net-SP & 38.3239$\pm$1.1305 & 0.9282$\pm$0.0269  & 0.0539$\pm$0.0075 & 42.8924$\pm$0.9336 & 0.9760$\pm$0.0054 & 0.0237$\pm$0.0034 & 5 & 5.03 M\\
pMRI-Net-K & 38.8717$\pm$1.1330  & 0.9389$\pm$0.0209 & 0.0504$\pm$0.0057 &  42.9060$\pm$0.8765 &  0.9802$\pm$0.0028  &  0.0236$\pm$0.0028 & 4 & 4.11 M\\
pMRI-$\C$Net-K \eqref{eq:loss_ablation} & 38.9661$\pm$1.4382  & 0.9421$\pm$0.0177 & 0.0498$\pm$0.0056 &  43.2604$\pm$0.7610 & 0.9833$\pm$0.0022 & 0.0226$\pm$0.0022 & 4 & 4.11 M\\
pMRI-$\C$Net-K & 39.3360$\pm$1.1854 & 0.9497$\pm$0.0208 &  0.0477$\pm$0.0048 &  43.5653$\pm$0.8265 & 0.9844$\pm$0.0022 & 0.0217$\pm$0.0010 & 4 & 4.11 M\\
Proposed \eqref{eq:loss} & \textbf{40.7101$\pm$1.5357} & \textbf{0.9619$\pm$0.0144} & \textbf{0.0408$\pm$0.0051} & \textbf{44.8120$\pm$1.3185} & \textbf{0.9886$\pm$0.0023} & \textbf{0.0189$\pm$0.0018} & 4 & 2.92 M\\
\bottomrule
\end{tabular}}
\end{center}
\label{tab:result_our}
\end{table*}
The complete structure of the network with the three types of initialization is shown in Fig.~\ref{fig:framework}.
\begin{figure}[t]
    \centering
    \includegraphics[width=1\linewidth]{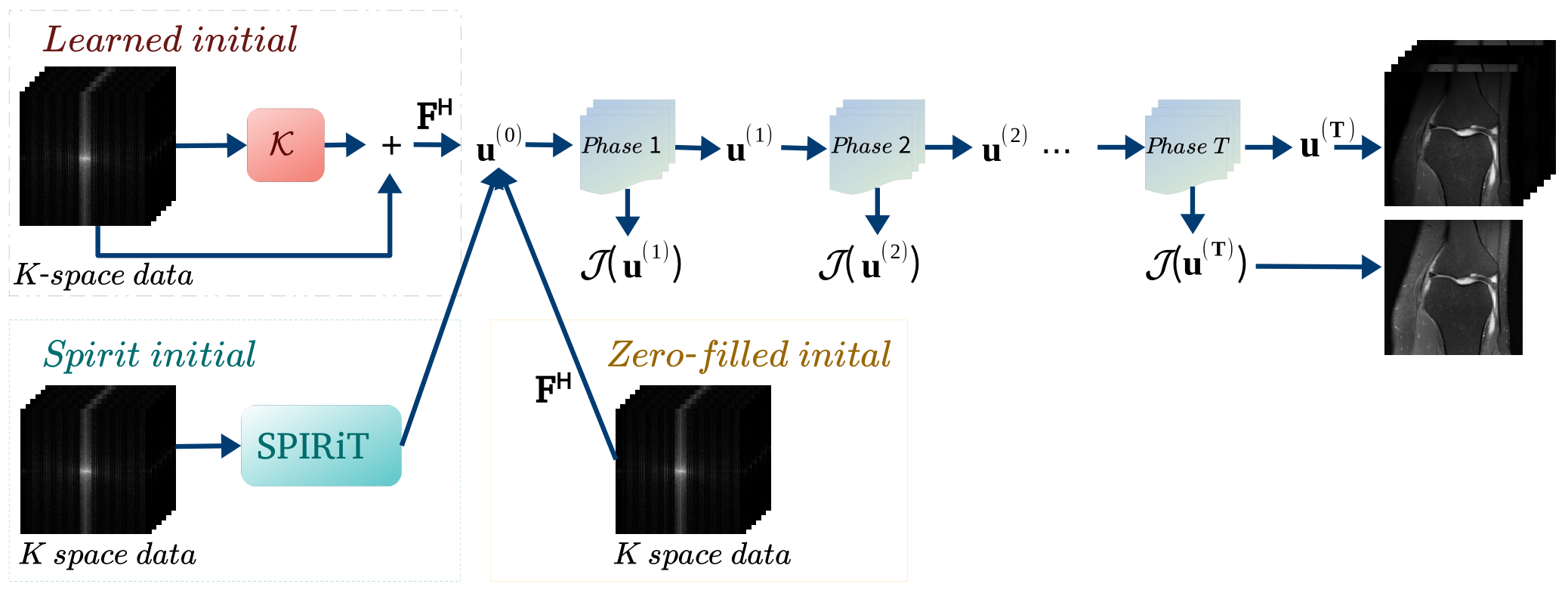}
    \caption{The framework paradigm for all phases with three different initial reconstruction methods, including zero-filled initial, SPIRiT reconstruction initial, and the learned initial. For $ t =1, \cdots, T$, each phase follows the algorithm introduced in our previous work \cite{10.1007/978-3-030-61598-7_2}.}  
    \label{fig:framework}
\end{figure}

We consider two types of training datasets in the ablation experiments and design a proper loss function for each of them in Section \ref{subsec:Ablation}.
In the first case, the ground truth is $\ubf^*$, we set loss function:
\begin{equation}
\begin{aligned}
   \ell(\ubf) = & \sum\nolimits^{c}_{i=1}  \gamma \| \ubf_i - \ubf^*_i\| + \| |\J(\ubf)| - \text{RSS}(\ubf^*)\|\\
   &+ \beta  \|  \text{RSS}(\ubf_i{(0)}) - \text{RSS}(\ubf^*_i)\|, \label{eq:loss_ablation}  
\end{aligned}
\end{equation}
$\ubf$ means $\ubf{(T)}$ for simplicity. If the training dataset consists ground truth single-body image $\vbf^* \in \mathbb{R}^{m \times n} = \text{RSS}(\ubf^*)$, then we test our network performance using the loss function:
\begin{equation}
\label{eq:loss_body_ablation}
\ell(\ubf) = \gamma \| \text{RSS}(\ubf) - \vbf^*\| + \| |\J(\ubf)| - \vbf^*\| + \beta\| \text{RSS}( \ubf{(0)}) - \vbf^*\|.
\end{equation}

The loss functions are indicated in Tables \ref{tab:result_our} and \ref{tab:ui_mse} follow the experiments.  We set $ \gamma = 1, \beta = 10^{-3}$ in \eqref{eq:loss_ablation} and \eqref{eq:loss_body_ablation} in the implementations.

%
\subsubsection{Combination operator $\J$ vs. RSS}\label{subsec:RSS}
In order to justify the effectiveness of the learned nonlinear combination operator $\J$, we modified pMRI-$\C$Net-ZF by substituting $ \J$ with RSS which is widely used to combine multi-coil images into a single-body image.
The other operators $ \G,\tilde{\G}$ and $ \tilde{\J}$ remain to perform complex convolutions.
Specifically, the output of RSS: $(\sum_{i=1}^{c} |\ubf_i{(t)}|^2)^{1/2} \in\mathbb{R}^{mn} $ is a single-channel real-valued image, which is input into both real and imaginary parts of the nonlinear operator $ \G$, so the output split into complex value. We refer  pMRI-Net-ZF/pMRI-$\C$Net-ZF as the networks with combination operator $\J$ and pMRI-$\C$Net-RSS as the network with RSS, all other settings remain the same as before.

The reconstructed images and average evaluation results of these two types of networks are shown in Fig.~\ref{FSPD} and Table \ref{tab:result_our}. pMRI-Net-ZF and pMRI-$\C$Net-ZF outperform pMRI-$\C$Net-RSS with mean improvements of 1.36 dB and 1.63 dB in PSNR respectively. It indicates that the learned combination operator $\J$ gives more favorable performance compared with applying RSS.
\subsubsection{Effect of initialization}\label{subsec:initial}
The choice of the input of the reconstruction network $\ubf{(0)}$, also has impacts on the final reconstruction quality. Instead of directly using the partial k-space data $\fbf$ as the input of our network, we use three different choices of input $\ubf{(0)}$: (i) Zero-filling reconstruction $\Fbf^{H} \fbf$; (ii) SPIRiT \cite{doi:10.1002/mrm.22428} reconstruction; (iii) $\Fbf^{H}(\fbf+\K(\fbf))$, one can treat $\fbf+\K(\fbf)$ as an interpolated pseudo full k-space.

We observe that from Table \ref{tab:result_our},
SPIRiT initial makes a slight improvement over the zero-filled initial, whereas the learned initial achieves the highest reconstruction quality compared to the other two initializations. 
Fig.~\ref{Init} displays the three types of initials. We observe that both SPIRiT and the learned initial obtain higher spatial resolution over zero-filling.
SPIRiT is a classical k-space method, this initial did a better job on reducing the aliasing artifacts and keeping edges compare to zero-filling and the learned initial, but SPIRiT introduces more noise. 

Comparing to zero-filling, the learned initial preserves structure features of major tissue thanks to the k-space network $\K$. Comparing to the SPIRiT initial, the learned initial reduces resolution noises in the image space. Learning-based initial obtain a balanced performance between zero-filling and SPIRiT in the sense of avoiding the weakness of these two initials.

\begin{figure}
\centering
\includegraphics[width=0.24\linewidth]{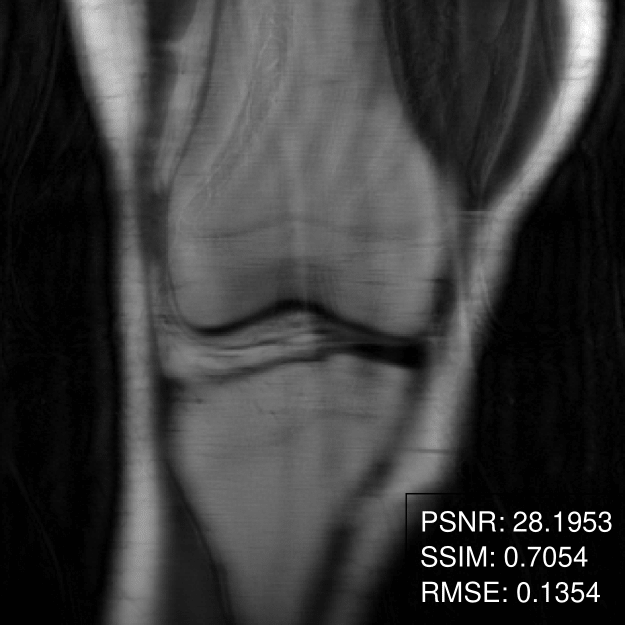}
\includegraphics[width=0.24\linewidth]{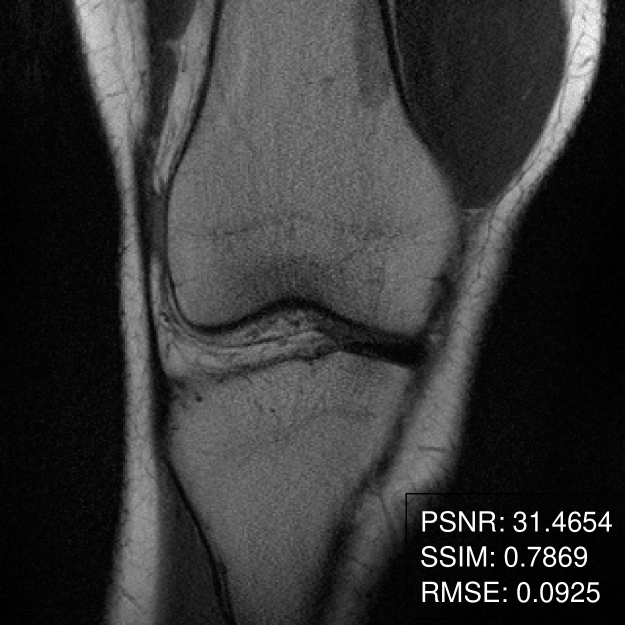}
\includegraphics[width=0.24\linewidth]{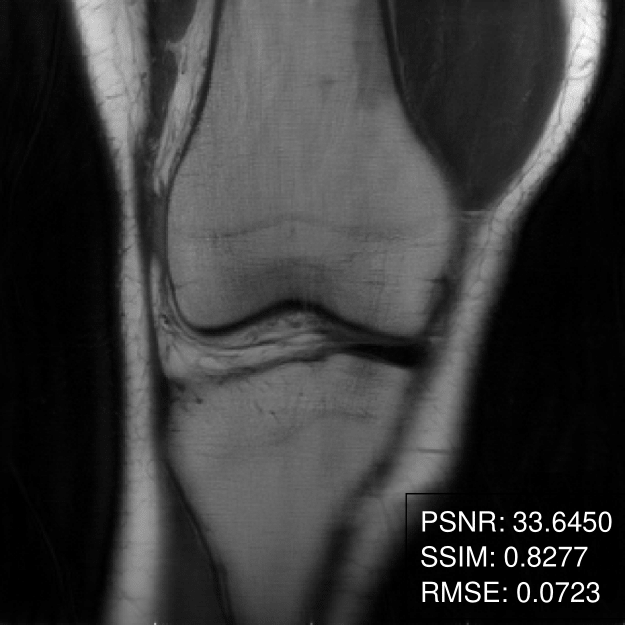}
\includegraphics[width=0.24\linewidth]{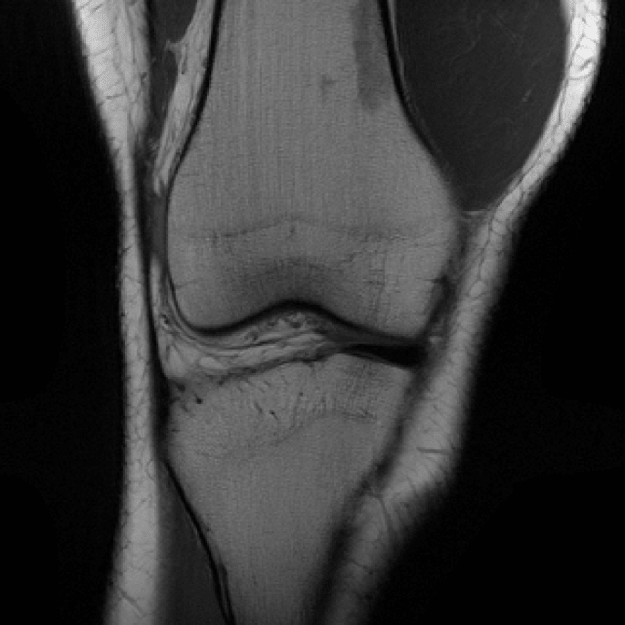}\\
\includegraphics[width=0.24\linewidth]{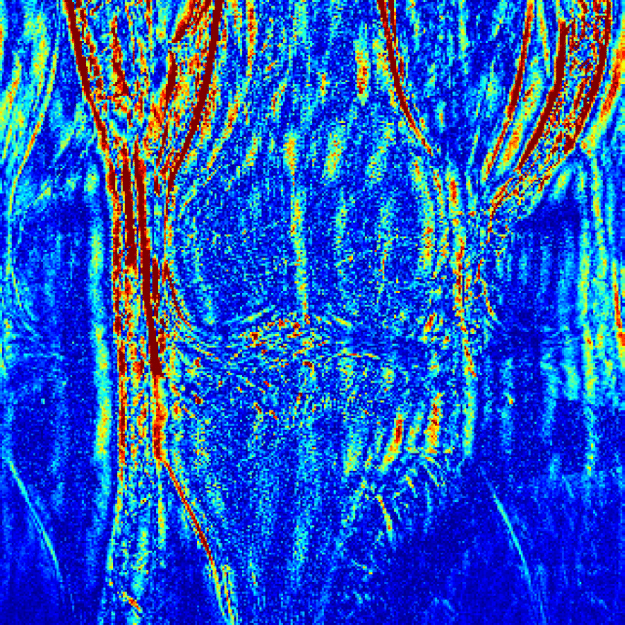}
\includegraphics[width=0.24\linewidth]{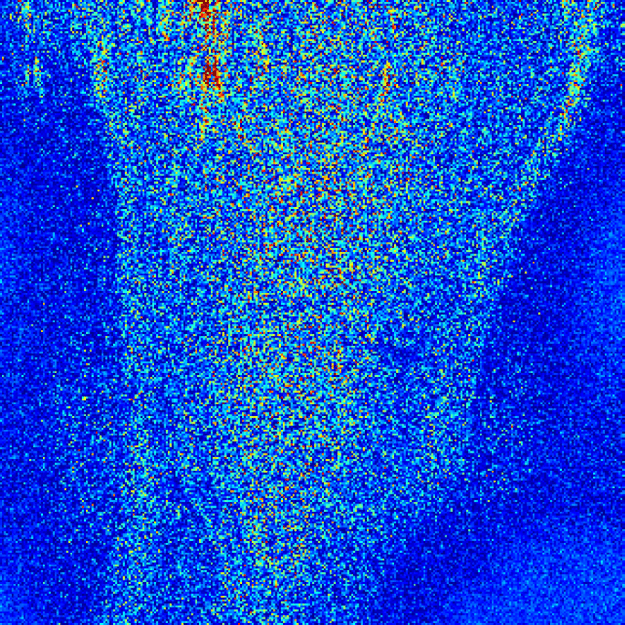}
\includegraphics[width=0.24\linewidth]{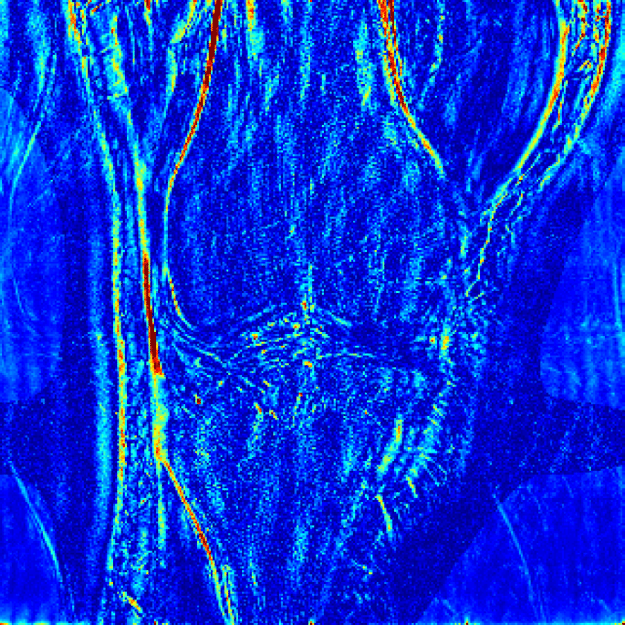}
\includegraphics[width=0.24\linewidth]{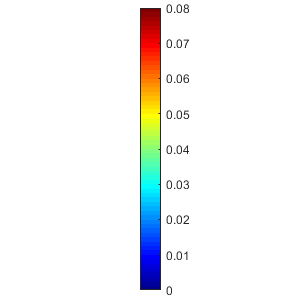}
\caption{The first row (from left to right) shows the RSS of initial $ \Bar{\ubf}{(0)}$ for pMRI-Net-ZF/pMRI-$\C$Net-ZF, pMRI-Net-SP/pMRI-$\C$Net-SP, pMRI-Net-K/pMRI-$\C$Net-K, and reference image on the Coronal PD knee image. The second row shows their pointwise error maps and color bar. }
\label{Init}
\end{figure}

\subsubsection{Complex convolutions}\label{subsec:complexnetwork}

In this experiment, we compared -Net and -$\C$Net.
The quantitative evaluation from Table \ref{tab:result_our} indicates that the proposed complex-valued networks are outstanding in terms of PSNR/SSIM/RMSE over proposed real-valued networks. 
Table \ref{tab:ui_mse} informs the complex convolutions extract the features in each $\ubf_i$ and obtain the lowest RMSE. The phase image for one channel of the reconstructed coil-image is displayed in Fig.~\ref{fig:phase_images}.
These results demonstrate complex-valued networks are playing important roles in updating multi-coil images and preserving phase information of each channel.

\begin{figure*}
\includegraphics[width=0.10\linewidth, angle=180]{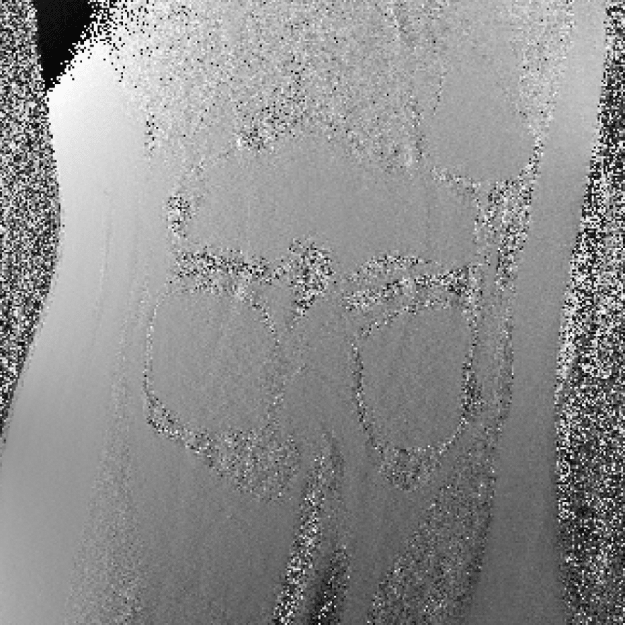}
\includegraphics[width=0.10\linewidth, angle=180]{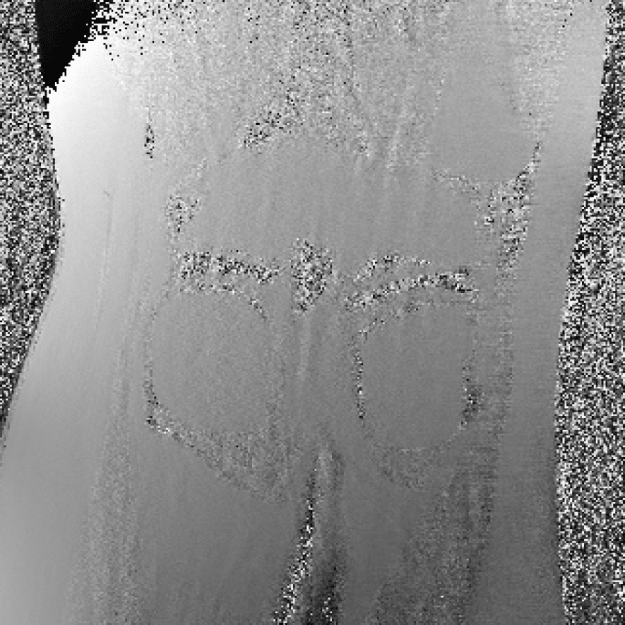}
\includegraphics[width=0.10\linewidth, angle=180]{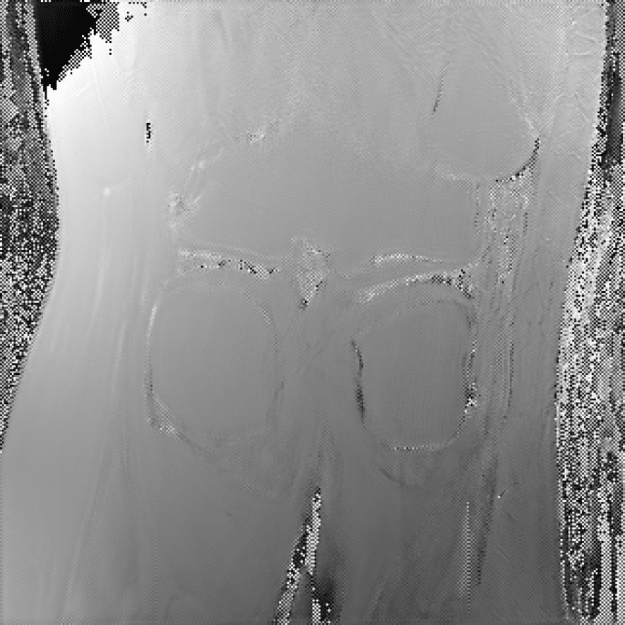}
\includegraphics[width=0.10\linewidth, angle=180]{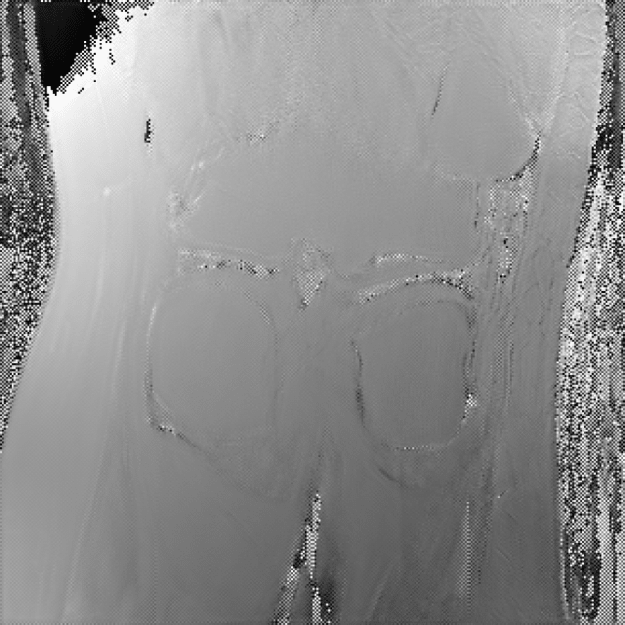}
\includegraphics[width=0.10\linewidth, angle=180]{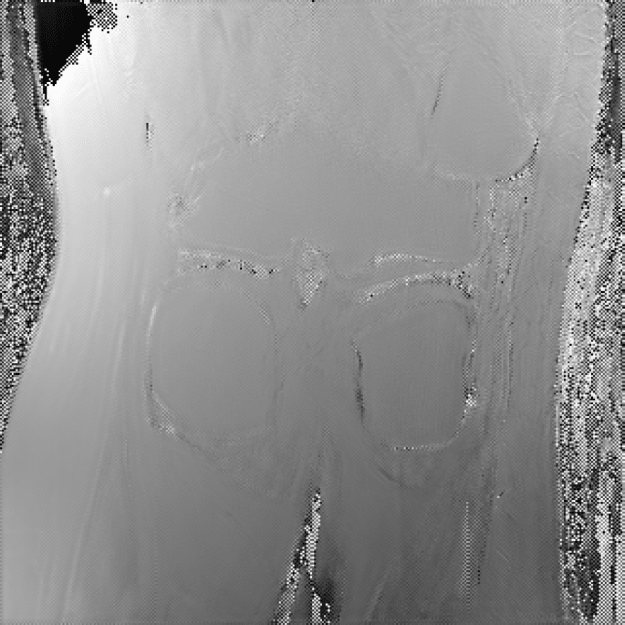}
\includegraphics[width=0.10\linewidth, angle=180]{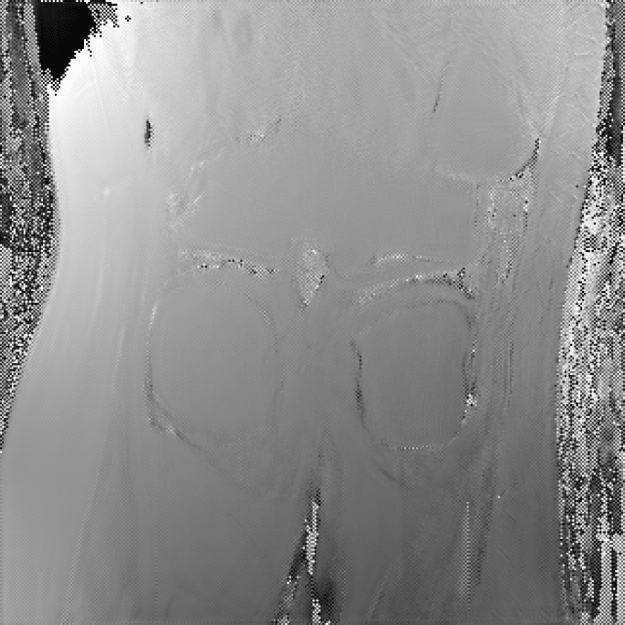}
\includegraphics[width=0.10\linewidth, angle=180]{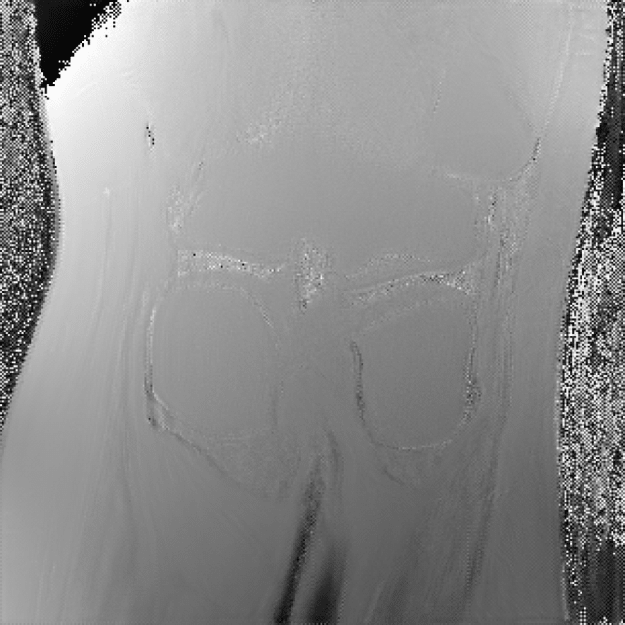}
\includegraphics[width=0.10\linewidth, angle=180]{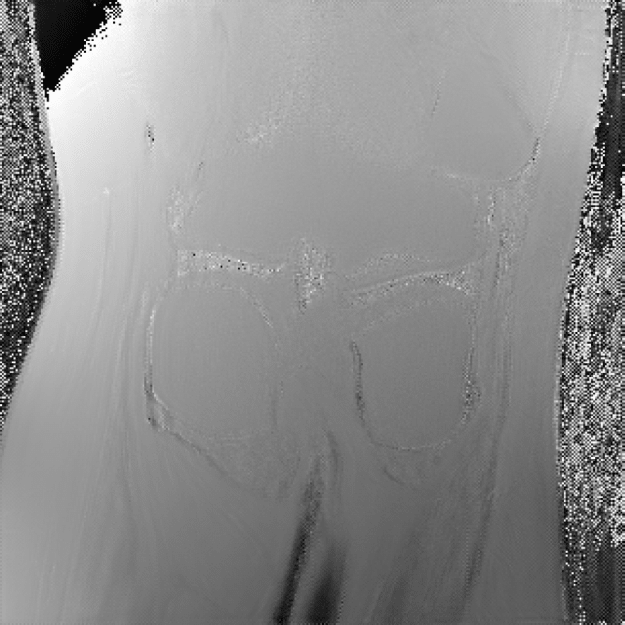}
\includegraphics[width=0.10\linewidth, angle=180]{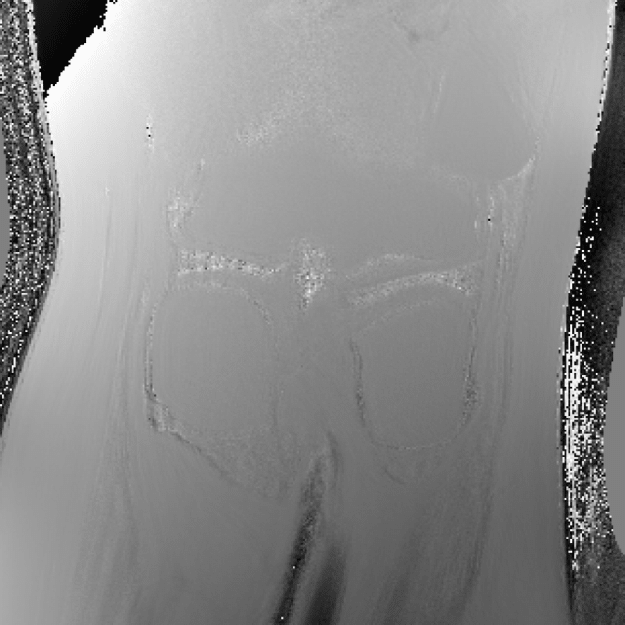}\\
\includegraphics[width=0.10\linewidth, angle=180]{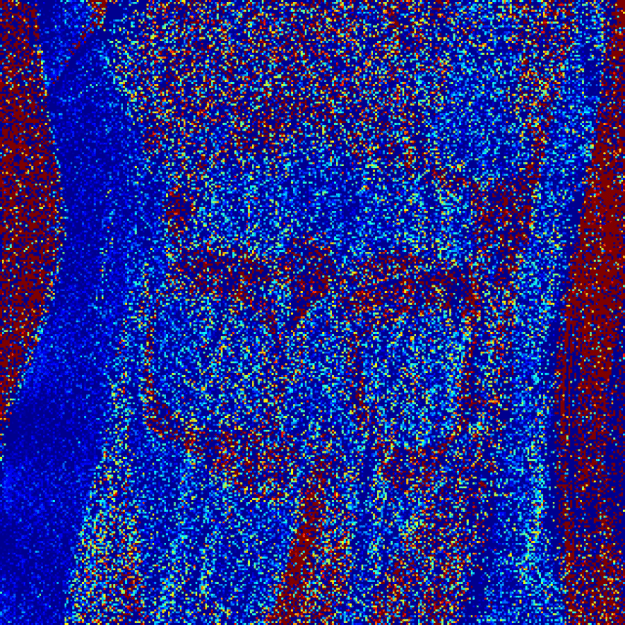}
\includegraphics[width=0.10\linewidth, angle=180]{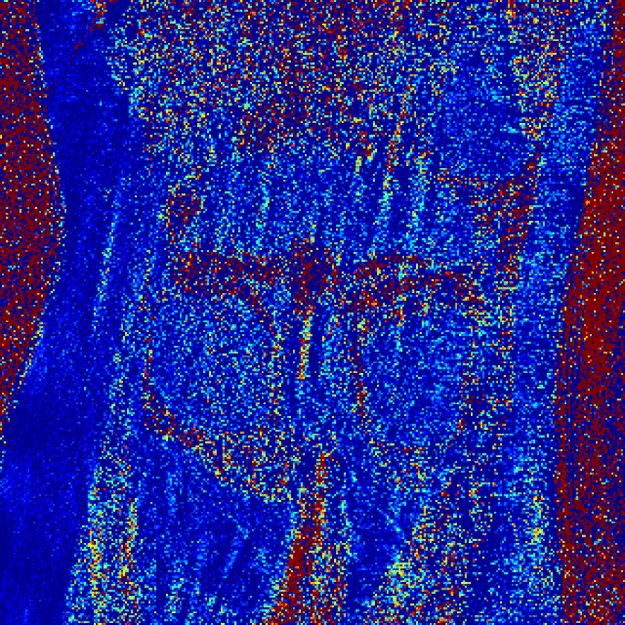}
\includegraphics[width=0.10\linewidth, angle=180]{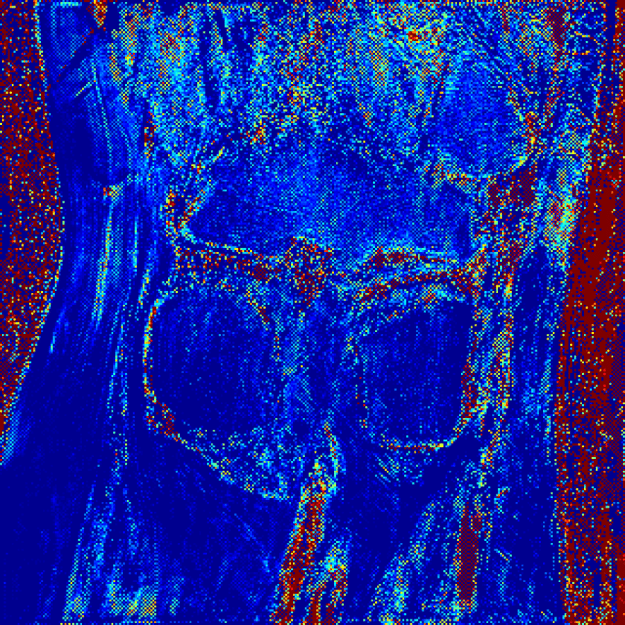}
\includegraphics[width=0.10\linewidth, angle=180]{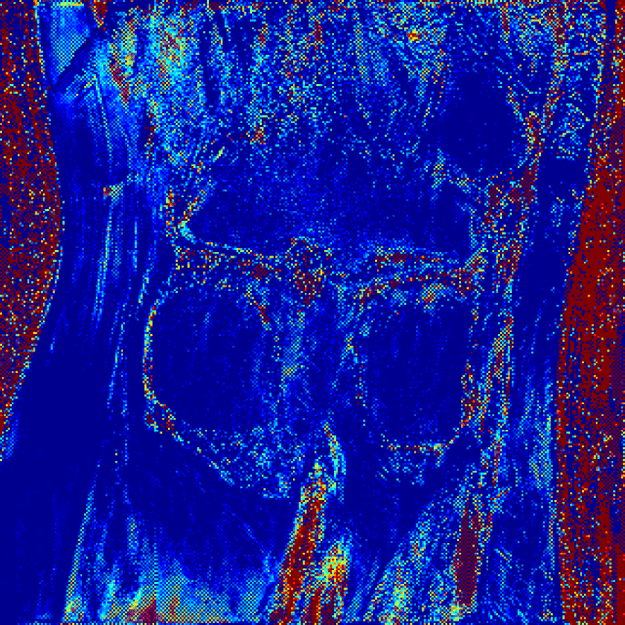}
\includegraphics[width=0.10\linewidth, angle=180]{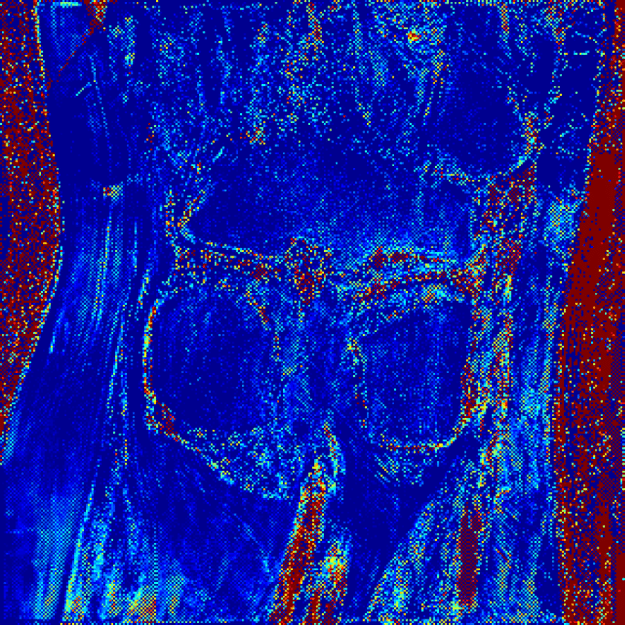}
\includegraphics[width=0.10\linewidth, angle=180]{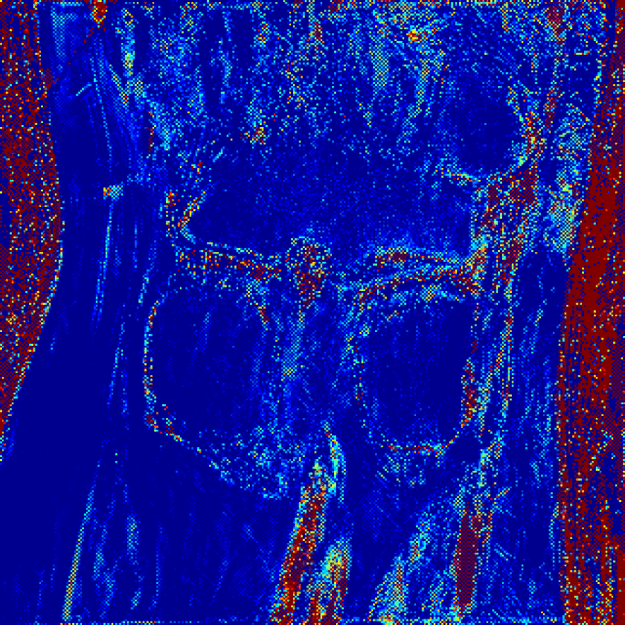}
\includegraphics[width=0.10\linewidth, angle=180]{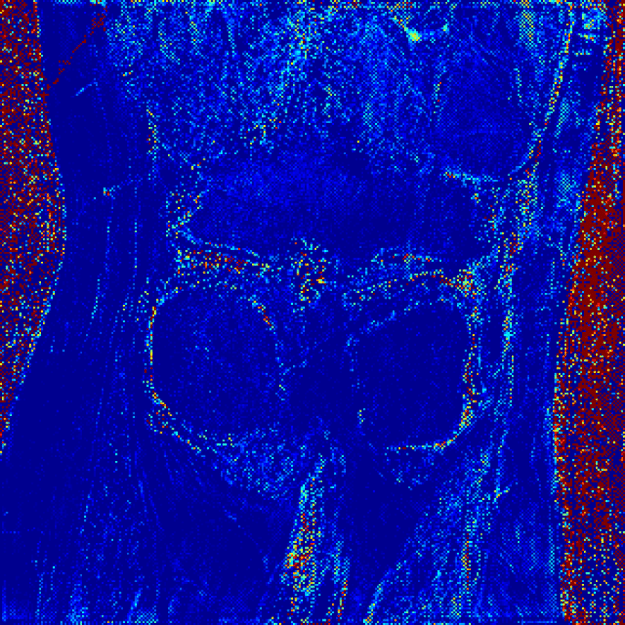}
\includegraphics[width=0.10\linewidth, angle=180]{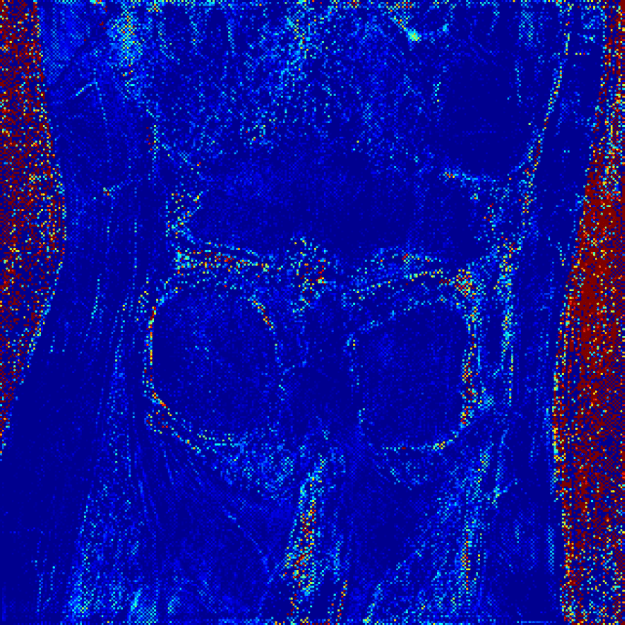}
\includegraphics[width=0.10\linewidth, angle=180]{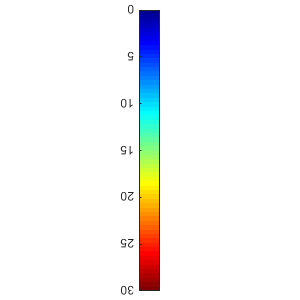}
\caption{The first row shows phase information of one coil in the reconstructed coil-image for GRAPPA, SPIRiT, De-AliasingNet, DeepcomplexMRI,  CDF-Net, Adaptive-CS-Net, pMRI-$\C$Net-K with loss function \eqref{eq:loss_ablation}, the proposed method and referenced image. The second row shows the corresponding pointwise error maps and color bar, the maximum error is $30^{\circ}$. }
\label{phase}
\end{figure*}

\subsubsection{Comparison with proposed image domain reconstruction and domain hybrid reconstruction}
\label{sec:discussion}
The major difference between pMRI-$\C$Net-K and the proposed network is that pMRI-$\C$Net-K  only iterates \eqref{eq:u_ablation} in the image domain, while the proposed method iterates \eqref{eq:scheme} which conducts a domain hybrid reconstruction. The denoising network in \eqref{eq:u_ablation} is $ \tilde{\J} \circ \tilde{\G} \circ \soft_{\alpha_t}(  \G \circ \J)$, and in the proposed network we use $ \M + \Fbf^{H} \K \Fbf + \Fbf^{H} \K \Fbf \M$, where  $ \M = \tilde{\J} \circ \tilde{\G} \circ \G \circ \J $ is the image domain network. The soft-thresholding operator $  \soft_{\alpha_t}$ was eliminated in the proposed method since we found the results of adding the soft-thresholding does not make an obvious difference.

The average reconstruction outcomes in Table \ref{tab:result_our} and Table \ref{tab:ui_mse} suggest that domain hybrid approach achieves better performance. Comparing to the pMRI-$\C$Net-K, our proposed method improved 0.598 dB in PSNR, 0.003 in SSIM, and reduced 0.014 in RMSE, which shows in Table \ref{tab:result_our} for PD dataset. 
\section{Conclusion}
\label{sec:conclusion}
This paper introduces a discrete-time optimal control framework for the calibration-free pMRI reconstruction model. We apply a convolutional combination operator to combine channels of the multi-coil images and apply a parametrized regularization function to the channel-combined image to reconstruct channel-wise multi-coil images. The proposed method is inspired by the proximal gradient algorithm. The proximal point is learned by two denoising networks, which conducts in the image domain and k-space domain.
We cast the reconstruction network as a structured discrete-time optimal control system, resulting in an optimal control formulation of parameter training, which provides an interpretable and high-performance deep architecture for pMRI reconstruction. We design network training from the view of the Method of Lagrangian Multipliers. We showed that the method of Lagrangian multipliers is equivalent to back-propagation, and we can employ SGD based algorithms to obtain a solution satisfying the necessary condition of the optimal control problem.
The reconstruction results are of high perceived quality demonstrate the superior performance of the proposed pMRI-Net.


\bibliographystyle{cas-model2-names}

\bibliography{cas-refs}

\begin{thebibliography}{68}
\expandafter\ifx\csname natexlab\endcsname\relax\def\natexlab#1{#1}\fi
\providecommand{\url}[1]{\texttt{#1}}
\providecommand{\href}[2]{#2}
\providecommand{\path}[1]{#1}
\providecommand{\DOIprefix}{doi:}
\providecommand{\ArXivprefix}{arXiv:}
\providecommand{\URLprefix}{URL: }
\providecommand{\Pubmedprefix}{pmid:}
\providecommand{\doi}[1]{\href{http://dx.doi.org/#1}{\path{#1}}}
\providecommand{\Pubmed}[1]{\href{pmid:#1}{\path{#1}}}
\providecommand{\bibinfo}[2]{#2}
\ifx\xfnm\relax \def\xfnm[#1]{\unskip,\space#1}\fi
\bibitem[{Abadi et~al.(2016)}]{abadi2016tensorflow}
\bibinfo{author}{Abadi, M.}, et~al., \bibinfo{year}{2016}.
\newblock \bibinfo{title}{Tensorflow: A system for large-scale machine
  learning}, in: \bibinfo{booktitle}{12th $\{$USENIX$\}$ symposium on operating
  systems design and implementation ($\{$OSDI$\}$ 16)}, pp.
  \bibinfo{pages}{265--283}.
\bibitem[{Adler and {\"O}ktem(2018)}]{adler2018learned}
\bibinfo{author}{Adler, J.}, \bibinfo{author}{{\"O}ktem, O.},
  \bibinfo{year}{2018}.
\newblock \bibinfo{title}{Learned primal-dual reconstruction}.
\newblock \bibinfo{journal}{IEEE transactions on medical imaging}
  \bibinfo{volume}{37}, \bibinfo{pages}{1322--1332}.
\bibitem[{Aggarwal et~al.(2019)Aggarwal, Mani and Jacob}]{Aggarwal_2019}
\bibinfo{author}{Aggarwal, H.K.}, \bibinfo{author}{Mani, M.P.},
  \bibinfo{author}{Jacob, M.}, \bibinfo{year}{2019}.
\newblock \bibinfo{title}{Modl: Model-based deep learning architecture for
  inverse problems}.
\newblock \bibinfo{journal}{IEEE Transactions on Medical Imaging}
  \bibinfo{volume}{38}, \bibinfo{pages}{394–405}.
\bibitem[{{Bengio} et~al.(1994){Bengio}, {Simard} and {Frasconi}}]{279181}
\bibinfo{author}{{Bengio}, Y.}, \bibinfo{author}{{Simard}, P.},
  \bibinfo{author}{{Frasconi}, P.}, \bibinfo{year}{1994}.
\newblock \bibinfo{title}{Learning long-term dependencies with gradient descent
  is difficult}.
\newblock \bibinfo{journal}{IEEE Transactions on Neural Networks}
  \bibinfo{volume}{5}, \bibinfo{pages}{157--166}.
\newblock \DOIprefix\doi{10.1109/72.279181}.
\bibitem[{Bian et~al.(2020)Bian, Chen and Ye}]{10.1007/978-3-030-61598-7_2}
\bibinfo{author}{Bian, W.}, \bibinfo{author}{Chen, Y.}, \bibinfo{author}{Ye,
  X.}, \bibinfo{year}{2020}.
\newblock \bibinfo{title}{Deep parallel mri reconstruction network without coil
  sensitivities}, in: \bibinfo{editor}{Deeba, F.}, \bibinfo{editor}{Johnson,
  P.}, \bibinfo{editor}{W{\"u}rfl, T.}, \bibinfo{editor}{Ye, J.C.} (Eds.),
  \bibinfo{booktitle}{Machine Learning for Medical Image Reconstruction},
  \bibinfo{publisher}{Springer International Publishing},
  \bibinfo{address}{Cham}. pp. \bibinfo{pages}{17--26}.
\bibitem[{Buonocore and Gao(1997)}]{buonocore1997ghost}
\bibinfo{author}{Buonocore, M.H.}, \bibinfo{author}{Gao, L.},
  \bibinfo{year}{1997}.
\newblock \bibinfo{title}{Ghost artifact reduction for echo planar imaging
  using image phase correction}.
\newblock \bibinfo{journal}{Magnetic resonance in medicine}
  \bibinfo{volume}{38}, \bibinfo{pages}{89--100}.
\bibitem[{Chen et~al.(2020)Chen, Chen and Sun}]{chen2020mri}
\bibinfo{author}{Chen, E.Z.}, \bibinfo{author}{Chen, T.}, \bibinfo{author}{Sun,
  S.}, \bibinfo{year}{2020}.
\newblock \bibinfo{title}{Mri image reconstruction via learning optimization
  using neural odes}, in: \bibinfo{booktitle}{International Conference on
  Medical Image Computing and Computer-Assisted Intervention},
  \bibinfo{organization}{Springer}. pp. \bibinfo{pages}{83--93}.
\bibitem[{Chen et~al.(2018)Chen, Rubanova, Bettencourt and
  Duvenaud}]{chen2018neural}
\bibinfo{author}{Chen, R.T.}, \bibinfo{author}{Rubanova, Y.},
  \bibinfo{author}{Bettencourt, J.}, \bibinfo{author}{Duvenaud, D.},
  \bibinfo{year}{2018}.
\newblock \bibinfo{title}{Neural ordinary differential equations}.
\newblock \bibinfo{journal}{arXiv preprint arXiv:1806.07366} .
\bibitem[{Chen et~al.(2019)}]{10.1007/978-3-030-32248-9_4}
\bibinfo{author}{Chen, Y.}, et~al., \bibinfo{year}{2019}.
\newblock \bibinfo{title}{Model-based convolutional de-aliasing network
  learning for parallel mr imaging}, in: \bibinfo{booktitle}{Medical Image
  Computing and Computer Assisted Intervention -- MICCAI 2019},
  \bibinfo{publisher}{Springer International Publishing},
  \bibinfo{address}{Cham}. pp. \bibinfo{pages}{30--38}.
\bibitem[{Cheng et~al.(2019)}]{cheng2019model}
\bibinfo{author}{Cheng, J.}, et~al., \bibinfo{year}{2019}.
\newblock \bibinfo{title}{Model learning: Primal dual networks for fast mr
  imaging}, in: \bibinfo{booktitle}{International Conference on Medical Image
  Computing and Computer-Assisted Intervention},
  \bibinfo{organization}{Springer}. pp. \bibinfo{pages}{21--29}.
\bibitem[{Chernousko and Lyubushin(1982)}]{MSA}
\bibinfo{author}{Chernousko, F.L.}, \bibinfo{author}{Lyubushin, A.A.},
  \bibinfo{year}{1982}.
\newblock \bibinfo{title}{Method of successive approximations for solution of
  optimal control problems}.
\newblock \bibinfo{journal}{Optimal Control Applications and Methods}
  \bibinfo{volume}{3}, \bibinfo{pages}{101--114}.
\bibitem[{Choi et~al.(2013)Choi, Kim, Oh, Han and Park}]{choi2013iterative}
\bibinfo{author}{Choi, J.}, \bibinfo{author}{Kim, D.}, \bibinfo{author}{Oh,
  C.}, \bibinfo{author}{Han, Y.}, \bibinfo{author}{Park, H.},
  \bibinfo{year}{2013}.
\newblock \bibinfo{title}{An iterative reconstruction method of complex images
  using expectation maximization for radial parallel mri}.
\newblock \bibinfo{journal}{Physics in Medicine \& Biology}
  \bibinfo{volume}{58}, \bibinfo{pages}{2969}.
\bibitem[{Cole et~al.(2020)}]{cole2020analysis}
\bibinfo{author}{Cole, E.K.}, et~al., \bibinfo{year}{2020}.
\newblock \bibinfo{title}{Analysis of deep complex-valued convolutional neural
  networks for mri reconstruction}.
\newblock \bibinfo{journal}{arXiv:2004.01738} \URLprefix
  \url{https://arxiv.org/abs/2004.01738}.
\bibitem[{Constantinides et~al.(1997)Constantinides, Atalar and
  McVeigh}]{RSS_noise}
\bibinfo{author}{Constantinides, C.D.}, \bibinfo{author}{Atalar, E.},
  \bibinfo{author}{McVeigh, E.R.}, \bibinfo{year}{1997}.
\newblock \bibinfo{title}{Signal-to-noise measurements in magnitude images from
  nmr phased arrays}.
\newblock \bibinfo{journal}{Magnetic Resonance in Medicine}
  \bibinfo{volume}{38}, \bibinfo{pages}{852--857}.
\newblock \DOIprefix\doi{10.1002/mrm.1910380524}.
\bibitem[{Deshmane et~al.(2012)}]{doi:10.1002/jmri.23639}
\bibinfo{author}{Deshmane, A.}, et~al., \bibinfo{year}{2012}.
\newblock \bibinfo{title}{Parallel mr imaging}.
\newblock \bibinfo{journal}{Journal of Magnetic Resonance Imaging}
  \bibinfo{volume}{36}, \bibinfo{pages}{55--72}.
\bibitem[{Eo et~al.(2018)}]{doi:10.1002/mrm.27201}
\bibinfo{author}{Eo, T.}, et~al., \bibinfo{year}{2018}.
\newblock \bibinfo{title}{Kiki-net: cross-domain convolutional neural networks
  for reconstructing undersampled magnetic resonance images}.
\newblock \bibinfo{journal}{Magnetic Resonance in Medicine}
  \bibinfo{volume}{80}, \bibinfo{pages}{2188--2201}.
\newblock \DOIprefix\doi{10.1002/mrm.27201}.
\bibitem[{Glorot and Bengio(2010)}]{glorot2010understanding}
\bibinfo{author}{Glorot, X.}, \bibinfo{author}{Bengio, Y.},
  \bibinfo{year}{2010}.
\newblock \bibinfo{title}{Understanding the difficulty of training deep
  feedforward neural networks}, in: \bibinfo{booktitle}{Proceedings of the
  thirteenth international conference on artificial intelligence and
  statistics}, pp. \bibinfo{pages}{249--256}.
\bibitem[{Griswold et~al.(2002)}]{griswold2002generalized}
\bibinfo{author}{Griswold, M.A.}, et~al., \bibinfo{year}{2002}.
\newblock \bibinfo{title}{Generalized autocalibrating partially parallel
  acquisitions (grappa)}.
\newblock \bibinfo{journal}{Magnetic Resonance in Medicine: An Official Journal
  of the International Society for Magnetic Resonance in Medicine}
  \bibinfo{volume}{47}, \bibinfo{pages}{1202--1210}.
\bibitem[{Halkin(1966)}]{PMP}
\bibinfo{author}{Halkin, H.}, \bibinfo{year}{1966}.
\newblock \bibinfo{title}{A maximum principle of the pontryagin type for
  systems described by nonlinear difference equations}.
\newblock \bibinfo{journal}{SIAM Journal on Control} \bibinfo{volume}{4},
  \bibinfo{pages}{90--111}.
\newblock \DOIprefix\doi{10.1137/0304009}.
\bibitem[{Hammernik et~al.(2021)Hammernik, Schlemper, Qin, Duan, Summers and
  Rueckert}]{hammernik2021systematic}
\bibinfo{author}{Hammernik, K.}, \bibinfo{author}{Schlemper, J.},
  \bibinfo{author}{Qin, C.}, \bibinfo{author}{Duan, J.},
  \bibinfo{author}{Summers, R.M.}, \bibinfo{author}{Rueckert, D.},
  \bibinfo{year}{2021}.
\newblock \bibinfo{title}{Systematic evaluation of iterative deep neural
  networks for fast parallel mri reconstruction with sensitivity-weighted coil
  combination}.
\newblock \bibinfo{journal}{Magnetic Resonance in Medicine} .
\bibitem[{Hammernik et~al.(2018)}]{doi:10.1002/mrm.26977}
\bibinfo{author}{Hammernik, K.}, et~al., \bibinfo{year}{2018}.
\newblock \bibinfo{title}{Learning a variational network for reconstruction of
  accelerated mri data}.
\newblock \bibinfo{journal}{Magnetic Resonance in Medicine}
  \bibinfo{volume}{79}, \bibinfo{pages}{3055--3071}.
\bibitem[{{He} et~al.(2016){He}, {Zhang}, {Ren} and {Sun}}]{7780459}
\bibinfo{author}{{He}, K.}, \bibinfo{author}{{Zhang}, X.},
  \bibinfo{author}{{Ren}, S.}, \bibinfo{author}{{Sun}, J.},
  \bibinfo{year}{2016}.
\newblock \bibinfo{title}{Deep residual learning for image recognition}, in:
  \bibinfo{booktitle}{2016 IEEE Conference on Computer Vision and Pattern
  Recognition (CVPR)}, pp. \bibinfo{pages}{770--778}.
\newblock \DOIprefix\doi{10.1109/CVPR.2016.90}.
\bibitem[{Hochreiter(1998)}]{VanishingGradient}
\bibinfo{author}{Hochreiter, S.}, \bibinfo{year}{1998}.
\newblock \bibinfo{title}{The vanishing gradient problem during learning
  recurrent neural nets and problem solutions}.
\newblock \bibinfo{journal}{International Journal of Uncertainty, Fuzziness and
  Knowledge-Based Systems} \bibinfo{volume}{06}, \bibinfo{pages}{107--116}.
\newblock \DOIprefix\doi{10.1142/S0218488598000094}.
\bibitem[{Huang(2020)}]{huang2020medical}
\bibinfo{author}{Huang, F.}, \bibinfo{year}{2020}.
\newblock \bibinfo{title}{Medical imaging using neural networks}.
\newblock \bibinfo{note}{US Patent App. 16/904,981}.
\bibitem[{Huang and Chen(2020)}]{huang2020mri}
\bibinfo{author}{Huang, F.}, \bibinfo{author}{Chen, M.}, \bibinfo{year}{2020}.
\newblock \bibinfo{title}{Magnetic resonance imaging method and device}.
\newblock \bibinfo{note}{US Patent 10,852,376}.
\bibitem[{Huang et~al.(2020)Huang, Han and Mei}]{huang2020magnetic}
\bibinfo{author}{Huang, F.}, \bibinfo{author}{Han, D.}, \bibinfo{author}{Mei,
  L.}, \bibinfo{year}{2020}.
\newblock \bibinfo{title}{Magnetic resonance imaging with deep neutral
  networks}.
\newblock \bibinfo{note}{US Patent App. 16/735,874}.
\bibitem[{Islam et~al.(2021)Islam, Maity and Ray}]{islam2021compressed}
\bibinfo{author}{Islam, S.R.}, \bibinfo{author}{Maity, S.P.},
  \bibinfo{author}{Ray, A.K.}, \bibinfo{year}{2021}.
\newblock \bibinfo{title}{Compressed sensing regularized calibrationless
  parallel magnetic resonance imaging via deep learning}.
\newblock \bibinfo{journal}{Biomedical Signal Processing and Control}
  \bibinfo{volume}{66}, \bibinfo{pages}{102399}.
\bibitem[{Jiang et~al.(2020)}]{10.1007/978-3-030-59713-9_34}
\bibinfo{author}{Jiang, J.}, et~al., \bibinfo{year}{2020}.
\newblock \bibinfo{title}{Unified cross-modality feature disentangler for
  unsupervised multi-domain mri abdomen organs segmentation}, in:
  \bibinfo{booktitle}{Medical Image Computing and Computer Assisted
  Intervention -- MICCAI 2020}, \bibinfo{publisher}{Springer International
  Publishing}, \bibinfo{address}{Cham}. pp. \bibinfo{pages}{347--358}.
\bibitem[{Kingma and Ba(2015)}]{kingma2014adam}
\bibinfo{author}{Kingma, D.P.}, \bibinfo{author}{Ba, J.}, \bibinfo{year}{2015}.
\newblock \bibinfo{title}{Adam: {A} method for stochastic optimization}
  \URLprefix \url{http://arxiv.org/abs/1412.6980}.
\bibitem[{Knoll et~al.(2020)Knoll, Hammernik, Zhang, Moeller, Pock, Sodickson
  and Akcakaya}]{knoll2020deep}
\bibinfo{author}{Knoll, F.}, \bibinfo{author}{Hammernik, K.},
  \bibinfo{author}{Zhang, C.}, \bibinfo{author}{Moeller, S.},
  \bibinfo{author}{Pock, T.}, \bibinfo{author}{Sodickson, D.K.},
  \bibinfo{author}{Akcakaya, M.}, \bibinfo{year}{2020}.
\newblock \bibinfo{title}{Deep-learning methods for parallel magnetic resonance
  imaging reconstruction: A survey of the current approaches, trends, and
  issues}.
\newblock \bibinfo{journal}{IEEE signal processing magazine}
  \bibinfo{volume}{37}, \bibinfo{pages}{128--140}.
\bibitem[{LeCun et~al.(2015)LeCun, Bengio and Hinton}]{lecun2015deep}
\bibinfo{author}{LeCun, Y.}, \bibinfo{author}{Bengio, Y.},
  \bibinfo{author}{Hinton, G.}, \bibinfo{year}{2015}.
\newblock \bibinfo{title}{Deep learning}.
\newblock \bibinfo{journal}{nature} \bibinfo{volume}{521},
  \bibinfo{pages}{436--444}.
\bibitem[{Li()}]{lidynamical}
\bibinfo{author}{Li, Q.}, .
\newblock \bibinfo{title}{Dynamical systems and machine learning} .
\bibitem[{Li et~al.(2017)Li, Chen, Tai et~al.}]{li2017maximum}
\bibinfo{author}{Li, Q.}, \bibinfo{author}{Chen, L.}, \bibinfo{author}{Tai,
  C.}, et~al., \bibinfo{year}{2017}.
\newblock \bibinfo{title}{Maximum principle based algorithms for deep
  learning}.
\newblock \bibinfo{journal}{arXiv preprint arXiv:1710.09513} .
\bibitem[{Li and Hao(2018)}]{pmlr}
\bibinfo{author}{Li, Q.}, \bibinfo{author}{Hao, S.}, \bibinfo{year}{2018}.
\newblock \bibinfo{title}{An optimal control approach to deep learning and
  applications to discrete-weight neural networks}, in: \bibinfo{editor}{Dy,
  J.}, \bibinfo{editor}{Krause, A.} (Eds.), \bibinfo{booktitle}{Proceedings of
  the 35th International Conference on Machine Learning},
  \bibinfo{publisher}{PMLR}. pp. \bibinfo{pages}{2985--2994}.
\bibitem[{Li et~al.(2019)Li, Lin and Shen}]{li2019deep}
\bibinfo{author}{Li, Q.}, \bibinfo{author}{Lin, T.}, \bibinfo{author}{Shen,
  Z.}, \bibinfo{year}{2019}.
\newblock \bibinfo{title}{Deep learning via dynamical systems: An approximation
  perspective}.
\newblock \bibinfo{journal}{arXiv preprint arXiv:1912.10382} .
\bibitem[{Lu et~al.(2020)Lu, Zhang, Huang, Guo, Huang, Xu, Hu, Ou-Yang, Lin,
  Yan et~al.}]{lu2020pfista}
\bibinfo{author}{Lu, T.}, \bibinfo{author}{Zhang, X.}, \bibinfo{author}{Huang,
  Y.}, \bibinfo{author}{Guo, D.}, \bibinfo{author}{Huang, F.},
  \bibinfo{author}{Xu, Q.}, \bibinfo{author}{Hu, Y.}, \bibinfo{author}{Ou-Yang,
  L.}, \bibinfo{author}{Lin, J.}, \bibinfo{author}{Yan, Z.}, et~al.,
  \bibinfo{year}{2020}.
\newblock \bibinfo{title}{pfista-sense-resnet for parallel mri reconstruction}.
\newblock \bibinfo{journal}{Journal of Magnetic Resonance}
  \bibinfo{volume}{318}, \bibinfo{pages}{106790}.
\bibitem[{Lustig and Pauly(2010)}]{doi:10.1002/mrm.22428}
\bibinfo{author}{Lustig, M.}, \bibinfo{author}{Pauly, J.M.},
  \bibinfo{year}{2010}.
\newblock \bibinfo{title}{Spirit: Iterative self-consistent parallel imaging
  reconstruction from arbitrary k-space}.
\newblock \bibinfo{journal}{Magnetic Resonance in Medicine}
  \bibinfo{volume}{64}, \bibinfo{pages}{457--471}.
\newblock \DOIprefix\doi{10.1002/mrm.22428}.
\bibitem[{Lv et~al.(2021)Lv, Wang and Yang}]{lv2021pic}
\bibinfo{author}{Lv, J.}, \bibinfo{author}{Wang, C.}, \bibinfo{author}{Yang,
  G.}, \bibinfo{year}{2021}.
\newblock \bibinfo{title}{Pic-gan: A parallel imaging coupled generative
  adversarial network for accelerated multi-channel mri reconstruction}.
\newblock \bibinfo{journal}{Diagnostics} \bibinfo{volume}{11},
  \bibinfo{pages}{61}.
\bibitem[{{Mardani} et~al.(2019)}]{8417964}
\bibinfo{author}{{Mardani}, M.}, et~al., \bibinfo{year}{2019}.
\newblock \bibinfo{title}{Deep generative adversarial neural networks for
  compressive sensing mri}.
\newblock \bibinfo{journal}{IEEE Transactions on Medical Imaging}
  \bibinfo{volume}{38}, \bibinfo{pages}{167--179}.
\bibitem[{Meng et~al.(2019)}]{10.1007/978-3-030-32251-9_80}
\bibinfo{author}{Meng, N.}, et~al., \bibinfo{year}{2019}.
\newblock \bibinfo{title}{A prior learning network for joint image and
  sensitivity estimation in parallel mr imaging}, in:
  \bibinfo{booktitle}{Medical Image Computing and Computer Assisted
  Intervention -- MICCAI 2019}, \bibinfo{publisher}{Springer International
  Publishing}, \bibinfo{address}{Cham}. pp. \bibinfo{pages}{732--740}.
\bibitem[{Nitski et~al.(2020)}]{10.1007/978-3-030-59713-9_41}
\bibinfo{author}{Nitski, O.}, et~al., \bibinfo{year}{2020}.
\newblock \bibinfo{title}{Cdf-net: Cross-domain fusion network for accelerated
  mri reconstruction}, in: \bibinfo{booktitle}{Medical Image Computing and
  Computer Assisted Intervention -- MICCAI 2020}, \bibinfo{publisher}{Springer
  International Publishing}, \bibinfo{address}{Cham}. pp.
  \bibinfo{pages}{421--430}.
\bibitem[{Parikh and Boyd(2014)}]{parikh2014proximal}
\bibinfo{author}{Parikh, N.}, \bibinfo{author}{Boyd, S.}, \bibinfo{year}{2014}.
\newblock \bibinfo{title}{Proximal algorithms}.
\newblock \bibinfo{journal}{Foundations and Trends in optimization}
  \bibinfo{volume}{1}, \bibinfo{pages}{127--239}.
\bibitem[{Pelc et~al.(1991)Pelc, Herfkens, Shimakawa, Enzmann
  et~al.}]{pelc1991phase}
\bibinfo{author}{Pelc, N.J.}, \bibinfo{author}{Herfkens, R.J.},
  \bibinfo{author}{Shimakawa, A.}, \bibinfo{author}{Enzmann, D.R.}, et~al.,
  \bibinfo{year}{1991}.
\newblock \bibinfo{title}{Phase contrast cine magnetic resonance imaging}.
\newblock \bibinfo{journal}{Magnetic resonance quarterly} \bibinfo{volume}{7},
  \bibinfo{pages}{229--254}.
\bibitem[{Pezzotti et~al.(2020)}]{pezzotti2020adaptive}
\bibinfo{author}{Pezzotti, N.}, et~al., \bibinfo{year}{2020}.
\newblock \bibinfo{title}{An adaptive intelligence algorithm for undersampled
  knee mri reconstruction: Application to the 2019 fastmri challenge}.
\newblock \bibinfo{journal}{arXiv:2004.07339} .
\bibitem[{Pruessmann et~al.(1999)}]{pruessmann1999sense}
\bibinfo{author}{Pruessmann, K.P.}, et~al., \bibinfo{year}{1999}.
\newblock \bibinfo{title}{Sense: sensitivity encoding for fast mri}.
\newblock \bibinfo{journal}{Magnetic Resonance in Medicine: An Official Journal
  of the International Society for Magnetic Resonance in Medicine}
  \bibinfo{volume}{42}, \bibinfo{pages}{952--962}.
\bibitem[{Quan et~al.(2018)Quan, Nguyen-Duc and Jeong}]{Quan2018CompressedSM}
\bibinfo{author}{Quan, T.M.}, \bibinfo{author}{Nguyen-Duc, T.},
  \bibinfo{author}{Jeong, W.K.}, \bibinfo{year}{2018}.
\newblock \bibinfo{title}{Compressed sensing mri reconstruction using a
  generative adversarial network with a cyclic loss}.
\newblock \bibinfo{journal}{IEEE Transactions on Medical Imaging}
  \bibinfo{volume}{37}, \bibinfo{pages}{1488--1497}.
\bibitem[{Sandino et~al.(2020)}]{sandino2020compressed}
\bibinfo{author}{Sandino, C.M.}, et~al., \bibinfo{year}{2020}.
\newblock \bibinfo{title}{Compressed sensing: From research to clinical
  practice with deep neural networks: Shortening scan times for magnetic
  resonance imaging}.
\newblock \bibinfo{journal}{IEEE Signal Processing Magazine}
  \bibinfo{volume}{37}, \bibinfo{pages}{117--127}.
\bibitem[{{Schlemper} et~al.(2018)}]{8067520}
\bibinfo{author}{{Schlemper}, J.}, et~al., \bibinfo{year}{2018}.
\newblock \bibinfo{title}{A deep cascade of convolutional neural networks for
  dynamic mr image reconstruction}.
\newblock \bibinfo{journal}{IEEE Transactions on Medical Imaging}
  \bibinfo{volume}{37}, \bibinfo{pages}{491--503}.
\bibitem[{Sodickson and Manning(1997)}]{sodickson1997simultaneous}
\bibinfo{author}{Sodickson, D.K.}, \bibinfo{author}{Manning, W.J.},
  \bibinfo{year}{1997}.
\newblock \bibinfo{title}{Simultaneous acquisition of spatial harmonics
  (smash): fast imaging with radiofrequency coil arrays}.
\newblock \bibinfo{journal}{Magnetic resonance in medicine}
  \bibinfo{volume}{38}, \bibinfo{pages}{591--603}.
\bibitem[{Souza and Frayne(2019)}]{souza2019hybrid}
\bibinfo{author}{Souza, R.}, \bibinfo{author}{Frayne, R.},
  \bibinfo{year}{2019}.
\newblock \bibinfo{title}{A hybrid frequency-domain/image-domain deep network
  for magnetic resonance image reconstruction}, in: \bibinfo{booktitle}{2019
  32nd SIBGRAPI Conference on Graphics, Patterns and Images (SIBGRAPI)},
  \bibinfo{organization}{IEEE}. pp. \bibinfo{pages}{257--264}.
\bibitem[{Souza et~al.(2019)Souza, Lebel and Frayne}]{pmlr-v102-souza19a}
\bibinfo{author}{Souza, R.}, \bibinfo{author}{Lebel, R.M.},
  \bibinfo{author}{Frayne, R.}, \bibinfo{year}{2019}.
\newblock \bibinfo{title}{A hybrid, dual domain, cascade of convolutional
  neural networks for magnetic resonance image reconstruction},
  \bibinfo{publisher}{PMLR}, \bibinfo{address}{London, United Kingdom}. pp.
  \bibinfo{pages}{437--446}.
\bibitem[{Sriram et~al.(2020)}]{sriram2020grappanet}
\bibinfo{author}{Sriram, A.}, et~al., \bibinfo{year}{2020}.
\newblock \bibinfo{title}{Grappanet: Combining parallel imaging with deep
  learning for multi-coil mri reconstruction}, in:
  \bibinfo{booktitle}{Proceedings of the IEEE/CVF Conference on Computer Vision
  and Pattern Recognition}, pp. \bibinfo{pages}{14315--14322}.
\bibitem[{Tavaf et~al.(2021)Tavaf, Torfi, Ugurbil and Van~de
  Moortele}]{tavaf2021grappa}
\bibinfo{author}{Tavaf, N.}, \bibinfo{author}{Torfi, A.},
  \bibinfo{author}{Ugurbil, K.}, \bibinfo{author}{Van~de Moortele, P.F.},
  \bibinfo{year}{2021}.
\newblock \bibinfo{title}{Grappa-gans for parallel mri reconstruction}.
\newblock \bibinfo{journal}{arXiv preprint arXiv:2101.03135} .
\bibitem[{Walsh et~al.(2000a)Walsh, Gmitro and Marcellin}]{Walsh2000682}
\bibinfo{author}{Walsh, D.}, \bibinfo{author}{Gmitro, A.},
  \bibinfo{author}{Marcellin, M.}, \bibinfo{year}{2000}a.
\newblock \bibinfo{title}{Adaptive reconstruction of phased array mr imagery}.
\newblock \bibinfo{journal}{Magnetic Resonance in Medicine}
  \bibinfo{volume}{43}, \bibinfo{pages}{682--690}.
\bibitem[{Walsh et~al.(2000b)Walsh, Gmitro and Marcellin}]{RSS_bias}
\bibinfo{author}{Walsh, D.O.}, \bibinfo{author}{Gmitro, A.F.},
  \bibinfo{author}{Marcellin, M.W.}, \bibinfo{year}{2000}b.
\newblock \bibinfo{title}{Adaptive reconstruction of phased array mr imagery}.
\newblock \bibinfo{journal}{Magnetic Resonance in Medicine}
  \bibinfo{volume}{43}, \bibinfo{pages}{682--690}.
\newblock
  \DOIprefix\doi{10.1002/(SICI)1522-2594(200005)43:5<682::AID-MRM10>3.0.CO;2-G}.
\bibitem[{Wang et~al.(2009)Wang, Du, O'Halloran, Fain, Kecskemeti, Wieben,
  Johnson and Mistretta}]{wang2009ultrashort}
\bibinfo{author}{Wang, K.}, \bibinfo{author}{Du, J.},
  \bibinfo{author}{O'Halloran, R.}, \bibinfo{author}{Fain, S.},
  \bibinfo{author}{Kecskemeti, S.}, \bibinfo{author}{Wieben, O.},
  \bibinfo{author}{Johnson, K.M.}, \bibinfo{author}{Mistretta, C.},
  \bibinfo{year}{2009}.
\newblock \bibinfo{title}{Ultrashort te spectroscopic imaging (utesi) using
  complex highly-constrained backprojection with local reconstruction (hypr
  lr)}.
\newblock \bibinfo{journal}{Magnetic Resonance in Medicine: An Official Journal
  of the International Society for Magnetic Resonance in Medicine}
  \bibinfo{volume}{62}, \bibinfo{pages}{127--134}.
\bibitem[{Wang et~al.(2017)Wang, Tan, Gao, Liu, Ying, Xiao, Liu, Liu, Zheng and
  Liang}]{wang2017learning}
\bibinfo{author}{Wang, S.}, \bibinfo{author}{Tan, S.}, \bibinfo{author}{Gao,
  Y.}, \bibinfo{author}{Liu, Q.}, \bibinfo{author}{Ying, L.},
  \bibinfo{author}{Xiao, T.}, \bibinfo{author}{Liu, Y.}, \bibinfo{author}{Liu,
  X.}, \bibinfo{author}{Zheng, H.}, \bibinfo{author}{Liang, D.},
  \bibinfo{year}{2017}.
\newblock \bibinfo{title}{Learning joint-sparse codes for calibration-free
  parallel mr imaging}.
\newblock \bibinfo{journal}{IEEE transactions on medical imaging}
  \bibinfo{volume}{37}, \bibinfo{pages}{251--261}.
\bibitem[{Wang et~al.(2020a)}]{WANG2020136}
\bibinfo{author}{Wang, S.}, et~al., \bibinfo{year}{2020}a.
\newblock \bibinfo{title}{Deepcomplexmri: Exploiting deep residual network for
  fast parallel mr imaging with complex convolution}.
\newblock \bibinfo{journal}{Magnetic Resonance Imaging} \bibinfo{volume}{68},
  \bibinfo{pages}{136 -- 147}.
\bibitem[{Wang et~al.(2004)}]{wang2004image}
\bibinfo{author}{Wang, Z.}, et~al., \bibinfo{year}{2004}.
\newblock \bibinfo{title}{Image quality assessment: from error visibility to
  structural similarity}.
\newblock \bibinfo{journal}{IEEE transactions on image processing}
  \bibinfo{volume}{13}, \bibinfo{pages}{600--612}.
\bibitem[{Wang et~al.(2020b)}]{wang2020ikwi}
\bibinfo{author}{Wang, Z.}, et~al., \bibinfo{year}{2020}b.
\newblock \bibinfo{title}{Ikwi-net: A cross-domain convolutional neural network
  for undersampled magnetic resonance image reconstruction}.
\newblock \bibinfo{journal}{Magnetic Resonance Imaging} \bibinfo{volume}{73},
  \bibinfo{pages}{1--10}.
\bibitem[{Weinan(2017)}]{weinan2017proposal}
\bibinfo{author}{Weinan, E.}, \bibinfo{year}{2017}.
\newblock \bibinfo{title}{A proposal on machine learning via dynamical
  systems}.
\newblock \bibinfo{journal}{Communications in Mathematics and Statistics}
  \bibinfo{volume}{5}, \bibinfo{pages}{1--11}.
\bibitem[{Yang et~al.(2016)}]{NIPS2016_6406}
\bibinfo{author}{Yang, Y.}, et~al., \bibinfo{year}{2016}.
\newblock \bibinfo{title}{Deep admm-net for compressive sensing mri}, in:
  \bibinfo{editor}{Lee, D.D.}, \bibinfo{editor}{Sugiyama, M.},
  \bibinfo{editor}{Luxburg, U.V.}, \bibinfo{editor}{Guyon, I.},
  \bibinfo{editor}{Garnett, R.} (Eds.), \bibinfo{booktitle}{Advances in Neural
  Information Processing Systems 29}. \bibinfo{publisher}{Curran Associates,
  Inc.}, pp. \bibinfo{pages}{10--18}.
\bibitem[{{Yang} et~al.(2020)}]{8550778}
\bibinfo{author}{{Yang}, Y.}, et~al., \bibinfo{year}{2020}.
\newblock \bibinfo{title}{Admm-csnet: A deep learning approach for image
  compressive sensing}.
\newblock \bibinfo{journal}{IEEE Transactions on Pattern Analysis and Machine
  Intelligence} \bibinfo{volume}{42}, \bibinfo{pages}{521--538}.
\bibitem[{Zbontar et~al.(2018)}]{zbontar2018fastmri}
\bibinfo{author}{Zbontar, J.}, et~al., \bibinfo{year}{2018}.
\newblock \bibinfo{title}{{fastMRI}: An open dataset and benchmarks for
  accelerated {MRI}}.
\newblock \bibinfo{journal}{arXiv:1811.08839} \URLprefix
  \url{https://arxiv.org/abs/1811.08839}.
\bibitem[{Zhang and Ghanem(2018)}]{zhang2018ista}
\bibinfo{author}{Zhang, J.}, \bibinfo{author}{Ghanem, B.},
  \bibinfo{year}{2018}.
\newblock \bibinfo{title}{Ista-net: Interpretable optimization-inspired deep
  network for image compressive sensing}, in: \bibinfo{booktitle}{Proceedings
  of the IEEE Conference on Computer Vision and Pattern Recognition}, pp.
  \bibinfo{pages}{1828--1837}.
\bibitem[{Zhou et~al.(2019)}]{zhou_pmri}
\bibinfo{author}{Zhou, Z.}, et~al., \bibinfo{year}{2019}.
\newblock \bibinfo{title}{Parallel imaging and convolutional neural network
  combined fast mr image reconstruction: Applications in low-latency
  accelerated real-time imaging}.
\newblock \bibinfo{journal}{Medical Physics} \bibinfo{volume}{46},
  \bibinfo{pages}{3399--3413}.
\bibitem[{Zhu et~al.(2018)}]{zhu2018}
\bibinfo{author}{Zhu, B.}, et~al., \bibinfo{year}{2018}.
\newblock \bibinfo{title}{Image reconstruction by domain-transform manifold
  learning}.
\newblock \bibinfo{journal}{Nature} \bibinfo{volume}{555},
  \bibinfo{pages}{487--492}.
\newblock \DOIprefix\doi{10.1038/nature25988}.
\bibitem[{Zhu et~al.(2020)}]{10.1007/978-3-030-59713-9_37}
\bibinfo{author}{Zhu, Y.}, et~al., \bibinfo{year}{2020}.
\newblock \bibinfo{title}{Cross-domain medical image translation by shared
  latent gaussian mixture model}, in: \bibinfo{booktitle}{Medical Image
  Computing and Computer Assisted Intervention -- MICCAI 2020},
  \bibinfo{publisher}{Springer International Publishing},
  \bibinfo{address}{Cham}. pp. \bibinfo{pages}{379--389}.

\end{thebibliography}


\begin{thebibliography}{0}
\expandafter\ifx\csname natexlab\endcsname\relax\def\natexlab#1{#1}\fi
\providecommand{\url}[1]{\texttt{#1}}
\providecommand{\href}[2]{#2}
\providecommand{\path}[1]{#1}
\providecommand{\DOIprefix}{doi:}
\providecommand{\ArXivprefix}{arXiv:}
\providecommand{\URLprefix}{URL: }
\providecommand{\Pubmedprefix}{pmid:}
\providecommand{\doi}[1]{\href{http://dx.doi.org/#1}{\path{#1}}}
\providecommand{\Pubmed}[1]{\href{pmid:#1}{\path{#1}}}
\providecommand{\bibinfo}[2]{#2}
\ifx\xfnm\relax \def\xfnm[#1]{\unskip,\space#1}\fi

\end{thebibliography}

\end{document}